\documentclass[11pt]{article} \onecolumn 
\usepackage{etex}
\usepackage{amsthm,array,  multirow, bm, amsmath, mathtools, amssymb,booktabs, subcaption}  
\usepackage[bookmarks=false, colorlinks=true, allcolors=blue]{hyperref}  
\usepackage{appendix} 
\usepackage[left=2cm,  right=2cm, top=2cm, bottom=2cm]{geometry}
\usepackage{algorithm}

\usepackage{algorithm, algorithmic}

\usepackage{chngcntr}

\renewcommand{\L}{\bm{L}}

\newcommand{\norm}[1]{\left\|#1\right\|}

\newtheorem{theorem}{Theorem}
\newtheorem{lem}[theorem]{Lemma}
\newtheorem{claim}[theorem]{Claim}

\newtheorem{corollary}[theorem]{Corollary}
\newtheorem{definition}[theorem]{Definition}
\newtheorem{remark}[theorem]{Remark}

\newcommand{\bi}{\begin{itemize}}
\newcommand{\ei}{\end{itemize}}
\newcommand{\ben}{\begin{enumerate}}
\newcommand{\een}{\end{enumerate}}

\newcommand{\bean}{\begin{eqnarray*} }
\newcommand{\eean}{\end{eqnarray*} }
\newcommand{\bea}{\begin{eqnarray} }
\newcommand{\eea}{\end{eqnarray} }
\newcommand{\ba}{\begin{align*} }
\newcommand{\ea}{\end{align*} }

\newcommand{\nn}{\nonumber}

\newcommand{\dif}{{\text{dif}}}
\newcommand{\xhat}{\bm{\hat{x}}}

\newcommand{\bl}{\begin{frame}}
\newcommand{\el} {\end{frame}}

\newcommand{\cred}{\color{red}} 

\newcommand{\svdeq}{\overset{\mathrm{SVD}}=}

\renewcommand\thetheorem{\arabic{section}.\arabic{theorem}}

\newcommand{\tmax}{d} 

\newcommand{\wt}{\bm{w}_t}

\newcommand{\xt}{\bm{x}_t}
\newcommand{\x}{\bm{x}}
\newcommand{\xhatt}{\hat{\bm{x}}_t}
\renewcommand{\l}{\bm{\ell}}

\newcommand{\lt}{\bm{\ell}_t}
\newcommand{\lhat}{\hat{\bm{\ell}}}

\newcommand{\lhatt}{\hat{\l}_t}   
\newcommand{\yt}{\bm{y}_t}
\newcommand{\y}{\bm{y}}
\newcommand{\tty}{\tilde{\bm{y}}}
\newcommand{\w}{\bm{w}}
\renewcommand{\v}{{\bm{\nu}}}
\newcommand{\vt}{\v_t}

\renewcommand{\a}{\bm{a}}
\newcommand{\e}{\bm{e}} 

\newcommand{\et}{\bm{e}_t}

\newcommand{\new}{\mathrm{new}}

\newcommand{\at}{\bm{a}_t}

\newcommand{\I}{\bm{I}}

\newcommand{\Lam}{\bm{\Lambda}}

\newcommand{\T}{\mathcal{T}}
\newcommand{\J}{\mathcal{J}}

\newcommand{\D}{\bm{D}}
\newcommand{\A}{\bm{A}}

\newcommand{\old}{\mathrm{old}}

\newcommand{\Lhat}{\hat{\bm{L}}}

\renewcommand{\P}{\bm{P}}

\newcommand{\V}{\bm{V}}
\newcommand{\R}{\bm{R}}
\newcommand{\ch}{{\mathrm{ch}}}
\newcommand{\fx}{{\mathrm{fix}}}

\renewcommand{\b}{\bm{b}}
\newcommand{\Phat}{\hat{\bm{P}}}

\newcommand{\Span}{\operatorname{span}} 

\newcommand{\basis}{\operatorname{basis}}

\newcommand{\E}{\mathbb{E}}

\newcommand{\train}{\mathrm{train}}
\newcommand{\That}{\hat{\mathcal{T}}}

\newcommand{\SE}{\sin\theta_{\max}}

\newcommand{\that}{{\hat{t}}}

\newcommand{\M}{\bm{M}}

\newcommand{\X}{\bm{X}}
\newcommand{\Y}{\bm{Y}}
\newcommand{\W}{\bm{W}}
\newcommand{\Z}{\bm{Z}}

\newcommand{\outfracrow}{\small{\text{max-outlier-frac-row}}}
\newcommand{\outfraccol}{\small{\text{max-outlier-frac-col}}}

\newcommand\diagmat[2]{\begin{bmatrix} #1 & \bm{0} \\ \bm{0} & #2\end{bmatrix}} 
\newcommand{\one}{\bm{1}}

\newcommand{\ed}{\mathrm{end}}

\newcommand{\SVD}{{SVD}}
\newcommand{\B}{\bm{B}}
\newcommand{\alphadel}{\alpha}
\renewcommand{\Re}{\mathbb{R}}

\newcolumntype{C}[1]{>{\centering\let\newline\\\arraybackslash\hspace{0pt}}m{#1}}

\newcommand{\init}{\mathrm{init}}

\newcommand{\offline}{\mathrm{offline}}

\newcommand{\bz}{b} 
\newcommand{\rmat}{r_{\mat}}
\newcommand{\semax}{\Delta_{\dist}}
\newcommand{\Aa}{\bm{Q}_1}
\newcommand{\Ba}{\bm{Q}_2}
\newcommand{\Ca}{\bm{Q}_3}

\newcommand{\subup}{\text{SubUpd}}

\newcommand{\st}{\bm{x}_t}
\newcommand{\Mt}{\bm{M}_t}
\newcommand{\mot}{\bm{M}_{1,t}}
\newcommand{\mtt}{\bm{M}_{2,t}}
\newcommand{\ep}{\mathbb{E}}
\newcommand{\atf}{\bm{a}_{t,\fx}}
\newcommand{\atr}{\bm{a}_{t,\ch}}

\newcommand{\zt}{\bm{Z}_t}

\newcommand{\rfix}{r}

\newcommand{\lfp}{\lambda^{+}}

\newcommand{\zz} {\varepsilon} 

\newcommand{\Tt}{\mathcal{T}_t}

\newcommand{\Thatt}{\hat{\mathcal{T}}_t}

\newcommand{\shatcs}{\xhat_{t,cs}}

\newcommand{\ezero}{\mathcal{E}_0}

\newcommand{\smin}{x_{\min}}

\newcommand{\pt}{\bm{P}}
\newcommand{\shatt}{\hat{\bm{x}}_t}

\newcommand{\itt}{\bm{I}_{\Tt}}
\newcommand{\bphi}{\bm{\Phi}}
\newcommand{\bpsi}{\bm{\Psi}}

\newcommand{\lthres}{\omega_{evals}} 

\newcommand{\tildej}{j}

\usepackage{tikz}
\usepackage{pgfplots}
\usepgfplotslibrary{groupplots}
\newcommand{\tikzmark}[1]{\tikz[overlay,remember picture] \node (#1) {};}

\newcommand*{\AddNote}[4]{%
    \begin{tikzpicture}[overlay, remember picture]
        \draw [decoration={brace,amplitude=0.5em},decorate,ultra thick,red]
            ($(#3)!(#1.north)!($(#3)-(0,1)$)$) --  
            ($(#3)!(#2.south)!($(#3)-(0,1)$)$)
                node [align=center, text width=2.5cm, pos=0.5, anchor=west] {#4};
    \end{tikzpicture}
}%

\newcommand*{\AddNoteOne}[4]{%
    \begin{tikzpicture}[overlay, remember picture]
        \draw [red]
            ($(#3)!(#1.north)!($(#3)-(0,1)$)$) --  
            ($(#3)!(#2.south)!($(#3)-(0,1)$)$)
                node [align=center, text width=2.5cm, pos=0.5, anchor=west] {#4};
    \end{tikzpicture}
}

\pgfplotsset{
        my stylecompare/.style={
			width=.4\textwidth,
            height=4.5cm,
            label style={font=\Large},
            title style={font=\Large},
            x tick label style={font =\small, /pgf/number format/1000 sep=},
	    axis lines=left,
	    major x tick style = transparent,
	    major y tick style = transparent,
	    every y tick label/.style={
 		   xshift=-.7cm, yshift=-2pt,anchor=south west,inner sep=0pt,font=\small, 
 		   scaled y ticks=false,
	    },
	    every x tick label/.style={font = \small},
            scaled x ticks=false,
        },
        my legend style compare/.style={
            legend entries={
            		GRASTA,
            		ORPCA,
            		s-ReProCS,	
            		NORST,	
            		Offline-NORST,
            		Alt Proj,
            		RPCA-GD,
            },
            legend style={
                at={(0.1,1.5)},
                anchor=north west,
            },
            legend columns=2,
	    legend style={font=\small},
        },
        cycle multi list={
        {blue, line width=0.6pt, mark=o,mark size=3pt}, 
        {black, line width=0.6pt, mark=square,mark size=2.5pt}, 
        {teal, line width=0.6pt, mark=oplus,mark size=2.5pt},
        {red, line width=0.6pt, mark=triangle,mark size=2.5pt},
        {red, solid, line width=0.5pt, mark=oplus,mark size=2.8pt}, 
        {olive, line width=0.6pt, mark=10-pointed star,mark size=2.5pt},  
        {cyan, line width=0.6pt, mark=Mercedes star,mark size=3pt}, 
        {teal, line width=0.6pt, mark=oplus,mark size=2.5pt}, 
		},
}

\pgfplotstableread[col sep = comma]{figures/final_files/err_SE_medrop_icml_mo.dat}\ltmodata
\pgfplotstableread[col sep = comma]{figures/final_files/err_SE_medrop_icml_bern.dat}\ltmodatabern

\usetikzlibrary{decorations.markings}
\usetikzlibrary{decorations.pathreplacing}
\def\MarkLt{4pt}
\def\MarkSep{2pt}

\tikzset{
  TwoMarks/.style={
    postaction={decorate,
      decoration={
        markings,
        mark=at position #1 with
          {
              \begin{scope}[xslant=0.2]
              \draw[line width=\MarkSep,white,-] (0pt,-\MarkLt) -- (0pt,\MarkLt) ;
              \draw[-] (-0.5*\MarkSep,-\MarkLt) -- (-0.5*\MarkSep,\MarkLt) ;
              \draw[-] (0.5*\MarkSep,-\MarkLt) -- (0.5*\MarkSep,\MarkLt) ;
              \end{scope}
          }
       }
    }
  },
  TwoMarks/.default={0.5},
}

\usetikzlibrary{positioning}

\usetikzlibrary{calc}

\tikzstyle{block}  = [rectangle, draw, rounded corners, text width=5cm, text centered, minimum height=1em]
\tikzstyle{smallblock}  = [rectangle, draw, rounded corners,text width=1.8cm, text centered, minimum height=1em]
\tikzstyle{input}  = [rectangle, draw, text width=1.2cm, text centered, minimum height=1em]
\tikzstyle{output}  = [rectangle, draw, text width=1.2cm, text centered, minimum height=1em]
\tikzstyle{block1}  = [rectangle, draw, rounded corners,text width=5cm, text centered, minimum height=1em]
\tikzstyle{blockl1}  = [rectangle, draw, rounded corners,text width=4cm, text centered, minimum height=1em]

\renewcommand{\zt}{\bm{z}_t}
\renewcommand{\semax}{\Delta}
\newcommand{\chd}{{\mathrm{chd}}}
\newcommand{\xmint}{x_{\min}}

\renewcommand{\dif}{\Delta}
\renewcommand{\rmat}{r_{\scriptscriptstyle{L}}}
\renewcommand{\S}{\bm{X}}
\renewcommand{\SE}{\mathrm{SE}}
\renewcommand{\atr}{\bm{a}_{t, \chd}}

\begin{document}

\title{Nearly Optimal Robust Subspace Tracking}
\author{Praneeth Narayanamurthy and Namrata Vaswani \\
\texttt{\{pkurpadn, namrata\}@iastate.edu} \\
Department of Electrical and Computer Engineering, \\
Iowa State University, Ames, IA \thanks{A shorter version of this manuscript \cite{rrpcp_icmltemp} will be presented at ICML, 2018. Another small part, Corollary \ref{thm:rpca}, will appear in \cite{rrpcp_merop}.}}

\maketitle

\renewcommand\thetheorem{\arabic{section}.\arabic{theorem}}
%
%
%
%
%
%
%
%
%
%
%
%
%
%
%
%

\begin{abstract}
In this work, we study the robust subspace tracking (RST) problem and obtain one of the first two provable guarantees for it. The goal of RST is to track sequentially arriving data vectors that lie in a slowly changing low-dimensional subspace, while being robust to corruption by additive sparse outliers. It can also be interpreted as a dynamic (time-varying) extension of robust PCA (RPCA), with the minor difference that RST also requires a short tracking delay. We develop a recursive projected compressive sensing  algorithm that we call Nearly Optimal RST via ReProCS (ReProCS-NORST) because its tracking delay is nearly optimal.
We prove that NORST solves both the RST and the dynamic RPCA problems under weakened standard RPCA assumptions, two simple extra assumptions (slow subspace change and  most outlier magnitudes lower bounded), and a few minor assumptions.

Our guarantee shows that NORST enjoys a near optimal tracking delay of $O(r \log n \log(1/\epsilon))$. Its required delay between subspace change times is the same, and its memory complexity is $n$ times this value. Thus both these are also nearly optimal. Here $n$ is the ambient space dimension, $r$ is the subspaces' dimension, and $\epsilon$ is the tracking accuracy. NORST also has the best outlier tolerance compared with all previous RPCA or RST methods, both theoretically and empirically (including for real videos), without requiring any model on how the outlier support is generated. This is possible because of the extra assumptions it uses.
\end{abstract}


\renewcommand{\subsubsection}[1]{{\bf #1. }} 

\section{Introduction}
Principal Components Analysis (PCA) is one of the most widely used dimension reduction techniques. It finds a small number of orthogonal basis vectors, called principal components, along which most of the variability of the dataset lies.
According to its modern definition \cite{rpca}, {robust PCA (RPCA)} is the problem of decomposing a given data matrix into the sum of a low-rank matrix (true data) and a sparse matrix (outliers). The column space of the low-rank matrix then gives the desired principal subspace (PCA solution).
A common application of RPCA is in video analytics in separating a video into a slow-changing background image sequence (modeled as a low-rank matrix) and a foreground image sequence consisting of moving objects or people (sparse) \cite{rpca}. 
{\em Robust Subspace Tracking (RST)} can be simply interpreted as a time-varying extension of RPCA. It assumes that the true data lies in a low-dimensional subspace that can change with time, albeit slowly. The goal is to track this changing subspace over time in the presence of additive sparse outliers. The offline version of this problem can be called {\em dynamic (or time-varying) RPCA}. RST requires the tracking delay to be small, while dynamic RPCA does not.
Time-varying subspace is a more appropriate model for long data sequences, e.g., long surveillance videos, since if a single subspace model is used the resulting matrix may not be sufficiently low-rank. Moreover the RST problem setting (short tracking delay) is most relevant for applications where real-time or near real-time estimates are needed, e.g., video-based surveillance (object tracking) \cite{cv_app}, monitoring seismological activity \cite{sigproc_app}, or detection of anomalous behavior in dynamic social networks \cite{selin_reprocs}.

In recent years, RPCA has since been extensively studied. Many fast and provably correct approaches now exist: PCP introduced in \cite{rpca} and studied in \cite{rpca,rpca2,rpca_zhang}, AltProj \cite{robpca_nonconvex}, RPCA-GD \cite{rpca_gd} and NO-RMC \cite{rmc_gd}. There is much lesser work on provable dynamic RPCA and RST: original-ReProCS \cite{rrpcp_perf,rrpcp_isit15,rrpcp_aistats} for dynamic RPCA and simple-ReProCS \cite{rrpcp_dynrpca} for both. The subspace tracking (ST) problem (without outliers), and with or without missing data, has been studied for much longer in \cite{past,adaptivesigproc_book,petrels,local_conv_grouse}. However, all existing guarantees for it only consider the statistically stationary setting of data being generated from a {\em single unknown} subspace.
Of course, the most general nonstationary model that allows the subspace to change at each time is not even identifiable since at least $r$ data points are needed to compute an $r$-dimensional subspace even in the no noise or missing entries case. 

In this work, we make the subspace tracking problem identifiable by assuming a piecewise constant model on subspace change. We show that it is possible to track the changing subspace to within $\epsilon$ accuracy as long as the subspace remains constant for at least $O(r \log n \log (1/ \epsilon))$ time instants, and some other assumptions hold. This is more than $r$ by only log factors. Here $n$ is the ambient space dimension.%

\subsubsection{Notation}
We use the interval notation $[a, b]$ to refer to all integers between $a$ and $b$, inclusive, and we use $[a,b): = [a,b-1]$.  $\|.\|$ denotes the $l_2$ norm for vectors and induced $l_2$ norm for matrices unless specified otherwise, and $'$ denotes transpose. We use $\M_\T$ to denote a sub-matrix of $\M$ formed by its columns indexed by entries in the set $\T$.
For a matrix $\P$ we use $\P^{(i)}$ to denote its $i$-th row. In our algorithm statements, we  use $\hat{\L}_{t; \alpha} := [\lhat_{t-\alpha + 1}, \cdots, \lhatt]$ and $\SVD_r[\M]$ to refer to the matrix of top of $r$ left singular vectors of the matrix $\M$.
A matrix $\P$ with mutually orthonormal columns is referred to as a {\em basis matrix} and is used to represent the subspace spanned by its columns.
For basis matrices $\P_1,\P_2$, we use $\SE(\P_1,\P_2):=\|(\I - \P_1 \P_1{}')\P_2\|$ as a measure of Subspace Error (distance) between their respective subspaces. This is equal to the sine of the largest principal angle between the subspaces. It is also called ``projection distance'' \cite{chordal_dist}. If $\P_1$ and $\P_2$ are of the same dimension, $\SE(\P_1, \P_2) = \SE(\P_2, \P_1)$.

We reuse the letters $C,c$ to denote different numerical constants in each use.

\subsubsection{Robust Subspace Tracking (RST) and Dynamic RPCA Problem Setting}
At each time $t$, we get a data vector $\yt \in \Re^n$ that satisfies%
\bea
\yt := \lt + \x_t + \v_t, \text{ for } t = 1, 2, \dots, \tmax \nn
\label{orpca_eq}
\eea
where  $\v_t$ is small unstructured noise, $\xt$ is the sparse outlier vector, and $\lt$ is the true data vector that lies in a fixed or slowly changing low-dimensional subspace of $\Re^n$, i.e., $\lt = \P_{(t)} \a_t$ where $\P_{(t)}$ is an $n \times r$ basis matrix with $r \ll n$ and with $\|(\I - \P_{(t-1)}\P_{(t-1)}{}')\P_{(t)}\|$ small compared to $\|\P_{(t)}\|=1$. 
We use $\T_t$ to denote the support set of $\xt$.
%
Given an initial subspace estimate, $\Phat_{0}$, the goal is to track $\Span(\P_{(t)})$ and $\lt$ either immediately or within a short delay. A by-product is that $\lt$, $\x_t$, and $\T_t$ can also be tracked on-the-fly.
The initial subspace estimate, $\Phat_0$, can be computed  by applying any of the existing RPCA solutions, e.g., PCP or AltProj, for the first roughly $r$ data points, i.e., for $\Y_{[1,t_\train]}$, with $t_\train=Cr$.

{\em Dynamic RPCA} is the offline version of the above problem. Define matrices $\L,\S,\W,\Y$ with $\L = [\l_1,\l_2, \dots \l_{\tmax}]$ and $\Y,\S,\W$ similarly defined.  The goal is to recover the matrix $\L$ and its column space with $\epsilon$ error.
We use $\rmat$ to denote the rank of $\L$. The maximum fraction of nonzeros in any row (column) of the outlier matrix $\S$ is denoted by $\outfracrow$ ($\outfraccol$).


\subsubsection{Identifiability and other assumptions}
The above problem definition does not ensure identifiability since either of $\L$ or $\S$ can be both low-rank and sparse. Moreover, if the subspace changes at every time, it is impossible to correctly estimate all the subspaces.
One way to ensure that $\L$ is not sparse is by requiring that its left and right singular vectors are dense (non-sparse) or ``incoherent'' w.r.t. a sparse vector \cite{rpca,rpca_zhang,robpca_nonconvex}.
\begin{definition} 
An $n \times r$ basis matrix $\P$ is $\mu$-incoherent if
$
\max_{i=1,2,.., n} \|\P^{(i)}\|_2^2 \le {\mu r/ n}.
$
Here $\mu$ is called the coherence parameter. It quantifies the non-denseness  of $\P$.
\label{defmu}
\end{definition}
A simple way to ensure that $\S$ is not low-rank is by imposing upper bounds on $\outfracrow$ and $\outfraccol$ \cite{rpca_zhang,robpca_nonconvex}.
One way to ensure identifiability of the changing subspaces is to assume that they are piecewise constant:%
\[
\P_{(t)} = \P_j \text{ for all } t \in [t_j, t_{j+1}), \ j=1,2,\dots, J,
\]
and to lower bound $t_{j+1}-t_j$. Let $t_0=1$ and $t_{J+1}=\tmax$.  
With this model, $\rmat = r J$ in general (except if subspace directions are repeated).
The union of the column spans of all the $\P_j$'s is equal to the span of the left singular vectors of $\L$. Thus, assuming that the $\P_j$'s are $\mu$-incoherent implies their incoherence.
We also assume that the subspace coefficients $\at$ are mutually independent over time, have identical and diagonal covariance matrices denoted by $\Lam$, and are element-wise bounded. Element-wise bounded-ness of $\at$'s, along with the statistical assumptions, is similar to incoherence of right singular vectors of $\L$ (right incoherence); see Remark \ref{rem:bounded}.
Because tracking requires an online algorithm that processes data vectors one at a time or in mini-batches, we need these statistical assumptions on the $\at$'s.
For the same reason, we also need to re-define $\outfracrow$ as the maximum fraction of nonzeroes in any row of any $\alpha$-consecutive-column sub-matrix of $\S$. Here $\alpha$ is the mini-batch size used by the RST algorithm. We will refer to it as $\outfracrow^\alpha$ to indicate this difference.


\subsubsection{Contributions}
(1) We develop a recursive projected compressive sensing (ReProCS) algorithm for RST that we call Nearly Optimal RST via ReProCS (ReProCS-NORST) because its tracking delay is nearly optimal. We will refer to it as just ``NORST'' in the sequel. NORST has a significantly improved (and simpler) subspace update step compared to all previous ReProCS-based methods. Our most important contribution is one of the first two provable guarantees for RST, and the first that ensures near optimal tracking delay, needs a near optimal lower bound on how long a subspace should remain constant, and needs minimal assumptions on subspace change.
Moreover, our guarantee also shows that NORST is online (after initialization), fast (has the same complexity as vanilla $r$-SVD), and has memory complexity of order $nr\log n \log(1/\epsilon)$.
Here ``online'' means the following. After each subspace change, the algorithm detects the change in at most $2\alpha= C f^2 r\log n$ time instants, and  after that, it improves the subspace estimate every $\alpha$ time instants\footnote{The reason that just $O(r\log n)$ samples suffice for each update is because we assume that the $\at$'s are element-wise bounded, $\vt$ is very small and with effective dimension $r$ or smaller (see Theorem \ref{thm1}). These, along with the specific structure of the PCA problem we encounter (noise/error seen by the PCA step depends on the $\lt$'s and thus has ``effective dimension" $r$), is why so few samples suffice.}. The improvement after each step is exponential and thus, one can get an $\epsilon$-accurate estimate within $K=C \log(1/\epsilon)$ such steps. 
Offline NORST also provably solves dynamic RPCA.

(2) Our guarantees for both NORST and offline NORST (Theorem \ref{thm1}) essentially hold under ``weakened'' standard RPCA assumptions and two simple extra assumptions: (i) slow subspace change and (ii) a lower bound on most outlier magnitudes. (i) is a natural assumption for static camera videos (with no sudden scene changes) and (ii) is also easy because, by definition, an ``outlier'' is a large magnitude corruption. The small magnitude ones get classified as $\vt$. Besides these, we also need that $\at$'s are mutually independent, have identical and diagonal covariance matrix $\Lam$, and are element-wise bounded. Element-wise bounded-ness,  along with the statistical assumptions on $\at$, is similar to right incoherence of $\L$; see Remark \ref{rem:bounded}.
For the initial $C r$ samples, NORST needs the outlier fractions to be $O(1/r)$ (needed to apply AltProj).
As explained in Sec. \ref{sec:why_weakened}, the extra assumptions help ensure that, {\em after initialization, NORST can tolerate a constant maximum fraction of outliers per row in any $\alpha$-column-sub-matrix of the data matrix} without assuming any outlier support generation model. This statement assumes that the condition number of the covariance of $\at$ is a constant (with $n$). As is evident from Table \ref{compare_assu}, this is better than what all existing RPCA approaches can tolerate.
For the video application, this implies that NORST tolerates slow moving and occasionally static foreground objects much better than all other approaches.  This is also corroborated by our experiments on real videos, e.g., see Fig \ref{fig:video_res} and Sec. \ref{sec:sims}.


(3) Unlike simple-ReProCS \cite{rrpcp_dynrpca} or original-ReProCS \cite{rrpcp_perf,rrpcp_isit15,rrpcp_aistats}, NORST needs only a coarse initialization which can be computed using just $C\log r$ iterations of any batch RPCA method such as AltProj applied to $C r$ initial samples. In fact, if the outlier magnitudes were very large for an initial set of $O(r\log n \log r)$ time instants, or if the outliers were absent for this much time, even a random initialization would suffice. This simple fact has two important implications. First, NORST with the subspace change detection step removed also provides an online, fast, memory-efficient, and provably correct approach for static RPCA (our problem with $J=1$, i.e., with $\lt = \P \at$). The other online solution for such a problem is ORPCA which comes with only a partial guarantee \cite{xu_nips2013_1} (the guarantee requires intermediate algorithm estimates to be satisfying certain properties).
Moreover, a direct corollary of our result is a guarantee that a minor modification of NORST-random (NORST with random initialization) also solves the subspace tracking with missing data (ST-missing) and the dynamic matrix completion (MC) problems. All existing guarantees for ST-missing \cite{petrels,local_conv_grouse} hold only for the case of a {\em single unknown subspace} and are only partial guarantees. 
From the MC perspective, NORST-random does not assume {\em any}  model on the set of observed entries. However, the tradeoff is that it needs many more observed entries.
Both these results are given in Sec. \ref{sec:extensions}.

\begin{figure*}[t!]
\centering
\resizebox{.9\linewidth}{!}{
\begin{tabular}{@{}c@{}c@{}c@{}c@{}c@{}c@{}}
\tiny{Video Frame} & \tiny{\bf{NORST} ($\bm{16.5}$\bf{ms})} & \tiny{AltProj($26.0$ms)} & \tiny{RPCA-GD($29.5$ms)} & \tiny{GRASTA} ($2.5$ms)& \tiny{PCP} ($44.6$ms) 
\\
	\includegraphics[scale=0.19, trim={2.1cm, 0cm, 2.1cm, 0.5cm}, clip]{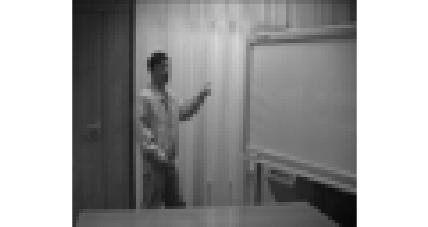}
&
	\includegraphics[scale=0.178, trim={2.1cm, 0cm, 2.1cm, 0.5cm}, clip]{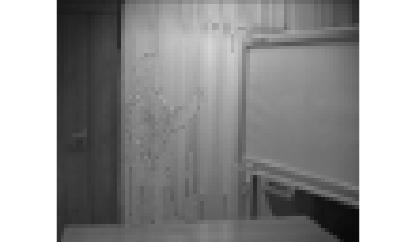}
&
	\includegraphics[scale=0.178, trim={2.1cm, 0cm, 2cm, 0.5cm}, clip]{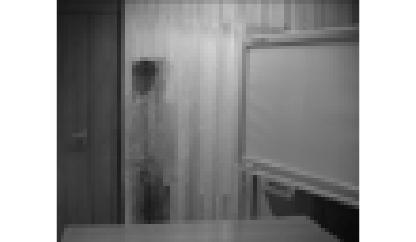}
&
	\includegraphics[scale=0.191, trim={2.5cm, 0cm, 2.1cm, 0.5cm}, clip]{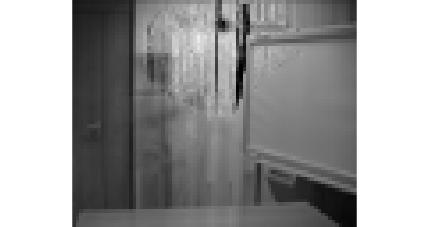}
&
	\includegraphics[scale=0.18, trim={2.1cm, 0cm, 2.1cm, 0.5cm}, clip]{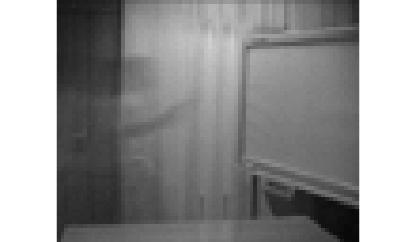}
&
	\includegraphics[scale=0.18, trim={2.1cm, 0cm, 2.1cm, 0.5cm}, clip]{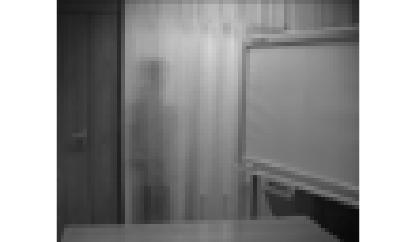}
\end{tabular}
}
\caption{\small{Background Recovery. NORST gives the best background estimate. All the other algorithms exhibit some artifacts. Only NORST does not contain the person or even his shadow. NORST is also faster than all except GRASTA (which does not work). The GRASTA output also slightly lags the actual frame. 
Time taken per frame is shown in parentheses.}}
\label{fig:video_res}
\end{figure*} 

\subsubsection{Paper Organization}
The rest of the paper is organized as follows. In Sec. \ref{sec:norst} we explain the main ideas of the NORST algorithm and present our main result for it (Theorem \ref{thm1}). We also discuss the implications of our guarantee, and provide detailed comparison with related work. In Sec. \ref{automatic_norost}, we give the complete NORST algorithm and carefully explain the subspace change detection approach. In Sec. \ref{sec:proof_outline}, we give the proof outline, the three main lemmas leading to the proof of Theorem \ref{thm1}, then also prove the lemmas.
In Sec. \ref{sec:extensions}, we provide useful corollaries for (a) Static RPCA, (b) Subspace Tracking with missing entries and (c) a simple extension to recover the guarantee of s-ReProCS from \cite{rrpcp_dynrpca}. Empirical evaluation on synthetic and real-world datasets is described in Sec. \ref{sec:sims}.
The complete proof of Theorem \ref{thm1}, of two auxiliary lemmas, and of the extensions is given in the Appendix.

\section{NORST Algorithm and Main Result}\label{sec:norst}

\subsection{NORST: Nearly-Optimal RST}
NORST starts with a ``good'' estimate of the initial subspace. This can be obtained  by $C \log r$ iterations\footnote{Using $C \log r$ iterations helps ensure that the initialization error is $O(1/\sqrt{r})$.} of AltProj applied to  $\Y_{[1,t_\train]}$ with $t_\train=Cr$.
It then iterates between (a) {\em Projected Compressive Sensing (CS) / Robust Regression\footnote{Robust Regression (with a sparsity model on the outliers) assumes that observed data vector $\y$ satisfies $\y = \Phat \a + \x + \b$ where $\Phat$ is a tall matrix (given), $\a$ is the vector of (unknown) regression coefficients, $\x$ is the (unknown) sparse outliers, $\b$ is (unknown) small noise/modeling error. An obvious way to solve this is by solving $\min_{\a,\x} \lambda \|\x\|_1+ \|\y - \Phat \a - \x\|^2$. In this, one can solve for $\a$ in closed form to get $\hat\a = \Phat'(\y-\x)$. Substituting this, the minimization simplifies to $\min_{\x} \lambda \|\x\|_1+ \|(\I - \Phat \Phat') (\y - \x)\|^2$. This is equivalent to the Lagrangian version of the projected CS problem that NORST solves (see line 7 of  Algorithm \ref{norst_basic}). 
}} in order to estimate the sparse outliers, $\xt$'s, and hence the $\lt$'s, and (b) {\em Subspace Update} to update the subspace estimate $\Phat_{(t)}$.
Projected CS proceeds as follows. At time $t$, if the previous subspace estimate, $\Phat_{(t-1)}$, is accurate enough, because of slow subspace change, projecting $\yt$ onto its orthogonal complement will nullify most of $\lt$. We compute $\tty_t:= \bpsi  \yt$ where $\bpsi := \I - \Phat_{(t-1)}\Phat_{(t-1)}{}'$. Thus, $\tty_t = \bpsi \xt + \bpsi (\lt+\vt)$ and $\| \bpsi  (\lt+\vt)\|$ is small due to slow subspace change and small $\vt$. Recovering $\xt$ from $\tty_t$ is now a CS / sparse recovery problem in small noise \cite{candes_rip}. 
We compute $\xhat_{t,cs}$ using noisy $l_1$ minimization followed by thresholding based support estimation to obtain $\That_t$. A Least Squares (LS) based debiasing step on $\That_t$ returns the final $\xhat_t$. We then estimate $\lt$ as $\lhat_t = \yt - \xhatt$.
The $\lhatt$'s are then used for the Subspace Update step which involves (i) detecting subspace change, and (ii) obtaining improved estimates of the new subspace by $K$ steps of $r$-SVD, each done with a new set of $\alpha$ samples of $\lhatt$. While this step is designed under the piecewise constant subspace assumption (needed for identifiability of $\P_{(t)}$'s), if the goal is only to get good estimates of $\lt$ or $\xt$, the method works even when this assumption may not hold, e.g., for real videos. 
For ease of understanding, we present a basic version of NORST in Algorithm \ref{norst_basic}. This assumes the change times $t_j$ are known. The actual algorithm, that we study and implement, detects these automatically. It is given as Algorithm \ref{algo:auto-reprocs-pca} in Sec. \ref{automatic_norost}.%


\subsection{Main Result}
Before stating the result, we precisely define $\outfraccol$ and $\outfracrow^\alpha$. Since NORST is an online approach that performs outlier support recovery one data vector at a time, it needs different bounds on both. Let $\outfraccol:=\max_t|\T_t|/n$.
We define $\outfracrow^\alpha$ as the maximum fraction of outliers (nonzeros) per row of any sub-matrix of $\X$ with $\alpha$ consecutive columns. To understand this precisely, for a time interval, $\J$, define
$
\gamma(\J): = \max_{i=1,2,\dots,n} \frac{1}{|\J|} \sum_{t \in \J} \one_{ \{i \in \T_t \} }
\label{def_gamma_J}
$
where $\one_{S}$ is the indicator function for statement $S$. Thus, $\sum_{t \in \J} \one_{ \{ i \in \T_t \} }$ counts the number of outliers (nonzeros) in row $i$ of $\X_\J$, and so $\gamma(\J)$ is the maximum outlier fraction in any row of the sub-matrix $\X_\J$ of $\X$. Let $\J^\alpha$ denote a time interval of duration $\alpha$. Then
$
\outfracrow^\alpha:= \max_{\J^\alpha \subseteq [1, \tmax]} \gamma(\J^\alpha).
\label{def_outfracrow}
$

We use $\that_j$ to denote the time instant at which the $j$-th subspace change time is detected by Algorithm \ref{algo:auto-reprocs-pca}.

\begin{theorem}
Consider Algorithm \ref{algo:auto-reprocs-pca} given in the next section. Let $\alpha := C f^2  r \log n$,
$\Lam:= \E[\a_1 \a_1{}']$, $\lambda^+:=\lambda_{\max}(\Lam)$, $\lambda^-:=\lambda_{\min}(\Lam)$, $f:=\lambda^+/\lambda^-$,
and let $\xmint:=\min_t \min_{i \in \T_t} (\xt)_i$ denote the minimum outlier magnitude.
Pick an $\zz \leq \min(0.01,0.03 \min_j \SE(\P_{j-1}, \P_j)^2/f)$. Let $K := C \log (1/\zz)$.
If
\ben
\item $\P_j$'s are $\mu$-incoherent; and $\at$'s are zero mean, mutually independent over time $t$, have identical covariance matrices, i.e. $\E[\at \at{}'] = \Lam$, are
element-wise uncorrelated ($\Lam$ is diagonal), are element-wise bounded (for a numerical constant $\eta$, $(\at)_i^2 \le \eta \lambda_i(\Lam))$, and are independent of all outlier supports $\T_t$;

\item  $\|\vt\|^2 \le c r \|\E[\vt \vt{}']\|$, $\|\E[\vt \vt{}']\| \le c \zz^2 \lambda^-$, $\vt$'s are zero mean, mutually independent, and independent of $\xt,\lt$;

\item $\outfraccol \le  {c_1}/{\mu r}$, $\outfracrow^\alpha \le b_0:= \frac{c_2}{f^2}$;

\item subspace change: let $\dif:=\max_j \SE(\P_{j-1}, \P_j)$, assume that
\ben
\item $t_{j+1}-t_j > (K+2)\alpha$, and
\item $\dif \le 0.8$ and $C_1 \sqrt{r \lambda^+} (\dif + 2\zz) \le \xmint$;
\een

\item initialization\footnote{This can be satisfied by applying $C \log r$ iterations of AltProj \cite{robpca_nonconvex} on the first $C r$ data samples and assuming that these have outlier fractions in any row or column bounded by $c/r$.}:
{$\SE(\Phat_0,\P_0) \le 0.25$, $C_1 \sqrt{r \lambda^+}  \SE(\Phat_0,\P_0) \le \xmint$;}%

\een
and (6) algorithm parameters are set as given in Algorithm \ref{algo:auto-reprocs-pca};
then, with probability (w.p.) at least $1 - 10 \tmax n^{-10} $, 
\[
\SE(\Phat_{(t)}, \P_{(t)}) \le
 \left\{
\begin{array}{ll}
(\zz + \dif) & \text{ if }  t \in [t_j, \that_j+\alpha), \\
 (0.3)^{k-1} (\zz + \dif) & \text{ if }  t \in [\that_j+(k-1)\alpha, \that_j+ k\alpha), \\
\zz   & \text{ if }  t \in [\that_j+ K\alpha+\alpha, t_{j+1}).
\end{array}
\right.
\]
Treating $f$ as a numerical constant, the memory complexity is $O(n \alpha) = O(n r \log n)$ and time complexity is $O(n \tmax r \log (1/\zz) )$.%
\label{thm1}
\end{theorem}

\begin{corollary}\label{cor:thm1}
Under Theorem \ref{thm1} assumptions, the following also hold:
\ben
\item {$\|\xhat_t-\xt\| = \|\lhat_t-\lt\| \le 1.2 (\SE(\Phat_{(t)}, \P_{(t)}) + \zz) \|\lt\|$ } with $\SE(\Phat_{(t)}, \P_{(t)})$ bounded as above,

\item  at all times, $t$, $\That_t = \T_t$,

\item $t_j \le \that_j \le t_j+2 \alpha$,

\item Offline-NORST (last few lines of Algorithm \ref{algo:auto-reprocs-pca}): $\SE(\Phat_{(t)}^{offline}, \P_{(t)})  \le \zz$, $\|\xhat_t^{offline}-\xt\| = \|\lhat_t^{offline}-\lt\| \le \zz \|\lt\|$ at all $t$. Its memory complexity is $O(K n \alpha) = O(n r \log n \log (1/\zz) )$.
\een
\end{corollary}

\begin{remark}[Relaxing outlier magnitudes lower bound]
The assumption on $\xmint$ (outlier magnitudes) required by Theorem \ref{thm1} can be significantly relaxed to the following which only requires that {most} outlier magnitudes are lower bounded. Assume that the outlier magnitudes are such that the following holds: $\xt$ can be split as $\xt = (\xt)_{small} + (\xt)_{large}$ with the two components being such that,
 in the $k$-th subspace update interval \footnote{$k$-th subspace update interval refers to $\J_k := [\that_j + (k-1)\alpha, \that_j + k\alpha)$ for $k>1$ and $\J_1 = [t_j, \that_j + \alpha)$ for $k=1$. The first interval also includes the subspace detection interval, $[t_j, \that_j)$, since the analysis of the projected CS step for this interval  is the same as for $[\that_j, \that_j+\alpha)$.},
 $\|(\xt)_{small}\|  \le 0.3^{k-1} (\zz+\dif)\sqrt{r \lambda^+} $ and the smallest nonzero entry of $(\xt)_{large}$ is larger than $C_1 \cdot 0.3^{k-1} (\zz+\dif)\sqrt{r \lambda^+} $. For the case of $j=0$, we need the bound to hold with $\Delta$ replaced by $\Delta_{\init} = \SE(\Phat_0, \P_0)$, and $\zz$ replaced by zero.

If there were a way to bound the element-wise error of the CS step (instead of the $l_2$ norm error), one could relax the above requirement even more. 
\label{rem:r_change}
\end{remark}

%

\setlength{\arraycolsep}{2pt}
\begin{algorithm}[t!]
\caption{\small{Basic-NORST (with $t_j$ known).  The actual algorithm that detects $t_j$ automatically is Algorithm \ref{algo:auto-reprocs-pca}.%
\\ Obtain $\Phat_0$ by $C (\log r)$ iterations of AltProj on $\Y_{[1,t_\train]}$ with $t_\train=C r$ followed by SVD on the output $\Lhat$.
}
}
\label{norst_basic}
\begin{algorithmic}[1]
\STATE \textbf{Input}:   $\yt$,  \textbf{Output}:  $\shatt$, $\lhatt$,  $\Phat_{(t)}$, $\That_t$
\STATE \textbf{Parameters:} $K \leftarrow C\log(1/\zz)$, $\alpha \leftarrow C f^2 r \log n$, $\omega_{supp} \leftarrow \xmint/2$, $\xi \leftarrow \xmint/15$,  $r$ 
\STATE {\bf Initialize: } $\tildej~\leftarrow~1$, $k~\leftarrow~1$ $\Phat_{(t_\train)} \leftarrow \hat{\pt}_{0}$
\FOR {$t > t_\train$}
\STATE $\bpsi \leftarrow \bm{I} - \hat{\pt}_{(t-1)}\hat{\pt}_{(t-1)}{}'$ \tikzmark{top} \hspace{7cm}\tikzmark{right} 
\STATE $\tty_t \leftarrow \bpsi \yt$.
\STATE $\xhat_{t,cs} \leftarrow \arg\min_{\tilde{\bm{x}}} \norm{\tilde{\bm{x}}}_1 \ \text{s.t.}\ \norm{\tilde{\bm{y}}_t - \bpsi \tilde{\bm{x}}} \leq \xi$.
\STATE $\That_t \leftarrow \{i:\ |\xhat_{t,cs}| > \omega_{supp} \}$.
\STATE $\xhat_t \leftarrow \I_{\That_t} ( \bpsi_{\That_t}{}' \bpsi_{\That_t} )^{-1} \bpsi_{\That_t}{}'\tty_t$. \tikzmark{bottom}
\STATE $\hat{\bm{\ell}}_t \leftarrow \yt - \hat{\bm{x}}_t$.  \hspace{8.6cm} \tikzmark{right2}
\IF {$t = t_j + k \alpha-1$}
\STATE $\Phat_{j, k} \leftarrow  \SVD_r[\Lhat_{t; \alpha}]$, $\Phat_{(t)} \leftarrow \Phat_{j,k}$, $k \leftarrow k + 1$.  \tikzmark{top2} 
\ELSE
\STATE $\Phat_{(t)} \leftarrow \Phat_{(t-1)}$.
\ENDIF
\IF{$t = t_j + K\alpha - 1$}
\STATE $\Phat_{j} \leftarrow \Phat_{(t)}$, $k \leftarrow 1$, $j \leftarrow j+1$\tikzmark{bottom2}
\ENDIF 
\ENDFOR
\end{algorithmic}
\AddNote{top}{bottom}{right}{Projected-CS \\ (Robust Regression).}
\AddNote{top2}{bottom2}{right2}{Subspace Update.}
\end{algorithm}

\subsubsection{Discussion}
This discussion assumes that $f$ is a constant (does not increase with $n$), i.e., it is $O(1)$.
Theorem \ref{thm1} shows that, with high probability (whp), when using NORST, the subspace change gets detected within a delay of at most $2\alpha = C f^2 (r \log n)$ time instants, and the subspace gets estimated to $\zz$ error within at most $(K+2) \alpha = C f^2 (r \log n)\log (1/\zz)$ time instants. The same is also true for the recovery error of $\xt$ and $\lt$.
Both the detection and tracking delay are within log factors of the optimal since $r$ is the minimum delay needed even in the noise-free, i..e, $\xt=\vt=0$, case. The fact that NORST can detect subspace change within a short delay can be an important feature for certain applications, e.g., this feature is used in \cite{selin_reprocs} to detect structural changes in a dynamic social network.
Moreover, if offline processing is allowed, we can guarantee recovery within normalized error $\zz$ at all time instants. This implies that offline-NORST solves the dynamic RPCA problem.

Observe that Theorem \ref{thm1} {allows a constant maximum fraction of outliers per row (after initialization), without making any assumption on how the outlier support is generated, as long as the extra assumptions discussed below hold.} We explain why this is possible in Sec. \ref{sec:why_weakened}. Of course, for the initial $C r$ samples, NORST needs $\outfracrow^{Cr} \in O(1/r)$ (needed to apply AltProj).
Also, the memory complexity guaranteed by Theorem \ref{thm1} is nearly $\tmax/r$ times better than that of all existing RPCA solutions; see Table \ref{compare_assu}. The time complexity is worse than that of only NO-RMC\footnote{NO-RMC is so fast because it is actually a robust matrix completion solution and it deliberately undersamples the entire data matrix $\Y$ to get a faster RPCA algorithm.}, but NO-RMC needs $\tmax \ge c n$ (unreasonable requirement for videos which often have much fewer frames $\tmax$ than the image size $n$). Finally, NORST also needs outlier fraction per column to be $O(1/r)$ instead of $O(1/\rmat)$. If $J$ is large, e.g. if $J=d/(r\log n)$, it is possible that $\rmat \gg r$.

We should clarify that NORST allows $\outfracrow^\alpha \in O(1)$ but this does not necessarily imply that the number of outliers in each row can be this high. The reason is it only allows the fraction per column to only be $O(1/r)$. Thus, for a matrix of size $n \times \alpha$, it allows the total number of outliers to be $O(\min(n\alpha, n\alpha/r)) = O(n\alpha/r)$. Thus the average fraction allowed is only $O(1/r)$.%

NORST needs the following extra assumptions. The main extra requirement is that $\xmint$ be lower bounded as given in the last two assumptions of Theorem \ref{thm1}, or as stated in Remark \ref{rem:r_change}. The lower bound on $\xmint$ is reasonable\footnote{requires $\xmint$ to be $C \sqrt{\lambda^+}$ or larger.} as long as the initial subspace estimate is accurate enough and the subspace changes slowly enough so that both $\dif$ and $\SE(\Phat_0,\P_0)$ are $O(1/\sqrt{r})$. This requirement may seem restrictive on first glance but actually is not. The reason is that $\SE(.)$ is only measuring the largest principal angle. This bound on $\SE$ still allows the chordal distance between the two subspaces to be $O(1)$. Chordal distance  \cite{chordal_dist}  is the $l_2$ norm of the vector containing the sine of all principal angles. 
The second related extra requirement is an upper bound on $\dif$ (slow subspace change) which depends on the value of $\xmint$. We discuss this point next. Other than these two, NORST only needs simple statistical assumptions on $\at$'s. The zero-mean assumption  is a minor one. The assumption that $\Lam$ be diagonal is also minor\footnote{It only implies that $\P_j$ is the matrix of principal components of $\E[\L_j \L_j']$ where $\L_j:=[\l_{t_j}, \l_{t_j+1},\dots, \l_{t_{j+1}-1}]$.}.
In the video setting, zero-mean can be ensured by subtracting the empirical average of the background images computing using the first $t_\train$ frames. Mutual independence of $\at$'s holds if the changes in each background image w.r.t. a ``mean'' background are independent, when conditioned on their subspace. This is valid, for example, if the background changes are due to illumination variations or due to moving curtains (see Fig. \ref{fig:mr_full}). Moreover, by using the approach of \cite{rrpcp_aistats}, it is possible to relax this to just requiring that the $\at$'s satisfy an autoregressive model over time. Element-wise boundedness, along with the above, is similar to right incoherence (see Remark \ref{rem:bounded}).

\begin{table*}[t!]
\caption{\small{Comparing RPCA and RST solutions. All algorithms also require left and right incoherence or left incoherence and $\at$'s element-wise bounded (this, along with the i.i.d. assumption on $\at$'s, is similar to right incoherence), and hence these are not compared. The incoherence parameter $\mu$ and the condition numbers are treated as constants in this table.
In general $\rmat=rJ$. {\cred \bf Strong or unrealistic assumptions are shown in red.}
PCP(C) \cite{rpca}, PCP(H) \cite{rpca2,rpca_zhang}, mod-PCP \cite{zhan_pcp_jp},  AltProj \cite{robpca_nonconvex}, RPCA-GD \cite{rpca_gd}, NO-RMC \cite{rmc_gd}, orig-ReProCS \cite{rrpcp_isit15,rrpcp_aistats}, s-ReProCS \cite{rrpcp_dynrpca}. 
}}
\begin{center}
\renewcommand*{\arraystretch}{1.3}
\resizebox{\linewidth}{!}{
\begin{tabular}{|l|l|l|l|c|}
\hline
Algorithm & Outlier tolerance & Other Assumptions  & Memory, Time, & \# params \\
\hline
PCP(C) \cite{rpca}  & $\outfracrow \in O(1)$& {\cred \bf outlier support: unif. random,} & Memory: $O(n \tmax)$   & zero \\
       & $\outfraccol \in O(1) $ &   $\rmat \leq {c\min(n,\tmax)}/{\log^2 n}$     &  Time: $O(n \tmax^2 \frac{1}{\epsilon})$   & \ \\  
\hline
PCP(H) \cite{rpca_zhang}   & $\outfracrow \in O\left(1/\rmat\right)$ & $\tmax \ge c \rmat $ & Memory: $O(n \tmax)$   & 1 \\
   & $\outfraccol \in O\left( 1/\rmat\right)$      & \    & Time: $O(n \tmax^2 \frac{1}{\epsilon})$     &   \\
\hline
AltProj \cite{robpca_nonconvex}   & $\outfracrow \in O\left(1/\rmat\right)$ & $\tmax \ge c \rmat $ & Memory: $O(n \tmax)$   & 2 \\
   & $\outfraccol \in O\left( 1/\rmat\right)$      & \    & Time: $O(n \tmax \rmat^2 \log \frac{1}{\epsilon})$   &   \\
\hline
RPCA-GD \cite{rpca_gd}  & $\outfracrow \in O(1/\rmat^{1.5})$     &   $\tmax \ge c \rmat$ & Memory: $O(n \tmax)$   & 5 \\
       & $\outfraccol \in O(1/\rmat^{1.5})$ & \ & Time: $O(n \tmax \rmat \log \frac{1}{\epsilon})$   &  \\
\hline
NO-RMC \cite{rmc_gd} & $\outfracrow \in O\left(1/\rmat\right)$ & {\cred \bf $\mathbf{c_2 n \ge \tmax \ge c n}$}  & Memory: $O(n \tmax)$   & 4 \\
      & $\outfraccol  \in O(1/\rmat)$            &  & Time: $O(n \rmat^3 \log^2 n \log^2 \frac{1}{\epsilon})$   &  \\
\hline
\hline
mod-PCP \cite{zhan_pcp_jp}  & $\outfracrow \in O(1)$ & {\cred \bf outlier support: unif. random,} & Memory: $O(nr \log^2 n)$   &  \\
                                        & $\outfraccol \in O(1)$ & {\cred \bf slow subspace change (unrealistic),} & Time: $O(n \tmax r \log^2 n \frac{1}{\epsilon} )$   &   \\
& \ & $\rmat \leq \frac{c\min(n,\tmax)}{\log^2 n}$ & Delay: $\infty$ & \ \\
\hline
%
{orig-ReProCS} \cite{rrpcp_isit15,rrpcp_aistats} &         {$\outfracrow^\alpha  \in O(1) $} &  {\cred \bf outlier support: moving object model,}      & {Memory: $O(nr^2/{\epsilon^2} )$}  & 5 \\
tracking delay  &         {$\outfraccol \in O(1/\rmat) $} & {\cred \bf unrealistic subspace change model,}    &  {Time: $O(n \tmax r \log \frac{1}{\epsilon})$ }  & \\
too large  &         \ &  {\cred \bf changed eigenvalues small for some time,}    &  \  & \ \\
  &         \ & outlier mag. lower bounded,      &  \ & \ \\
        &  & init data: AltProj assu's, &  \  & \ \\
\                           &   & $\at$'s independent, {\cred \bf $\tmax \ge C r^2/{\epsilon^2}$}    &   & \ \\
\hline
{s-ReProCS:} \cite{rrpcp_dynrpca} &        {$\outfracrow^\alpha  \in O(1) $} & {\cred \bf subspace change: only 1 direc at a time},      & {Memory: $O(nr \log n \log \frac{1}{\epsilon} )$}  & 4 \\
solves  &       {$\outfraccol \in O(1/r) $} & outlier mag. lower bounded,      & {Time: $O(n \tmax r \log \frac{1}{\epsilon})$}  &  \\
RST with  &  & $\at$'s i.i.d., $\tmax \ge C r \log n  \log \frac{1}{\epsilon}$.                                     &  \ & \ \\
sub-optimal delay                           &   &  init data: AltProj assumptions    & \  & \ \\
\hline
{\bf NORST} &        {$\outfracrow^\alpha  \in O(1) $} & subspace change: mild,       & {Memory: $O(nr \log n \log \frac{1}{\epsilon} )$}  & 4 \\
{\bf  (this work):} solves    &       {$\outfraccol \in O(1/r) $} & outlier mag. lower bounded,  & {Time: $O(n \tmax r \log \frac{1}{\epsilon})$}  &  \\
{\bf RST with}                              &       &   $\at$'s i.i.d., $\tmax \ge C r \log n  \log \frac{1}{\epsilon}$,        &  \ & \ \\
{\bf near-optimal delay}            &  & first $C r$ samples: AltProj assumptions    & \  & \ \\
\hline
\end{tabular}
}
\label{compare_assu}
\end{center}
\vspace{-0.15in}
\end{table*}


\subsubsection{Outlier v/s Subspace Assumptions}
When there are fewer outliers in the data or when outliers are easy to detect, one would expect to need weaker assumptions on the true data's subspace and/or on its rate of change. This is indeed true. The $\outfraccol$ bound relates $\outfraccol$ to $\mu$ (not-denseness parameter) and $r$ (subspace dimension). The upper bound on $\Delta$ implies that, if $\xmint$ is larger (outliers are easier to detect), a larger amount of subspace change $\Delta$ can be tolerated. The relation of $\outfracrow$ to rate of subspace change is not evident from the way the guarantee is stated above because we have assumed $\outfracrow \le b_0:=c/f^2$ with $c$ being a numerical constant, and used this to get a simple expression for $K$. If we did not do this, we would get $K = C \lceil \frac{1}{-\log (\sqrt{b_0} f )} \log (\frac{c \Delta}{0.8\zz}) \rceil$, see Remark \ref{K_long_rem}. Since we need $t_{j+1}-t_j \ge (K+2)\alpha$, a smaller $b_0$ means a larger $\Delta$ can be tolerated for the same delay, or vice versa.%


\subsubsection{Algorithm Parameters}
Algorithm \ref{norst} assumes knowledge of 4 model parameters: $r$, $\lambda^+$, $\lambda^-$ and $\xmint$ to set the algorithm parameters.
The initial dataset used for estimating $\Phat_0$ (using AltProj) can be used to get an accurate estimate of $r$, $\lambda^-$ and $\lambda^+$ using standard techniques.
Thus one really only needs to set $\xmint$. If continuity over time is assumed, we can let it be time-varying and set it as $\min_{i \in \That_{t-1}}|(\xhat_{t-1})_i|$ at $t$.

\subsubsection{Related Work}
For a summary of comparisons, see Table \ref{compare_assu}.
In terms of other solutions for provably correct RST or dynamic RPCA, there is very little work. For RST, there is only one other provable algorithm, simple-ReProCS (s-ReProCS) \cite{rrpcp_dynrpca}. 
This has the same tracking delay and memory complexity as NORST, however, it assumes that {\em only one} subspace direction can change at each change time. This is a more restrictive model than ours. Moreover, it implies that the tracking delay of s-ReProCS is $r$-times sub-optimal. Also, s-ReProCS uses a projection-SVD step for subspace update (as opposed to simple SVD in NORST). These two facts imply that it needs an $\epsilon$-accurate subspace initialization in order to ensure that the later changed subspaces can be tracked with $\epsilon$-accuracy. Thus, it does not provide a static RPCA or subspace tracking with missing data solution. 

For dynamic RPCA, the earliest result was a partial guarantee (a guarantee that depended on intermediate algorithm estimates satisfying certain assumptions) for the original reprocs approach (original-ReProCS) \cite{rrpcp_perf}. This was followed up by two complete guarantees for reprocs-based approaches with minor modifications \cite{rrpcp_isit15,rrpcp_aistats}. For simplicity we will still call these ``original-ReProCS''. These guarantees needed very strong assumptions and their tracking delay was $O(nr^2/\epsilon^2)$. Since $\epsilon$ can be very small, this factor can be quite large, and hence one cannot claim that original-ReProCS solves RST. Our work is a very significant improvement over all these works. (i) The guaranteed memory complexity, tracking delay, and required delay between subspace change times of NORST are all $r/\epsilon^2$ times lower than that of original-ReProCS.
%
(ii) All the original-ReProCS guarantees needed a very specific assumption on how the outlier support could change. They required an outlier support model inspired by a video moving object that moves in one direction for a long time; and whenever it moves, it must move by a fraction of $s:=\max_t|\T_t|$. This is very specific model with the requirement of moving by a fraction of $s$ being the most restrictive. Our result removes this model and replaces it with just a bound on $\outfracrow$. We explain in the last para of Sec. \ref{only_outline} why this is possible.
(iii) The subspace change model assumed in \cite{rrpcp_isit15, rrpcp_aistats} required a few new directions, that were orthogonal to $\P_{j-1}$, to be added at time $t_j$ and some others to be removed. This is an unrealistic model for slow subspace change, e.g., in 3D, it implies that the subspace needs to change from the x-y plane to the y-z plane. Moreover because of this model, their results needed the ``energy'' (eigenvalues) along the newly added directions to be small for a period of time after each subspace change. This is a strong (and not easy to interpret) requirement. Our result removes all these requirements and replaces them with a bound on $\SE(\P_{j-1},\P_j)$ which is much more realistic. Thus, in 3D, our result allows the x-y plane to change to a slightly tilted x-y plane.
%

An approach called modified-PCP (mod-PCP) was proposed to solve the problem of {\em RPCA with partial subspace knowledge} \cite{zhan_pcp_jp}. A corollary of its guarantee shows that it can also be used to solve dynamic RPCA \cite{zhan_pcp_jp}. However, since it adapted the PCP proof techniques from \cite{rpca}, its pros and cons are similar to those of PCP, e.g., it also needs a uniformly randomly generated outlier support.
As can be seen from Table \ref{compare_assu}, its pros and cons are similar to those of the PCP result by \cite{rpca} (PCP(C)) discussed below. 

We also provide a comparison with provably correct RPCA approaches in Table \ref{compare_assu}. In summary, NORST has significantly better memory complexity than all of them, all of which are batch; it has the best outlier tolerance (after initialization), and the second-best time complexity, as long as its extra assumptions hold. It can also detect subspace change quickly, which can be a useful feature. Consider outlier tolerance. PCP(H), AltProj, RPCA-GD, and NO-RMC need both $\outfracrow$ and $\outfraccol$ to be $O(1/\rmat)$; PCP(C) \cite{rpca} and modified-PCP \cite{zhan_pcp_jp} need the outlier support to uniformly random (strong requirement: for video it implies that objects are very small sized and jumping around randomly); and original-ReProCS needs it to satisfy a very specific moving object model described above (restrictive). Instead, after initialization, NORST only needs $\outfracrow^\alpha \in O(1)$ and $\outfraccol \in O(1/r)$.


\subsection{The need for extra assumptions}\label{sec:why_weakened}
As noted in \cite{robpca_nonconvex}, the standard RPCA problem (that only assumes left and right incoherence of $\L$ and nothing else) cannot tolerate a bound on outlier fractions in any row or any column that is larger than $1/\rmat$ \footnote{The reason is this: let $b_0 = \outfracrow$, one can construct a matrix $\X$ with $b_0$ outliers in some rows that has rank equal to $1/b_0$. A simple way to do this would be to let the support and nonzero entries of $\X$ be constant for $b_0 \tmax$ columns before letting either of them change. Then the rank of $\X$ will be $\tmax/(b_0 \tmax)$. A similar argument can be used for $\outfraccol$.}.
The reason NORST can tolerate a constant $\outfracrow^\alpha$ bound is because it uses extra assumptions. We explain the need for these here. It recovers $\xt$ first and then $\lt$ and does this at each time $t$.
When recovering $\xt$, it exploits ``good'' knowledge of the subspace of $\lt$ (either from initialization or from the previous subspace's estimate and slow subspace change), but it has no way to deal with the residual error, $\b_t:= (\I - \Phat_{(t-1)} \Phat_{(t-1)}{}') \lt$, in this knowledge. Since the individual vector $\b_t$ does not have any structure that can exploited\footnote{However the $\b_t$'s arranged into a matrix do form a low-rank matrix whose approximate rank is $r$ or even lower (if not all directions change). If we try to exploit this structure we end up with a modified-PCP \cite{zhan_pcp_jp} type approach. This needs the uniform random support assumption (used in its guarantee). Or, if the \cite{rpca_zhang} approach were used for its guaratee, for identifiability reasons similar to the one described above, it will still not tolerate outlier fractions larger than $1/r_\new$ where $r_\new$ is the (approximate) rank of the matrix formed by the $\b_t$'s.}, the error in recovering $\xt$ cannot be lower than $C \|\b_t\|$. This means that, to correctly recover the support of $\xt$, $\xmint$ needs to be larger than $C\|\b_t\|$. This is where the $\xmint$ lower bound comes from.
As we will see in Sec. \ref{sec:proof_outline}, correct support recovery is needed to ensure that the subspace estimate can be improved with each update. In particular, it helps ensure that the error vectors $\et:=\xt - \xhat_t$ in a given subspace update interval are mutually independent, when conditioned on the $\yt$'s from all past intervals.
This step also uses element-wise boundedness of the $\at$'s along with their mutual independence and identical covariances. 

\begin{algorithm}[ht!]
\caption{Automatic-NORST.\\ Obtain $\Phat_0$ by $C (\log r)$ iterations of AltProj on $\Y_{[1,t_\train]}$ with $t_\train=C r$ followed by SVD on the output $\Lhat$.
}
\label{algo:auto-reprocs-pca}
\label{norst}
\begin{algorithmic}[1]
\STATE \textbf{Input}:  $\Phat_0$, $\yt$,  \textbf{Output}:  $\shatt$, $\lhatt$,  $\Phat_{(t)}$
\STATE \textbf{Parameters:} $K \leftarrow C\log(1/\zz)$, $\alpha \leftarrow C f^2 r \log n$, $\omega_{supp} \leftarrow \xmint/2$, $\xi \leftarrow \xmint/15$,  $\lthres \leftarrow 2 \zz^2 \lambda^+$, $r$.
\STATE $\Phat_{(t_\train)} \leftarrow \hat{\pt}_{0}$;  $\tildej~\leftarrow~1$, $k~\leftarrow~1$
\STATE $\mathrm{phase} \leftarrow \mathrm{update}$; $\that_{0} \leftarrow t_\train$;
\FOR {$t > t_\train$}
\STATE Lines $5-10$ of Algorithm \ref{norst_basic} \tikzmark{top0} \hspace{6.5cm}\tikzmark{right0} \tikzmark{bottom0}
\IF{$\text{phase} = \text{detect}$ and $t = \hat{t}_{j-1, fin} + u\alpha$}
\STATE $\bphi \leftarrow (\I - \Phat_{j-1}\Phat_{j-1}{}')$. \tikzmark{top} \hspace{7cm}\tikzmark{right}
\STATE $\bm{B} \leftarrow \bphi\Lhat_{t, \alpha}$
\IF {$\lambda_{\max}(\bm{B}\bm{B}{}') \geq \alpha \lthres$}
\STATE $\text{phase} \leftarrow \text{update}$, $\hat{t}_j \leftarrow t$,
\ENDIF
\tikzmark{bottom}
\ENDIF \tikzmark{top2} \hspace{10.3cm}\tikzmark{right2}
\IF{$\text{phase} = \text{update}$} 
\IF {$t = \that_j + u \alpha - 1$ for $u = 1,\ 2,\ \cdots,$}
\STATE $\Phat_{j, k} \leftarrow  \SVD_r[\Lhat_{t; \alpha}]$, $\Phat_{(t)} \leftarrow \Phat_{j,k}$, $k \leftarrow k + 1$.
\ELSE
\STATE $\Phat_{(t)} \leftarrow \Phat_{(t-1)}$ \tikzmark{bottom2}
\ENDIF 
\IF{$t = \that_j + K\alpha - 1$}
\STATE $\hat{t}_{j, fin} \leftarrow t$, $\Phat_{j} \leftarrow \Phat_{(t)}$
\STATE $k \leftarrow 1$, $j \leftarrow j+1$, $\text{phase} \leftarrow \text{detect}$.
\ENDIF
\ENDIF
\ENDFOR
\STATE {\bf Offline NORST: }
At $t = \that_j + K \alpha$, for all $t \in [\that_{j-1}+ K \alpha,  \that_j + K \alpha-1]$,  \tikzmark{top3} \hspace{0.3cm}\tikzmark{right3}
\STATE $\Phat_{(t)}^{\mathrm{offline}} \leftarrow [\Phat_{j-1}, (\I - \Phat_{j-1} \Phat_{j-1}{}') \Phat_{j}]$
\STATE $\bpsi \leftarrow \I - \Phat_{(t)}^{\mathrm{offline}} \Phat_{(t)}^{\mathrm{offline}}{}'$
\STATE $\xhatt^{\mathrm{offline}} \leftarrow \I_{\That_t} (\bpsi_{\That_t}{}'\bpsi_{\That_t})^{-1} \bpsi_{\That_t}{}' \yt$
\STATE $\lhatt^{\mathrm{offline}} \leftarrow \yt - \xhatt^{\mathrm{offline}}$. \tikzmark{bottom3}
\end{algorithmic}
\AddNoteOne{top0}{bottom0}{right0}{Projected CS.}
\AddNote{top}{bottom}{right}{Subspace Detect Phase.}
\AddNote{top2}{bottom2}{right2}{Subspace Update Phase.}
\AddNote{top3}{bottom3}{right3}{Offline NORST.}
\end{algorithm}

\section{Automatic NORST}\label{automatic_norost}
We present the actual NORST algorithm (automatic NORST) in Algorithm \ref{algo:auto-reprocs-pca}. The main idea why automatic NORST works is the same as that of the basic algorithm with the exception of the additional subspace detection step. The subspace detection idea is borrowed from \cite{rrpcp_dynrpca}, although its correctness proof has differences because we assume a much simpler subspace change model.
In Algorithm \ref{algo:auto-reprocs-pca}, the subspace update stage toggles between the ``detect'' phase and the ``update'' phase. It starts in the ``update'' phase with $\that_0 = t_\train$. We then perform $K$ $r$-SVD steps with the $k$-th one done at $t = \that_0 + k \alpha-1$. Each such step uses the last $\alpha$ estimates, i.e., uses $\Lhat_{t; \alpha}$. Thus at $t = \that_0 + K \alpha - 1$, the subspace update of $\P_0$ is complete. At this point, the algorithm enters the ``detect'' phase.

For any $j$, if the $j$-th subspace change is detected at time $t$, we set  $\that_j=t$. At this time, the algorithm enters the ``update'' (subspace update) phase. We then perform $K$ $r$-SVD steps with  the $k$-th $r$-SVD step done at $t = \that_j + k \alpha-1$ (instead of at $t=t_j+ k\alpha-1$). Each such step uses the last $\alpha$ estimates, i.e., uses $\Lhat_{t; \alpha}$ Thus, at $t = \that_{j,fin} = \that_j + K\alpha - 1$, the update is complete. At this time, the algorithm enters the ``detect'' phase again.%

To understand the change detection strategy, consider the $j$-th subspace change. Assume that the previous subspace $\P_{j-1}$ has been accurately estimated by $t= \that_{j-1,fin} = \that_{j-1}+K\alpha-1$ and that $\that_{j-1,fin} < t_j$. Let $\Phat_{j-1}$ denote this estimate. At this time, the algorithm enters the ``detect'' phase in order to detect the next ($j$-th) change. Let $\B_t:=(\I-\Phat_{j-1}\Phat_{j-1}{}')  \Lhat_{t;\alpha}$. For every $t = \that_{j-1,fin} + u \alpha-1$, $u=1,2,\dots$, we detect change by checking if the maximum singular value of $\B_t$ is above a pre-set threshold, $ \sqrt{\lthres \alpha}$, or not.

We claim that, whp, under assumptions of Theorem \ref{thm1}, this strategy has no ``false subspace detections'' and correctly detects change within a delay of at most $2\alpha$ samples. The former is true because, for any $t$ for which $[t-\alpha+1, t] \subseteq [ \that_{j-1,fin}, t_j)$, all singular values of the matrix $\B_t$ will be close to zero (will be of order $\zz \sqrt{\lambda^+}$) and hence its maximum singular value will be below $ \sqrt{\lthres \alpha}$. Thus, whp, $\that_j \ge t_j$. To understand why the change {\em is} correctly detected within $2 \alpha$ samples, first consider $t =\that_{j-1,fin} + \lceil \frac{t_j - \that_{j-1,fin}}{\alpha} \rceil \alpha:= t_{j,*} $. Since we assumed that $\that_{j-1,fin}< t_j$ (the previous subspace update is complete before the next change), $t_j$ lie in the interval $[t_{j,*}-\alpha+1, t_{j,*}]$. Thus, not all of the $\lt$'s in this interval lie in the new subspace. Depending on where in the interval $t_j$ lies, the algorithm may or may not detect the change at this time.  However, in the {\em next} interval, i.e., for $t \in [t_{j,*}+1, t_{j,*} + \alpha]$,  all of the $\lt$'s lie in the new subspace. We can prove that $\B_t$ for this time $t$ {\em will} have maximum singular value that is above the threshold.
Thus, if the change is not detected at $t_{j,*}$, whp, it {\em will} get detected at $t_{j,*} + \alpha$. Hence one can show that, whp, either $\that_j = t_{j,*}$, or $\that_j =t_{j,*} + \alpha$, i.e., $t_j \le \that_j \le t_j + 2\alpha$ (see Appendix \ref{sec:proof}).

\subsubsection{Time complexity}
Consider initialization.
To ensure that $\SE(\Phat_0, \P_0) \in O(1/\sqrt{r})$, we need to use $C \log r$ iterations of AltProj. Since there is no lower bound in the AltProj guarantee on the required number of matrix columns (except the trivial lower bound of rank) \cite{robpca_nonconvex}, we can use $t_\train = C r$ frames for initialization. Thus the initialization complexity is $O(n t_\train r^2 \log(\sqrt{r}) = O(n r^3 \log r)$ \cite{robpca_nonconvex}. The projected-CS step complexity is equal to the cost of a matrix vector multiplication with the measurement matrix times negative logarithm of the desired accuracy in solving the $l_1$ minimization problem. Since the measurement matrix for the CS step is $\I - \Phat_{(t-1)} \Phat_{(t-1)}{}'$, the cost per CS step (per frame) is $O(nr\log(1/\epsilon))$ \cite{l1_best} and so the total cost is $O((d-t_\train) nr\log(1/\epsilon))$.
 The subspace update involves at most $((d - t_\train)/ \alpha)$ rank $r$-SVD's on $n\times \alpha$ matrices all of which have constant eigen-gap (this is proved in the proof of tTheorem \ref{cor:pcasddn} from \cite{pca_dd} which we use to show correctness of this step). Thus the total time for subspace update steps is at most $((d-t_\train)/\alpha)*O(n \alpha r \log(1/\epsilon)) = O((d-t_\train) n r \log(1/\epsilon))$ \cite{musco2015randomized}. Thus the running time of the complete algorithm is $O(ndr\log(1/\epsilon) + nr^3 \log r)$. As long as $r^2 \log r \le d \log(1/\epsilon)$, the time complexity of the entire algorithm is $O(n d r\log(1/\epsilon))$.

\newcommand{\vv}{\bm{v}} 
\begin{remark}[Relating our assumptions to right incoherence of $\L_j := \L_{[t_j, t_{j+1})}$ \cite{rpca_zhang}]\label{rem:bounded}
From our assumptions, $\L_j = \P_j \A_j$ with $\A_j:= [\a_{t_j},\a_{t_j+1},\dots \a_{t_{j+1}-1}]$, the columns of $\A_j$ are zero mean, mutually independent, have identical covariance $\Lam$, $\Lam$ is diagonal, and are element-wise bounded as specified by Theorem \ref{thm1}. Let $d_j := t_{j+1}-t_j$.
Define a diagonal matrix $\Sigma$ with $(i,i)$-th entry $\sigma_i$ and with $\sigma_i^2 := \sum_t (a_t)_i^2 / d_j$. Define a $d_j \times r$ matrix $\tilde\V$ with the $t$-th entry of the $i$-th column being $(\tilde\vv_i)_t: = (\a_t)_i/ (\sigma_i \sqrt{d_j})$. Then, $\L_j = \P_j \Sigma \tilde{\V}'$ and each column of $\tilde\V$ is unit 2-norm. Also, from the bounded-ness assumption,
$(\tilde\vv_i)_t^2 \leq \eta \frac{\lambda_i}{\sigma_i^2} \cdot \frac{1}{d_j}$ where $\eta$ is a numerical constant.

Observe that $\P_j \Sigma \tilde\V'$ is not exactly the SVD of $\L_j$ since the columns of $\tilde\V$ are not necessarily exactly mutually orthogonal. However, if $d_j$ is large enough, one can argue using any law of large numbers' result (e.g., Hoeffding inequality), that the columns of $\tilde\V$ are approximately mutually orthogonal whp. Also, whp, $\sigma_i^2 \ge 0.99 \lambda_i$. This also follows using Hoeffding\footnote{The first claim uses all the four assumptions on $\at$; the second claim uses all assumptions except diagonal $\Lam$}.
Thus, our assumptions imply that, whp, $(\tilde\vv_i)_t^2 \leq C/{d_j}$.

If one interprets $\tilde\V$ as an ``approximation'' to the right singular vectors of $\L_j$,  this is the right incoherence assumed by \cite{rpca_zhang} and slightly stronger than what is assumed by \cite{rpca,robpca_nonconvex} and others (these  require that the squared norm of each row of the matrix of right singular vectors be bounded by $C r / d_j$). 

The claim that ``$\tilde\V$ can be interpreted as an ``approximation'' to the right singular vectors of $\L_j$'' is not rigorous. But it is also not clear how to make it rigorous since our work uses statistical assumptions on the $\at$'s. To get the exact SVD of $\L_j$, we need the SVD of $\A_j$. Suppose $\A_j \svdeq \R \Sigma \V'$, then $\L_j \svdeq (\P_j \R)  \Sigma \V'$. Here $\R$ will be an $r \times r$ orthonormal matrix. Now it is not clear how to relate the element-wise bounded-ness assumption on $\at$'s to an assumption on entries of $\V$, since now there is no easy expression for each entry of $\V$ or of the entries of $\Sigma$ in terms of $\at$ (since $\R$ is unknown).
\end{remark}

\section{Proof Outline and (most of the) Proof} \label{sec:proof_outline}
In this section we first give the main ideas of the proof (without formal lemmas). We then state the three main lemmas and explain how they help prove Theorem \ref{thm1}. After this, we prove the three lemmas.

\subsection{Main idea of the proof} \label{only_outline}
It is not hard to see that the ``noise'' $\b_t:=\bm\Psi (\lt+ \vt)$ seen by the projected CS step is proportional the error between the subspace estimate from $(t-1)$ and the current subspace. 
Moreover, incoherence (denseness) of the $\P_{(t)}$'s and slow subspace change together imply that $\bpsi$ satisfies the restricted isometry property (RIP) \cite{rrpcp_perf}.
Using this, a result for noisy $l_1$ minimization \cite{candes_rip}, and the lower bound assumption on outlier magnitudes, one can ensure that the CS step output is accurate enough and the outlier support $\T_t$ is correctly recovered. With this, we have that $\lhat_t = \lt + \vt - \et$ where $\et:=\x_t -\xhat_t$ satisfies
$
\et= \I_{\T_t} (\bm\Psi_{\T_t}{}'\bm\Psi_{\T_t})^{-1} \I_{\T_t}{}' \bm\Psi' \lt
$
and $\|\et\| \le C \|\b_t\|$. 
Consider subspace update. Every time the subspace changes, one can show that the change can be detected within a short delay. After that, the $K$ SVD steps help get progressively improved estimates of the changed subspace. To understand this, observe that, after a subspace change, but before the first update step, $\b_t$ is the largest and hence, $\et$, is also the largest for this interval. However, because of good initialization or because of slow subspace change and previous subspace correctly recovered (to error $\zz$), neither is too large. Both are proportional to $(\zz + \dif)$, or to the initialization error. 
Using the idea below, we can show that we get a ``good'' first estimate of the changed subspace.

The input to the PCA step is $\lhat_t$ and the noise seen by it is $\et$. Notice that $\et$ depends on the true data $\lt$. Hence this is a setting of PCA in data-dependent noise \cite{corpca_nips,pca_dd}.
From \cite{pca_dd}, it is known that the subspace recovery error of the PCA step is proportional to the ratio between the time averaged noise power plus time-averaged signal-noise correlation, $(\|\sum_t \E[\et \et{}']\| + \|\sum_t \E[\lt \et{}'\|)/\alpha$, and the minimum signal space eigenvalue, $\lambda^-$.
The instantaneous value of noise power is $(\dif+\zz)^2$ times $\lambda^+$ while that of signal-noise correlation is of order $(\dif+\zz)$ times $\lambda^+$. However, using the fact that $\et$ is sparse with support $\T_t$ that changes enough over time so that $\outfracrow^\alpha$ is bounded, one can argue (using Cauchy-Schwartz) that their time averaged values are $\sqrt{\outfracrow^\alpha}$ times smaller. As a result, after the first subspace update, the subspace recovery error is at most $4 \sqrt{\outfracrow^\alpha} (\lambda^+/\lambda^-)$ times $(\dif+\zz)$. Since $\outfracrow^\alpha (\lambda^+/\lambda^-)^2$ is bounded by a constant $c_2 < 1$, this means that, after the first subspace update, the subspace error is at most $\sqrt{c_2}$ times $(\dif+\zz)$.

This, in turn, implies that $\|\b_t\|$, and hence $\|\et\|$, is also $\sqrt{c_2}$ times smaller in the second subspace update interval compared to the first. This, along with repeating the above argument, helps show that the second estimate of the changed subspace is $\sqrt{c_2}$ times better than the first and hence its error is $(\sqrt{c_2})^2$ times $(\dif+\zz)$. Repeating the argument $K$ times, the $K$-th estimate has error $(\sqrt{c_2})^K$ times $(\dif+\zz)$. Since $K = C \log(1/\zz)$, this is an $\zz$ accurate estimate of the changed subspace.

A careful application of the result of \cite{pca_dd} is the reason why we are able to remove the moving object model assumption on the outlier support needed by the earlier guarantees for original-ReProCS \cite{rrpcp_isit15,rrpcp_aistats}. Applied to our problem, this result requires $\|\sum_{t \in \J^\alpha} \itt \itt{}' /\alpha\|$ to be bounded by a constant less than one. It is not hard to see that $\max_{\J^\alpha \in [1,\tmax]} \|\sum_{t \in \J^\alpha} \itt \itt{}'/\alpha\| = \outfracrow^{\alpha}$. To understand this simply, the matrix $\sum_{t \in \J^\alpha} \itt \itt{}'$ is diagonal, and the $i$-th diagonal entry counts the number of time the index $i$ appears in the support set $\T_t$ in the interval $\J^{\alpha}$ which is precisely the definition of $\outfracrow^{\alpha} \cdot \alpha$. This is also why a constant bound on $\outfracrow^{\alpha}$ suffices for our setting. On the other hand the guarantees of \cite{rrpcp_isit15,rrpcp_aistats} required that, for any sequence of positive semi-definite (p.s.d.) matrices, $\A_t$, $\|\sum_{t \in \J^\alpha} \itt \A_t \itt{}' /\alpha\|$,  be bounded by a constant less than one. This is a much more stringent requirement; one way to satisfy it is using the moving object model on outlier supports assumed there.

\subsection{Main Lemmas}


For simplicity, we give the proof for the $\vt=0$ case. The changes with $\vt \neq 0$ are minor, see Appendix \ref{sec:proof}.

First consider the simpler case when $t_j$'s are known, i.e., consider Algorithm \ref{norst_basic}. In this case, $\that_j = t_j$.
\begin{definition} 
Define
\begin{enumerate}



\item the constants used in Theorem \ref{thm1}: $c_1=0.01$, $c_2 = 0.01$, and $C_1= 15\sqrt{\eta}$

\item  $s:= \outfraccol \cdot n$
\item $\phi^+ = 1.2$
\item bound on $\outfracrow^\alpha$: $b_0:= 0.01/f^2$.

\item  $q_0 := 1.2(\zz + \SE(\P_{j-1}, \P_j))$, $q_{k} = (0.3)^{k}q_0$

\item $\et := \xhatt - \xt$. Since $\vt=0$,  $\et = \lt - \lhatt$

\item Events: $\Gamma_{0,0}:= \{\text{assumed bound on } \SE(\Phat_0,\P_0)\}$,
\\ $\Gamma_{0, k}:= \Gamma_{0,k-1} \cap \{\SE(\Phat_{0,k},\P_0) \le 0.3^k \SE(\Phat_0,\P_0) \}$, 
\\  $\Gamma_{j, 0} := \Gamma_{j-1, K}$, $\Gamma_{j, k} := \Gamma_{j, k-1} \cap \{\SE(\Phat_{j, k}, \P_j) \le q_{k-1}/4 \}$ for $j=1,2,\dots, J$ and $k=1,2,\dots, K$.

\item Using the expression for $K$ given in the theorem, and since $\Phat_j = \Phat_{j,k}$ (from the Algorithm), it follows that $\Gamma_{j,K}$ implies $\SE(\Phat_j, \P_j)=\SE(\Phat_{j, K}, \P_j) \leq \zz$.
\end{enumerate}
\label{defs}
\end{definition}
Observe that, if we can show that $\Pr(\Gamma_{J,K} | \Gamma_{0, 0}) \geq 1 - dn^{-10}$ we will have obtained all the subspace recovery bounds of Theorem \ref{thm1}. The next two lemmas applied sequentially help show that this is true for Algorithm \ref{norst_basic} ($t_j$ known). The correctness of the actual algorithm (Algorithm \ref{algo:auto-reprocs-pca}) follows using these, Corollary \ref{cor:etbnds}, and Lemma \ref{lem:sschangedet}. The Theorem's proof is in Appendix \ref{sec:proof}.

\begin{lem}[first subspace update interval]\label{lem:reprocspcalemone}
Under the conditions of Theorem \ref{thm1}, conditioned on $\Gamma_{j, 0}$,
\begin{enumerate}
\item for all $t \in [\hat{t}_j, \hat{t}_j + \alpha)$,
$\norm{\bm\Psi \lt} \le (\zz +  \semax) \sqrt{\eta r \lambda^+} < \xmint/15$,
$\norm{\xhat_{t,cs} - \xt} \le 7 \xmint/15 < \xmint/2$,  $\That_t = \T_t$, and
the error $\et := \xhatt - \xt = \lt - \lhatt$ satisfies
\begin{align}\label{eq:etdef}
\et = \itt \left( \bpsi_{\Tt}{}'\bpsi_{\Tt} \right)^{-1} \itt{}' \bpsi \lt,
\end{align}
and $\norm{\et} \leq 1.2 (\zz  + \semax) \sqrt{\eta r \lambda^+}$.
\item w.p. at least $1 - 10n^{-10}$, the first subspace estimate $\Phat_{j,1}$ satisfies $\SE(\Phat_{j, 1}, \P_j) \leq (q_{0}/4)$, i.e., $\Gamma_{j, 1}$ holds.

\end{enumerate}
\end{lem}

\begin{lem}[$k$-th subspace update interval]\label{lem:reprocspcalemk}
Under the conditions of Theorem \ref{thm1}, conditioned on $\Gamma_{j, k-1}$,
\begin{enumerate}

\item for all $t \in [\hat{t}_j + (k-1)\alpha, \hat{t}_j + k\alpha - 1)$,  all claims of the first part of Lemma \ref{lem:reprocspcalemone} holds,
$\norm{\bm\Psi \lt} \le 0.3^{k-1} (\zz +  \semax) \sqrt{\eta r \lambda^+} $, and
$\norm{\et} \leq (0.3)^{k-1} \cdot 1.2 (\zz+  \semax) \sqrt{\eta r \lambda^+}$. 
\item w.p. at least $1 - 10n^{-10}$ the subspace estimate $\Phat_{j, k}$ satisfies $\SE(\Phat_{j, k}, \P_j) \leq (q_{k-1}/4)$, i.e., $\Gamma_{j, k}$ holds.

\end{enumerate}
\end{lem}

\begin{remark}
For the case of $j=0$, in both the lemmas above, $\semax$ gets replaced by $\SE(\Phat_0,\P_0)$ and $\zz$ by zero.
\end{remark}

\begin{corollary}\label{cor:etbnds}
Under the conditions of Theorem \ref{thm1} the following hold
\begin{enumerate}
\item For all $t \in[t_j, \that_j)$, conditioned on $\Gamma_{j-1, K}$, all claims of the first item of Lemma \ref{lem:reprocspcalemone} hold. 
\item For all $t \in[\that_j + K\alpha, t_{j+1})$, conditioned on $\Gamma_{j, K}$, the first item of Lemma \ref{lem:reprocspcalemk} holds with $k=K$. 
\end{enumerate}
Thus, for all $t$, by the above two claims and Lemmas \ref{lem:reprocspcalemone}, \ref{lem:reprocspcalemk}, under appropriate conditioning, $\et$ satisfies \eqref{eq:etdef}.
\end{corollary}

We prove these lemmas in the next few subsections. The projected CS proof (item one of both lemmas) uses the following lemma from \cite{rrpcp_perf} that relates the $s$-Restricted Isometry Constant (RIC), $\delta_s(.)$,  \cite{candes_rip} of a projection matrix to the incoherence of its orthogonal complement.
\begin{lem}\label{kappadelta} [{\cite{rrpcp_perf}}]
For an $n \times r$ basis matrix $\P$,
(1) $\delta_s(\I - \P \P')  = \max_{|\T| \le s} \|\I_\T{}' \P\|^2$; and
(2) $ \max_{|\T| \le s} \|\I_\T{}' \P\|^2 \le s \max_{i=1,2,\dots,n} \|\I_i{}' \P\|^2 \le s \mu r / n$.
\end{lem}
The last bound of the above lemma is a consequence of Definition \ref{defmu}. We apply this lemma with $s =  \outfraccol \cdot n$. The subspace update step proof (item 2 of both the above lemmas) uses a guarantee for PCA in sparse data-dependent noise, Theorem \ref{cor:pcasddn}, due to \cite{pca_dd}. Notice that $\et = \lt - \lhat_t$ is the noise/error seen by the subspace update step. By \eqref{eq:etdef}, this is sparse and depends on the true data $\lt$.

Consider the actual $t_j$ unknown case. The following lemma is used to show that, whp, we can detect subspace change within $2\alpha$ time instants. This lemmas assumes detection threshold $\lthres = 2\zz^2 \lambda^+$ (see Algorithm \ref{algo:auto-reprocs-pca}).%

\begin{lem}[Subspace Change Detection]\label{lem:sschangedet}
Consider an $\alpha$-length time interval $\J^{\alpha} \subset [t_j, t_{j+1}]$ (so that $\lt = \P_j \at$).
\ben
\item If $\bphi := \I - \Phat_{j-1} \Phat_{j-1}{}'$ and $\SE(\Phat_{j-1}, \P_{j-1}) \leq \zz$, with probability at least $1 - 10n^{-10}$,
\begin{align*}
\lambda_{\max}\left( \frac{1}{\alpha}\sum_{t \in \J^{\alpha}} \bphi \lhatt \lhatt{}' \bphi \right) &\geq 0.28 \lambda^- \SE^2(\P_{j-1}, \P_j) 
> \lthres
\end{align*}
\item If $\bphi := \I - \Phat_j \Phat_j{}'$ and $\SE(\Phat_j, \P_j) \leq \zz $, with probability at least $1 - 10n^{-10}$,
\begin{align*}
\lambda_{\max}\left( \frac{1}{\alpha}\sum_{t \in \J^{\alpha}} \bphi \lhatt \lhatt{}' \bphi \right) &\leq 1.37 \zz^2 \lambda^+ < \lthres
\end{align*}

\een
\end{lem}

\subsection{Proof of Lemma \ref{lem:reprocspcalemone}: projected CS and subspace update in the first update interval}
We first state a simple lemma. This is proved in Appendix \ref{app:concm}.
\begin{lem}\label{lem:sumprinang}
Let $\Aa$, $\Ba$ and $\Ca$ be $r$-dimensional subspaces in $\Re^n$ such that $\SE(\Aa, \Ba) \leq \Delta_1$ and $\SE(\Ba, \Ca) \leq \Delta_2$. Then, $\SE(\Aa, \Ca) \leq \Delta_1 + \Delta_2$.
\end{lem}

\begin{proof}[Proof of Lemma \ref{lem:reprocspcalemone}]
Recall from Definition \ref{defs} that $s:= \outfraccol \cdot n$ and $\phi^+ = 1.2$. Recall also that, for simplicity, we are considering the $\vt=0$ case.

\noindent {\em Proof of item 1. }
First consider $j>0$.
We have conditioned on the event $\Gamma_{j,0}:= \Gamma_{j-1,K}$. This implies that $\SE(\Phat_{j-1}, \P_{j-1}) \leq \zz$.

We consider the interval $t \in [\that_j, \that_j + \alpha)$. For this interval, $\Phat_{(t-1)} = \Phat_{j-1}$ and thus $\bpsi = \I - \Phat_{j-1} \Phat_{j-1}{}'$ (from Algorithm).
For the sparse recovery step, we first need to bound the $2s$-RIC of $\bpsi$. To do this, we first obtain bound on $\max_{|\T| \leq 2s} \|\I_{\T}{}' \Phat_{j-1}\|$ as follows. Consider any set $\T$ such that $|\T| \leq 2s$. Then,
\begin{align*}
\norm{\I_{\T}{}' \Phat_{j-1}} &\leq \norm{\I_{\T}{}'(\I - \P_{j-1} \P_{j-1}{}') \Phat_{j-1}} + \norm{\I_{\T}{}' \P_{j-1} \P_{j-1}{}' \Phat_{j-1}} \\
&\leq \SE(\P_{j-1}, \Phat_{j-1}) + \norm{\I_{\T}{}' \P_{j-1}} = \SE(\Phat_{j-1}, \P_{j-1}) + \norm{\I_{\T}{}' \P_{j-1}}
\end{align*}
Using Lemma \ref{kappadelta}, and the bound on $\outfraccol$ from Theorem \ref{thm1},
\begin{align}\label{eq:dense}
\max_{|\T| \leq 2s} \|\I_{\T}{}' \P_{j-1}\|^2 \leq 2s \max_i \|\I_i{}'\P_{j-1}\|^2 \leq \frac{2s\mu r}{n} \leq 0.01
\end{align}
Thus, using $\SE(\Phat_{j-1}, \P_{j-1}) \leq \zz$,
\begin{align*}
\max_{|\T| \leq 2s} \norm{\I_{\T}{}' \Phat_{j-1}} \leq \zz + \max_{|\T| \leq 2s}\norm{\I_{\T}{}' \P_{j-1}} \leq \zz + 0.1
\end{align*}

Finally, using Lemma \ref{kappadelta}, $\delta_{2s}(\bpsi) \leq 0.11^2 < 0.15$. Hence
\begin{align*}
\norm{\left(\bpsi_{\Tt}{}'\bpsi_{\Tt}\right)^{-1}} \leq \frac{1}{1 - \delta_s(\bpsi)} \leq \frac{1}{1 - \delta_{2s}(\bpsi)} \leq \frac{1}{1- 0.15} < 1.2= \phi^+.
\end{align*}
When $j=0$, there are some minor changes. From the initialization assumption, we have $\SE(\Phat_0, \P_0) \leq 0.25$.  Thus, $ \max_{|\T| \leq 2s} \norm{\I_{\T}{}' \Phat_{0}} \leq 0.25 + 0.1 = 0.35$. Thus, using Lemma \ref{kappadelta}, $\delta_{2s}(\bpsi_0) \leq 0.35^2 < 0.15$. The rest of the proof given below is the same for $j=0$ and $j>0$.

Next we bound norm of $\b_t:=\bpsi \lt$. This and the RIC bound will then be used to bound $\|\xhat_{t,cs}- \xt\|$.
\begin{align*}
\norm{\b_t} = \norm{\bpsi \lt} &= \norm{(\I - \Phat_{j-1} \Phat_{j-1}{}') \P_j \at} \leq \SE(\Phat_{j-1}, \P_j) \norm{\at} \\
&\overset{(a)}{\leq} (\zz + \SE(\P_{j-1}, \P_j))\sqrt{\eta r \lambda^+} := b_b 
\end{align*}
where $(a)$ follows from Lemma \ref{lem:sumprinang} with $\Aa = \Phat_{j-1}$, $\Ba = \P_{j-1}$ and $\Ca = \P_j$. Under the assumptions of Theorem \ref{thm1}, $b_b < x_{\min} / 15$. This is why we have set $\xi = x_{\min}/15$ in the Algorithm. Using these facts, and $\delta_{2s} (\bpsi) \leq 0.15$, the CS guarantee from \cite[Theorem 1.3]{candes_rip} implies that
\begin{align*}
\norm{\xhat_{t,cs} - \xt} &\leq  7 \xi = 7x_{\min}/15 < x_{\min}/2
\end{align*}
Consider support recovery. From above,
\begin{align*}
| (\shatcs - \st)_i | \leq \norm{\shatcs - \st} \leq 7x_{\min}/15 < x_{\min}/2
\end{align*}
The Algorithm sets $\omega_{supp} = \smin/2$. Consider an index $i \in \Tt$. Since $|(\st)_i| \geq \smin$,
\begin{align*}
\smin - |(\shatcs)_i| \le  |(\st)_i| - |(\shatcs )_i| \ \le | (\st - \shatcs )_i | < \frac{\smin}{2}
\end{align*}
Thus, $|(\shatcs)_i| > \frac{\smin}{2} = \omega_{supp}$ which means $i \in \Thatt$. Hence $\Tt \subseteq \Thatt$. Next, consider any $j \notin \Tt$. Then, $(\st)_j = 0$ and so
\begin{align*}
|(\shatcs)_j| = |(\shatcs)_j)| - |(\st)_j| &\leq |(\shatcs)_j -(\st)_j| \leq b_{b} < \frac{\smin}{2}
\end{align*}
which implies $j \notin \Thatt$ and $\Thatt \subseteq \Tt$ implying that $\Thatt = \Tt$.

Finally, we get an expression for $\et$ and bound it. With $\That_t = \Tt$ and since $\Tt$ is the support of $\xt$, $\xt = \I_{\Tt} \I_{\Tt}{}' \xt$, and so
\begin{align*}
\shatt = \bm{I}_{\Tt}\left(\bpsi_{\Tt}{}'\bpsi_{\Tt}\right)^{-1}\bpsi_{\Tt}{}'(\bpsi \lt + \bpsi \st) = \bm{I}_{\Tt}\left(\bpsi_{\Tt}{}'\bpsi_{\Tt}\right)^{-1} \I_{\Tt}{}' \bpsi\lt + \st
\end{align*}
since $\bpsi_{\Tt}{}' \bpsi = \I_{\Tt}' \bpsi' \bpsi = \I_{\Tt}{}' \bpsi$. Thus $\et = \shatt-\st$ satisfies
\begin{align*}
\et &= \bm{I}_{\Tt}\left(\bpsi_{\Tt}{}'\bpsi_{\Tt}\right)^{-1} \I_{\Tt}{}' \bpsi\lt  \text{ and} \\
\norm{\et} &\leq \norm{\left(\bpsi_{\Tt}{}'\bpsi_{\Tt}\right)^{-1}} \norm{\itt{}'\bpsi \lt} \le \phi^+  \norm{\itt{}'\bpsi \lt} \leq 1.2b_b
\end{align*}

\renewcommand{\bz}{b}
\emph{Proof of Item 2}:
We will use the following result from  \cite[Remark 4.18]{pca_dd}.
\begin{theorem}[PCA in sparse data-dependent noise (PCA-SDDN)] \label{cor:pcasddn}
We are given data vectors $\yt := \lt + \wt$ with $\wt= \I_{\T_t} \M_{s,t} \lt$, $t=1,2,\dots,\alpha$, where $\T_t$ is the support set of $\wt$, and $\lt= \P \at$ with $\at$ satisfying the assumptions of Theorem \ref{thm1}.
Pick an $\varepsilon_{\text{SE}}>0$.
Assume that $\max_t \|\M_{s,t} \P\|_2 \le q < 1$, the fraction of non-zeroes in any row of the noise matrix $[\w_1, \w_2, \dots, \w_\alpha]$ is bounded by $\bz$, and $3 \sqrt{\bz} q f \le 0.9 \varepsilon_{\text{SE}}/(1+\varepsilon_{\text{SE}})$. Define
\begin{align*}
\alpha_0 := C\eta \frac{q^2f^2}{\varepsilon_{\text{SE}}^2} r \log n
\end{align*}
For an $\alpha \ge \alpha_0$, let $\Phat$ be the matrix of top $r$ eigenvectors of $\D:=\frac{1}{\alpha} \sum_{t=1}^\alpha \yt \yt'$. With probability at least $1- 10n^{-10}$, $\SE(\Phat,\P) \le \varepsilon_{\text{SE}}$.
\end{theorem}

Since $\lhat_t = \lt - \et$ with $\et$ satisfying \eqref{eq:etdef}, updating $\Phat_{(t)}$ from the $\lhatt$'s is a problem of PCA in sparse data-dependent noise (SDDN), $\et$. To analyze this, we use Theorem \ref{cor:pcasddn} given above. Recall from Item 1 of this lemma that, for $t \in [\that_j, \that_j + \alpha)$, $\et$ satisfies \eqref{eq:etdef}. Recall from the Algorithm that we compute the first estimate of the $j$-th subspace, $\Phat_{j, 1}$, as the top $r$ eigenvectors of $\frac{1}{\alpha}\sum_{t = \that_j}^{\that_j + \alpha - 1} \lhat_t \lhat_t{}'$. In the notation of  Theorem \ref{cor:pcasddn}, $\yt \equiv \lhat_t$, $\wt \equiv \et$, $\lt \equiv \lt$, $\M_{s, t} = -\left( \bpsi_{\T_t}{}' \bpsi_{\T_t}\right)^{-1} \bpsi_{\T_t}{}'$, $\Phat = \Phat_{j,1}$, $\P = \P_j$, and so $\norm{\M_{s, t} \P} = \| \left( \bpsi_{\T_t}{}' \bpsi_{\T_t}\right)^{-1} \bpsi_{\T_t}{}' \P_j\| \leq 1.2 (\zz + \SE(\P_{j-1}, \P_j)) := q_0$. Also, $\bz \equiv b_0$ which is the upper bound on $\outfracrow^\alpha$ (see Definition \ref{defs}).
Applying Theorem \ref{cor:pcasddn} with $q \equiv q_0$, $\bz \equiv b_0$ and using $\varepsilon_{\text{SE}} = q_0/4$, observe that we require
\begin{align*}
\sqrt{b_0} q_0 f \leq \frac{0.9 (q_0/4)}{1 + (q_0/4)}.
\end{align*}
Since $q_0 =  1.2 (\zz + \SE(\P_{j-1}, \P_j)) < 1.2(0.01+0.8) <  0.98$ (follows from the bounds on $\zz$ and on $\Delta$ given in Theorem \ref{thm1}), this holds if $\sqrt{b_0}f \leq 0.18$. This is true since we have assumed $b_0 = 0.01/f^2$ (see Definition \ref{defs}). Thus, from Theorem \ref{cor:pcasddn}, with probability at least $1 - 10n^{-10}$, $\SE(\Phat_{j, 1}, \P_j) \leq q_0/4$. Thus, conditioned on $\Gamma_{j,0}$, with this probability, $\Gamma_{j,1}$ holds. 
\end{proof}

\renewcommand{\old}{\mathrm{old}}
\begin{remark}[Clarification about conditioning]
In the proof above we have used Theorem \ref{cor:pcasddn} which assumes that, for $t \in \J^\alpha$, the $\at$'s are mutually independent and the matrices $\M_{s,t}$ are either non-random or are independent of the $\at$'s for this interval. When we apply the theorem for our proof, we are conditioning on $\Gamma_{j,0}$. This does not cause any problem since the event $\Gamma_{j,0}$ is a function of the random variable $\y_\old:=\{\y_1,\y_2, \dots, \y_{\that_j-1}\}$ where as our summation is over $\J^\alpha:=[\that_j, \that_j+\alpha)$. Also, by Theorem \ref{thm1}, $\at$'s are independent of the outlier supports $\T_t$.

To be precise, we are applying Theorem \ref{cor:pcasddn} conditioned on $\y_\old$, for any $\y_\old \in \Gamma_{j,0}$. 
Even conditioned on $\y_\old$, clearly, the matrices $\M_{s,t}$ used above are independent of the $\at$'s for this interval. Also, even conditioned on $\y_\old$, the $\at$'s for $t \in [\that_j, \that_j+\alpha)$ are clearly mutually independent. Thus, the theorem can be applied. Its conclusion then tells us that, for any $\y_\old \in \Gamma_{j,0}$, conditioned on $\y_\old$, with probability at least $1 - 10n^{-10}$, $\SE(\Phat_{j, 1}, \P_j) \leq q_0/4$. Since this holds with the same probability for all $\y_\old \in  \Gamma_{j,0}$, it also holds with the same probability when we condition on $\Gamma_{j,0}$. Thus, conditioned on $\Gamma_{j,0}$, with this probability, $\Gamma_{j,1}$ holds.

An analogous argument will also apply to the following proofs.
\end{remark}

\subsection{Proof of Lemma \ref{lem:reprocspcalemk}: lemma for projected CS and subspace update in $k$-th update interval}
\begin{proof}[Proof of Lemma \ref{lem:reprocspcalemk}]
We first present the proof for the $k = 2$ case and then generalize it for an arbitrary $k$.

\textbf{(A) $k=2$: } We have conditioned on $\Gamma_{j,1}$. This implies that $\SE(\Phat_{j,1}, \P_j) \leq q_0/4$.

\emph{Proof of Item 1}:
We consider the interval $t \in [\that_j + \alpha,\that_j + 2\alpha)$. For this interval, $\Phat_{(t-1)} = \Phat_{j,1}$ and thus $\bpsi = \I - \Phat_{j,1} \Phat_{j,1}{}'$ (from Algorithm).
For the sparse recovery step, we need to bound the $2s$-RIC of $\bpsi$. 
Consider any set $\T$ such that $|\T| \leq 2s$. We have
\begin{align*}
\norm{\I_{\T}{}' \Phat_{j, 1}} &\leq \norm{\I_{\T}{}'(\I - \P_{j} \P_{j }{}') \Phat_{j, 1}} + \norm{\I_{\T}{}' \P_{j} \P_{j}{}' \Phat_{j, 1}} \\
&\leq \SE(\P_{j}, \Phat_{j, 1}) + \norm{\I_{\T}{}' \P_{j}} = \SE(\Phat_{j, 1}, \P_{j}) + \norm{\I_{\T}{}' \P_{j}}
\end{align*}
The equality holds since $\SE$ is symmetric for subspaces of the same dimension.
Using $\SE(\Phat_{j,1}, \P_j) \leq q_0/4$, \eqref{eq:dense},
\begin{align*}
\max_{|\T| \leq 2s} \norm{\I_{\T}{}' \Phat_{j, 1}} \leq q_0/4 + \max_{|\T| \leq 2s}\norm{\I_{\T}{}' \P_{j}} \leq q_0/4 + 0.1.
\end{align*}
Finally, from using the assumptions of Theorem \ref{thm1}, $q_0 \leq 0.96$. Using this and Lemma \ref{kappadelta},
\[
\delta_{2s}(\bpsi_j) = \max_{|\T| \leq 2s} \norm{\I_{\T}{}' \Phat_{j, 1}}^2 \leq 0.35^2 < 0.15.
\]
From, this
\begin{align*}
\norm{\left(\bpsi_{\Tt}{}'\bpsi_{\Tt}\right)^{-1}} \leq \frac{1}{1 - \delta_s(\bpsi)} \leq \frac{1}{1 - \delta_{2s}(\bpsi)} \leq \frac{1}{1- 0.15} < 1.2= \phi^+.
\end{align*}
Consider $\norm{\b_t}$.
\begin{align*}
\norm{\b_t} = \norm{\bpsi \lt} &= \norm{(\I - \Phat_{j, 1} \Phat_{j, 1}{}') \P_j \at} \leq \SE(\Phat_{j, 1}, \P_j) \norm{\at} \leq (q_0/ 4) \sqrt{\eta r \lambda^+} \\
&\overset{(a)}{\leq} 0.3(\zz + \SE(\P_{j-1}, \P_j)) \sqrt{\eta r \lambda^+} := 0.3b_b
\end{align*}
We have $0.3b_b < b_b < x_{\min} / 15$ as in the proof of Lemma \ref{lem:reprocspcalemone}. The rest of the proof is the same too. 
Notice here that, we could have loosened the required lower bound on $\xmint$ for this interval.

\emph{Proof of Item 2}: Again, updating $\Phat_{(t)}$ using $\lhatt$'s is a problem of PCA in sparse data-dependent noise (SDDN), $\et$. We use the result of Theorem \ref{cor:pcasddn}. Recall from Item 1 of this lemma that, for $t \in [\that_j + \alpha, \that_j + 2\alpha)$, $\et$ satisfies \eqref{eq:etdef}. We compute $\Phat_{j,2}$ as the top $r$ eigenvectors of $\frac{1}{\alpha}\sum_{t = \that_j + \alpha}^{\that_j + 2\alpha - 1} \lhat_t \lhat_t{}'$. In notation of Theorem  \ref{cor:pcasddn}, $\yt \equiv \lhat_t$, $\wt \equiv \et$, $\lt \equiv \lt$, $\P \equiv \P_j$, $\Phat \equiv \P_{j,2}$, and $\M_{s, t} = -\left( \bpsi_{\T_t}{}' \bpsi_{\T_t}\right)^{-1} \bpsi_{\T_t}{}'$. So $\norm{\M_{s, t} \P_j} = \| \left( \bpsi_{\T_t}{}' \bpsi_{\T_t}\right)^{-1} \bpsi_{\T_t}{}' \P_j\| \leq (\phi^+/4) q_0 := q_1$.
Applying  Theorem \ref{cor:pcasddn} with $q \equiv q_1$, $\bz \equiv b_0$ ($b_0$ bounds $\outfracrow^\alpha$), and setting $\varepsilon_{\text{SE}} = q_1/4$, observe that we require
\begin{align*}
\sqrt{b_0} q_1 f \leq \frac{0.9 (q_1/4)}{1 + (q_1/4)}
\end{align*}
which holds if $\sqrt{b_0}f \leq 0.18$. This is ensured since $b_0 = 0.01/f^2$.
Thus, from Theorem \ref{cor:pcasddn}, with probability at least $1 - 10n^{-10}$, $\SE(\Phat_{j, 2}, \P_j) \leq (q_1/4) = 0.25 \cdot 0.3 q_0$. Thus, with this probability, conditioned on $\Gamma_{j, 1}$, $\Gamma_{j,2}$ holds.

\textbf{(B) General $k$: } We have conditioned on $\Gamma_{j,k-1}$. This implies that $\SE(\Phat_{j,k-1}, \P_j) \leq q_{k-1}/4$.

\emph{Proof of Item 1}: Consider the interval $[\that_j + (k-1)\alpha, \that_j + k\alpha)$. In this interval, $\Phat_{(t-1)} = \Phat_{j,k-1}$ and thus $\bpsi = \I - \Phat_{j,k-1} \Phat_{j,k-1}{}'$.
Using the same idea as for the $k=2$ case, we have that for the $k$-th interval, $q_{k-1} = (\phi^+/4)^{k-1} q_0$. Pick $\varepsilon_{\text{SE}} = (q_{k-1}/4)$. From this it is easy to see that
\begin{align*}
\delta_{2s}(\bpsi) &\leq \left(\max_{|\T| \leq 2s} \norm{\I_\T{}' \Phat_{j, k-1}}\right)^2 \leq (\SE(\Phat_{j,k-1}, \P_j) + \max_{|\T| \leq 2s} \norm{\I_\T{}'\P_j})^2\\
& \overset{(a)}{\leq} (\SE(\Phat_{j, k-1}, \P_j) + 0.1)^2 \leq ((\phi^+/4)^{k-1} (\zz + \SE(\P_{j-1}, \P_j)+ 0.1)^2 < 0.15
\end{align*}
where $(a)$ follows from \eqref{eq:dense}. Using the approach Lemma \ref{lem:reprocspcalemone},
\begin{align*}
\norm{\bpsi \lt} &\leq \SE(\Phat_{j, k-1}, \P_j) \norm{\at} \leq (\phi^+/4)^{k-1} (\zz + \SE(\P_{j-1}, \P_j)) \sqrt{\eta r \lambda^+} \\
&\overset{(a)}{\leq} (\phi^+/4)^{k-1} (\zz  + \semax) \sqrt{\eta r  \lambda^+} := (\phi^+/4)^{k-1}b_b
\end{align*}

\emph{Proof of Item 2}: Again, updating $\Phat_{(t)}$ from $\lhatt$'s is a problem of PCA in sparse data-dependent noise given in Theorem \ref{cor:pcasddn}. From {\em Item 1} of this lemma that, for $t \in [\that_j + (k-1) \alpha, \that_j + k\alpha]$,  $\et$ satisfies \eqref{eq:etdef}. We update the subspace, $\Phat_{j, k}$ as the top $r$ eigenvectors of $\frac{1}{\alpha}\sum_{t = \that_j + (k-1) \alpha}^{\that_j + k\alpha - 1} \lhat_t \lhat_t{}'$. In the setting above $\yt \equiv \lhat_t$, $\wt \equiv \et$, $\lt \equiv \lt$, and $\M_{s, t} = -\left( \bpsi_{\T_t}{}' \bpsi_{\T_t}\right)^{-1} \bpsi_{\T_t}{}'$, and so $\norm{\M_{s, t} \P_j} = \| \left( \bpsi_{\T_t}{}' \bpsi_{\T_t}\right)^{-1} \bpsi_{\T_t}{}' \P_j\| \leq (\phi^+/4)^{k-1} q_0 := q_{k-1}$. Applying  Theorem \ref{cor:pcasddn} with  $q \equiv q_{k-1}$, $\bz \equiv b_0$ ($b_0$ bounds $\outfracrow^\alpha$), and setting $\varepsilon_{\text{SE}} = q_{k-1}/4$, we require%
\begin{align*}
\sqrt{b_0} q_{k - 1} f \leq \frac{0.9 (q_{k-1}/4)}{1 + (q_{k-1}/4)}
\end{align*}
which holds if $\sqrt{b_0}f \leq 0.12$. This is true by our assumption. Thus, from Theorem \ref{cor:pcasddn}, with probability at least $1 - 10n^{-10}$, $\SE(\Phat_{j, k}, \P_j) \leq (\phi^+/4)^{k-1} q_1$. Thus, with this probability, conditioned on $\Gamma_{j, k-1}$, $\Gamma_{j, k}$ holds.
\end{proof}

\subsection{Proof of Lemma \ref{lem:sschangedet}: subspace change detection lemma}

\begin{proof}[Proof of Lemma \ref{lem:sschangedet}]
The proof uses the following lemma. It is proved in Appendix \ref{app:concm}. The proof uses Cauchy-Schwartz for sums of matrices, followed by either matrix Bernstein \cite{tail_bound} or Vershynin's sub-Gaussian result \cite{vershynin}.%

\begin{lem}[Concentration Bounds]\label{lem:concm}
Assume that the assumptions of Theorem \ref{thm1} hold. For this lemma assume that $\lt = \P \at$, $\bphi := \I -\Phat \Phat{}'$, $\et = \mtt \mot \lt$, with $\|\frac{1}{\alpha}\sum_{t \in \J^\alpha} \mtt \mtt{}'\| \leq b_0$ and $\norm{\mot \P} \leq q$.
Assume that event $\ezero$ holds.
Then,
\begin{align*}
&\Pr\left(\norm{ \frac{1}{\alpha}\sum_t \at \at{}' - \Lam} \leq \epsilon_0 \lambda^-  \right) \geq 1 - 10n^{-10}\\
&\Pr\left(\norm{\frac{1}{\alpha}\sum_t \bphi \lt \et{}' \bphi} \leq (1 + \epsilon_1)\SE(\Phat, \P) \sqrt{b_0}q  \lambda^+  \right) \geq 1 - 10n^{-10} \\
&\Pr\left( \norm{\frac{1}{\alpha}\sum_t \bphi \et \et{}' \bphi} \leq (1 + \epsilon_2)\sqrt{b_0}q^2 \lambda^+  \right) \geq 1 - 10n^{-10} \\
&\Pr\left( \norm{\frac{1}{\alpha}\sum_t \bphi \lt \vt{}' \bphi} \leq \epsilon_{l,v} \lambda^-  \right) \geq 1 - 2n^{-10} \\
&\Pr\left( \norm{\frac{1}{\alpha}\sum_t \bphi \vt \vt{}' \bphi} \leq (\zz^2 + \epsilon_{v,v}) \lambda^- \right) \geq 1 - 2n^{-10}
\end{align*}
\end{lem}

\emph{Proof of Item (a)}: 
First from Corollary \ref{cor:etbnds}, note that for $t \in [t_j, \that_j]$, the error $\et$ satisfies \eqref{eq:etdef}.
We have
\begin{align}\label{eq:lmaxzero}
&\lambda_{\max}\left( \frac{1}{\alpha} \sum_{t \in \J^{\alpha}} \bphi [\P_j \at \at{}' \P_j{}'  + \et \et{}' + \lt \et{}' + \et \lt{}']\bphi  \right) \nonumber \\
&\overset{(a)}{\geq} \lambda_{\max} \left( \frac{1}{\alpha} \sum_{t \in \J^{\alpha}} \bphi \P_j \at \at{}' \P_j{}' \bphi \right) + \lambda_{\min}\left(\frac{1}{\alpha}\sum_{t \in \J^{\alpha}}\bphi[ \et \et{}' + \lt \et{}' + \et \lt{}'] \bphi \right) \nonumber \\
&\geq \lambda_{\max} \left( \frac{1}{\alpha} \sum_{t \in \J^{\alpha}} \bphi \P_j \at \at{}' \P_j{}' \bphi \right) - \norm{\frac{1}{\alpha} \sum_{t \in \J^{\alpha}} \bphi \et \et{}' \bphi } - 2\norm{\frac{1}{\alpha} \sum_{t \in \J^{\alpha}} \bphi \lt \et{}' \bphi} \nonumber \\
&:= \lambda_{\max}(\bm{T})  - \norm{\frac{1}{\alpha} \sum_{t \in \J^{\alpha}} \bphi \et \et{}' \bphi } - 2\norm{\frac{1}{\alpha} \sum_{t \in \J^{\alpha}} \bphi \lt \et{}' \bphi}
\end{align}
where $(a)$ follows from Weyl's Inequality.
Now we bound the second and third terms by invoking Lemma \ref{lem:concm} with $\ezero := \{\SE(\Phat_{j-1}, \P_{j-1}) \leq \zz\}$, $\Phat \equiv \Phat_{j-1}$, $\P \equiv \P_j$,
$\mtt \equiv \itt$ and $\mot \equiv \left(\bpsi_{\T_t} {}' \bpsi_{\T_t} \right)^{-1} \bpsi_{\T_t}{}'$, where $\bpsi = \I - \Phat_{j-1} \Phat_{j-1}{}'$. Thus,  $b_0 \equiv b_0$, $q \equiv q_0$. Thus, with probability at least $1 - 4n^{-10}$,
\begin{align}\label{eq:lmaxone}
\norm{\frac{1}{\alpha} \sum_{t \in \J^{\alpha}} \bphi \et \et{}' \bphi } + 2\norm{\frac{1}{\alpha} \sum_{t \in \J^{\alpha}} \bphi \lt \et{}' \bphi} &\leq \sqrt{b_0} q_0^2 \lambda^+(1 + \epsilon_2) + 2 \sqrt{b_0} q_0(\SE(\P_{j-1}, \P_j) + \zz) \lambda^+(1 + \epsilon_1)
\end{align}
The above equation uses the fact that $\SE(\Phat_{j-1}, \P_j) \leq \zz + \SE(\P_{j-1}, \P_j)$ which is a direct consequence of using Lemma \ref{lem:sumprinang}. We bound the first term of \eqref{eq:lmaxzero} as follows. Let $\bphi \P_j \overset{QR}{=} \bm{E}_j \bm{R}_j$ be its reduced QR decomposition. Thus $\bm{E}_j$ is an $n \times r$ matrix with orthonormal columns and $\bm{R}_j$ is an $r \times r$ upper triangular matrix. Let
\begin{align*}
\bm{A} := \bm{R}_j\left( \frac{1}{\alpha} \sum_{t \in \J^{\alpha}} \at \at{}' \right) \bm{R}_j{}'.
\end{align*}
Observe that $\bm{T}$ can also be written as
\begin{align}\label{eq:decomp_T}
\bm{T} = \begin{bmatrix}
\bm{E}_j & \bm{E}_{j, \perp}
\end{bmatrix}
\diagmat{\bm{A}}{\bm{0}}
\begin{bmatrix}
\bm{E}_j{}' \\ \bm{E}_{j, \perp}{}'
\end{bmatrix}
\end{align}
and thus $\lambda_{\max}(\bm{A}) = \lambda_{\max}(\bm{T})$. We work with $\lambda_{\max}(\bm{A})$ in the sequel. We will use the following simple claim.
\begin{claim}\label{claim:psd}
If $\bm{X} \succeq 0$ (i.e., $\bm{X}$ is a p.s.d matrix), where $\bm{X} \in \Re^{r \times r}$, then $\bm{R} \bm{X} \bm{R}{}' \succeq 0$ for all $\bm{R} \in \Re^{r \times r}$.
\end{claim}
\begin{proof}
Since $\bm{X}$ is p.s.d., $\bm{y}{}' \bm{X} \bm{y} \geq 0$ for any vector $\bm{y}$. Use this with $\bm{y} = \bm{R}{}' \bm{z}$ for any $\bm{z} \in \Re^{r}$. We get $\bm{z}{}' \bm{R} \bm{X} \bm{R}{}' \bm{z} \geq 0$. Since this holds for all $\bm{z}$, $\bm{R}\bm{X}\bm{R}{}' \succeq 0$.
\end{proof}
\noindent Using Lemma \ref{lem:concm}, it follows that
\begin{align*}
\Pr\left( \frac{1}{\alpha} \sum_t \at \at{}' - (\lambda^- - \epsilon_0) \bm{I} \succeq 0 \right) \geq 1 - 2n^{-10}.
\end{align*}
Using Claim \ref{claim:psd}, with probability $ 1- 2n^{-10}$,
\begin{align*}
&\bm{R}_j \left( \frac{1}{\alpha}\sum_t \at \at{}' - (\lambda^- - \epsilon_0) \bm{I} \right) \bm{R}_j{}' \succeq 0 \\
& \implies \lambda_{\min}\left(\bm{R}_j \left( \frac{1}{\alpha}\sum_t \at \at{}' - (\lambda^- - \epsilon_0) \bm{I} \right) \bm{R}_j{}'\right) \geq 0
\end{align*}
Using Weyl's inequality \cite{hornjohnson}, with the same probability,
\begin{align*}
\lambda_{\min}\left(\bm{R}_j \left( \frac{1}{\alpha}\sum_t \at \at{}' - (\lambda^- - \epsilon_0) \bm{I} \right) \bm{R}_j{}'\right) \leq \lambda_{\max}\left(\bm{R}_j \left( \frac{1}{\alpha}\sum_t \at \at{}'\right) \bm{R}_j{}'\right) -(\lambda^- - \epsilon_0)\lambda_{\max}\left(\bm{R}_j\bm{R}_j{}'\right)
\end{align*}
and so,
\begin{align*}
\lambda_{\max}(\bm{A}) \geq (\lambda^- - \epsilon_0)\lambda_{\max}(\bm{R}_j\bm{R}_j{}').
\end{align*}
We now obtain a lower bound on the second term in the rhs above.
\begin{align}\label{eq:nine}
\lambda_{\max}(\bm{R}_j \bm{R}_j{}') &= \lambda_{\max}(\P_j{}'(\I - \Phat_{j-1}\Phat_{j-1}{}')\P_j)) = \lambda_{\max}(\P_j{}'(\I - \P_{j-1}\P_{j-1}{}' + \P_{j-1}\P_{j-1}{}' - \Phat_{j-1}\Phat_{j-1}{}')\P_j)) \nonumber \\
&\geq \lambda_{\max}(\P_j{}' (\I - \P_{j-1}\P_{j-1}{}')\P_j) + \lambda_{\min}(\P_j{}'(\P_{j-1}\P_{j-1}{}' - \Phat_{j-1}\Phat_{j-1}{}')\P_j) \nonumber \\
&= \sigma_{\max}^2((\I - \P_{j-1}\P_{j-1}{}')\P_j) + \lambda_{\min}(\P_j{}'(\P_{j-1}\P_{j-1}{}' - \Phat_{j-1}\Phat_{j-1}{}')\P_j) \nonumber \\
&\geq \SE^2(\P_{j-1}, \P_j) -\norm{\P_j{}'(\P_{j-1}\P_{j-1}{}' - \Phat_{j-1}\Phat_{j-1}{}')\P_j} \nonumber \\
&\overset{(a)}{\geq} \SE^2(\P_{j-1}, \P_j) -\norm{\P_{j-1}\P_{j-1}{}' - \Phat_{j-1}\Phat_{j-1}{}'} \geq \SE^2(\P_{j-1}, \P_j) - 2\zz
\end{align}
where we have used \cite[Lemma 2.10]{rrpcp_perf}. 
Thus, combining \eqref{eq:lmaxzero}, \eqref{eq:lmaxone}, \eqref{eq:nine}, and using $\epsilon_0 = 0.01$,  $\epsilon_1 = \epsilon_2 = 0.01$, with probability at least $1 - 10n^{-10}$,
\begin{align*}
\lambda_{\max} \left( \frac{1}{\alpha} \sum_{t \in \J^{\alpha}} \bphi \lhatt \lhatt{}' \bphi \right) \geq &\ 0.99 \lambda^-(\SE^2(\P_{j-1}, \P_j) - 2\zz) - \lambda^+[\sqrt{b_0} q_0 (1.01 q_0 + 2.02 (\zz + \SE(\P_{j-1}, \P_j)))] \\
&\overset{(a)}{\geq} \lambda^- \left[0.91 \SE^2(\P_{j-1}, \P_j) - 2.7 \sqrt{b_0} f (\zz + \SE(\P_{j-1}, \P_j)^2\right]   \\
&\overset{(b)}{\geq} \lambda^-\left[ 0.91 \SE^2(\P_{j-1}, \P_j) - 0.54(\zz^2 + \SE^2(\P_{j-1}, \P_j)\right]  \\
&\geq \lambda^-\SE^2(\P_{j-1}, \P_j)(0.91 - 0.54 \cdot 1.16) \geq 0.28 \lambda^- \SE^2(\P_{j-1}, \P_j)
\end{align*}
where $(a)$ uses $q_0 = 1.2(\zz + \SE(\P_{j-1}, \P_j))$ and $\zz \leq 0.03 \SE^2(\P_{j-1}, \P_j)/f^2 < 0.4\SE^2(\P_{j-1}, \P_j)$, $(b)$ uses $\sqrt{b_0} f = 0.1$ and $(a+b)^2 \leq 2(a^2 + b^2)$. and the last inequality again uses $\zz \leq 0.03 \SE^2(\P_{j-1}, \P_j)/f^2$.

\noindent \emph{Proof of Item (b)}: First, we recall that from Corollary \ref{cor:etbnds}, for $t \in [\that_j + K \alpha, t_{j+1})$, the error $\et$ satisfies \eqref{eq:dynmc}. 
\begin{align*}
\lambda_{\max}\left( \frac{1}{\alpha} \sum_{t \in \J^{\alpha}} \bphi \lhatt \lhatt{}'\bphi \right) &\leq \lambda_{\max} \left(  \frac{1}{\alpha} \sum_{t \in \J^{\alpha}} \bphi \lt \lt{}' \bphi \right) + \lambda_{\max}\left( \frac{1}{\alpha} \sum_{t \in \J^{\alpha}} \bphi [\lt \et{}' + \et \lt{}' + \et \et{}'] \bphi \right) \\
&\leq \lambda_{\max}\left( \frac{1}{\alpha} \sum_{t \in \J^{\alpha}} \bphi \lt \lt{}' \bphi \right) + \norm{\frac{1}{\alpha} \sum_{t \in \J^{\alpha}} \bphi \et \et{}' \bphi} + 2\norm{\frac{1}{\alpha} \sum_{t \in \J^{\alpha}} \bphi \lt \et{}' \bphi} \\
& := \lambda_{\max}(\bm{T}) + \norm{\frac{1}{\alpha} \sum_{t \in \J^{\alpha}} \bphi \et \et{}' \bphi} + 2\norm{\frac{1}{\alpha} \sum_{t \in \J^{\alpha}} \bphi \lt \et{}' \bphi}
\end{align*}
To obtain bounds on the second and third terms in the equation above we invoked Lemma \ref{lem:concm} with $\ezero := \{\SE(\Phat_j, \P_j) \leq \zz\}$, $\Phat \equiv \Phat_j$, $\P \equiv \P_j$, $\mtt \equiv \itt$, $\mot \equiv \left(\bpsi_{\T_t} {}' \bpsi_{\T_t} \right)^{-1} \bpsi_{\T_t}{}'$, where $\bpsi = \I - \Phat_{j} \Phat_{j}{}'$ and $b_0 \equiv b_0$, $ \equiv q_K$, . Thus, with probability at least $1 - 10n^{-10}$,
\begin{align*}
\norm{\frac{1}{\alpha} \sum_{t \in \J^{\alpha}} \bphi \et \et{}' \bphi} + 2\norm{\frac{1}{\alpha} \sum_{t \in \J^{\alpha}} \bphi \lt \et{}' \bphi}  &\leq \sqrt{b_0} q_K \lambda^+ (q_K(1 + \epsilon_2) + 2(1 + \epsilon_1)\zz)
\end{align*}
The above equation also uses $\SE(\Phat_{j}, \P_j) \leq \zz$. Proceeding as before to bound $\lambda_{\max}(\bm{T})$, define $\bphi \P_j \overset{QR}{=} \bm{E}_j \bm{R}_j$, define $\bm{A}$ as before, we know $\lambda_{\max}(\bm{T}) = \lambda_{\max}(\bm{E}_j{}' \bm{T} \bm{E}_j) = \lambda_{\max}(\bm{A})$. Further,
\begin{align*}
\lambda_{\max}(\bm{A}) &= \lambda_{\max}\left( \bm{R}_j \left( \frac{1}{\alpha} \sum_{t \in \J^{\alpha}} \at \at{}' \right) \bm{R}_j{}' \right) \overset{(a)}{\leq} \lambda_{\max}\left( \frac{1}{\alpha} \sum_{t \in \J^{\alpha}} \at \at{} \right) \lambda_{\max}(\bm{R}_j \bm{R}_j{}')
\end{align*}
where $(a)$ uses Ostrowski's theorem \cite[Theorem 5.4.9]{hornjohnson}.
We have
\begin{align*}
\lambda_{\max}(\bm{R}_j \bm{R}_j{}') = \sigma_{\max}^2(\bm{R}_j) = \sigma_{\max}^2((\I - \Phat_{j}\Phat_j{}')\P_j) \leq \zz^2
\end{align*}
and we can bound $\lambda_{\max}(\frac{1}{\alpha} \sum_{t \in \J^{\alpha}} \at \at{}')$ using the first item of Lemma \ref{lem:concm} with $\epsilon_0=0.01$.
Combining all of the above, and  setting
$\epsilon_1 = \epsilon_2 = 0.01$, when the subspace has not changed, with probability at least $1 - 10n^{-10}$,
\begin{align*}
\lambda_{\max}\left( \frac{1}{\alpha} \sum_{t \in \J^{\alpha}} \bphi \lhatt \lhatt{}' \bphi \right) \leq \lambda^+ [1.01 \zz^2 + \sqrt{b_0} q_K (1.01q_K + 2.01\zz)] \overset{(a)}{\leq} 1.37 \zz^2 \lambda^+
\end{align*}
where $(a)$ uses $q_K \leq \zz$ and $b_0 f^2 = 0.01$. Under the condition of Theorem \ref{thm1}, recall that $\lthres = 2\zz^2\lambda^+$ and thus, with high probability, $1.37 \zz^2 \lambda^+ < \lthres < 0.28 \lambda^- \SE^2(\P_{j-1}, \P_j)$.
\end{proof}

\section{Extensions}\label{sec:extensions}
\subsection{Static Robust PCA}
A useful corollary of our result for RST is that NORST is also the {first online algorithm that provides a provable finite sample guarantee for the static Robust PCA problem}. 
Static RPCA is our problem setting with $J=1$, or in other words, with $\lt = \P \at$.
A recent work, \cite{xu_nips2013_1}, developed an online stochastic optimization based reformulation of PCP, called ORPCA, to solve this. Their paper provides only a partial guarantee because the guarantee assumes that the intermediate algorithm estimates, $\Phat_{(t)}$, are full-rank. Moreover the guarantee is only asymptotic.  Instead our result given below is a complete guarantee and is non-asymptotic.

\begin{corollary}[Static RPCA]\label{thm:rpca}
Consider Algorithm \ref{norst_basic} with $t_2=\infty$. Theorem \ref{thm1} holds with following modification: replace the slow subspace change assumption with a fixed subspace $\P$. Everything else remains as is, but with $r \equiv \rmat$. Under the assumptions of Theorem \ref{thm1}, all the conclusions hold with same probability. The time and memory complexity  are $O(n d r \log(1/\epsilon))$ and $O(n r \log n \log(1/\epsilon))$.
\end{corollary}

In applications such as ``robust'' dimensionality reduction \cite{dimreduce_app1, dimreduce_app2} where the objective is to just obtain the top-$r$ directions along which the variability of data is maximized, we only need the first $K\alpha = Cf^2r\log n \log(1/\zz)$ samples to obtain an $\zz$-accurate subspace estimate. If only these are used, the time complexity reduces to $O(n K \alpha r \log(1/\zz)) = O(n r^2 \log n \log^2(1/\epsilon))$. This is faster than even NO-RMC \cite{rmc_gd} and does not require $\tmax \approx n$; of course it requires the other extra assumptions discussed earlier.


\subsection{Subspace Tracking with Missing Data (ST-missing) and Dynamic Matrix Completion}
Another useful corollary of our result is a guarantee for the ST-missing problem.
%
Consider the subspace tracking with missing data (ST-missing) problem. By setting the missing entries at time $t$ to zero, and by defining $\T_t$ to be the set of missing entries at time $t$, we observe $n$-dimensional vectors that satisfy
\bea
\yt := \lt - \itt \itt{}' \lt, \text{ for } t = 1, 2, \dots, \tmax.
\label{eq:dynmc}
\eea
with $\xt \equiv \itt \itt{}' \lt$. This can be interpreted as a special case of the RST problem where the set $\T_t$ is known. Because there are no sparse corruptions (outliers), there is no $\xmint$. Thus the initialization error need not be $O(1/\sqrt{r})$ (needed in the RST result to ensure a reasonable lower bound on $\xmint$) and so one can even use random initialization. We assume that the initialization is obtained using the Random Orthogonal Model described in \cite{matcomp_candes}. As explained in \cite{matcomp_candes}, a basis matrix generated from this model is already $\mu$-incoherent. We have the following corollary. The only change in its proof is the proof of the first subspace update interval for the $j=0$ case.



\begin{corollary}[ST-missing]\label{cor:stmiss}
Consider NORST-Random (Algorithm \ref{algo:auto-dyn-mc}). If the assumptions of Theorem \ref{thm1} on $\at$ and $\vt$ hold, $\P_j$'s are $\mu$-incoherent,  $t_{j+1}-t_j > (K+2)\alpha$, $\Delta < 0.8$, the outlier fraction bounds given in Theorem \ref{thm1} hold, and if, for $t \in [t_0,t_1]$, $\outfraccol \le c/ (\log n)$, then all conclusions of Theorem \ref{thm1} on $\Phat_{(t)}$ and on $\lt$ hold with the same probability.
\end{corollary}
To our knowledge, the above is the first complete non-asymptotic guarantee for ST-missing; and the first result that allows changing subspaces. All existing guarantees are either partial guarantees (make assumptions on intermediate algorithm estimates), e.g., \cite{local_conv_grouse}, or provide only asymptotic results \cite{past,petrels}.
Moreover, from a dynamic matrix completion viewpoint, it is also giving a matrix completion solution without assuming that the set of observed entries is generated from a uniform or a Bernoulli model. Of course the tradeoff is that it needs many more observed entries. All these points will be discussed in detail in follow-up work where we will also numerically evaluate NORST-random for this problem.


\begin{algorithm}[t!]
\caption{NORST-Random for subspace tracking with missing data (ST-missing)}
\label{algo:auto-dyn-mc}
Algorithm \ref{norst} with the following changes
\ben
\item Replace line 3 with: compute $\Phat_0 \leftarrow \Phat_\init \leftarrow \text{Generate an $n \times r$ basis matrix from the random orthogonal model}$;  $\tildej~\leftarrow~1$, $k~\leftarrow~1$

\item Replace line 6 with the following
\bi
\item $\bpsi \leftarrow \bm{I} - \hat{\pt}_{(t-1)}\hat{\pt}_{(t-1)}{}'$; $\tty_t \leftarrow \bpsi \yt$;
 $\xhat_t \leftarrow \I_{\T_t} ( \bpsi_{\T_t}{}' \bpsi_{\T_t} )^{-1} \bpsi_{\T_t}{}'\tty_t$;  $\hat{\bm{\ell}}_t \leftarrow \yt - \hat{\bm{x}}_t$.
\ei
\een
\end{algorithm}



\subsection{Fewer than $r$ directions change}
It is possible to relax the lower bound on outlier magnitudes if not all of the subspace directions change at a given subspace change time. Suppose that only $r_\ch < r$ directions change. When $r_\ch=1$, we recover the guarantee of \cite{rrpcp_dynrpca} but for NORST (which is a simpler algorithm than s-reprocs).

Let $\P_{j-1, \fx}$ denote a basis for the fixed component of $\P_{j-1}$ and let $\P_{j-1,\ch}$ denote a basis for its changing component. Thus,
$\P_{j-1} \R = [ \P_{j-1,\fx}, \P_{j-1,\ch}]$, where $\bm{R}$ is a $r \times r$ rotation matrix. 
We have
\begin{align}\label{eq:rchdef1}
\P_j = [\P_{j-1,\fx}, \P_{j,\chd}]
\end{align}
where $\P_{j, \chd}$ is the changed component and has the same dimension as $\P_{j-1, \ch}$. Thus,
\begin{align}\label{eq:rchdef2}
\SE(\P_{j-1},\P_j) = \SE( \P_{j-1,\ch}, \P_{j,\chd})
\end{align}
and so $\Delta = \max_j \SE(\P_{j-1},\P_j) = \max_j \SE( \P_{j-1,\ch}, \P_{j,\chd})$.
Let $\lambda_{\ch}^+$ denote the largest eigenvalue along any direction in $\P_{j,\chd}$.
\begin{corollary}\label{cor:rch}
In Algorithm \ref{algo:auto-reprocs-pca}, replace line 17 by $\Phat_{(t)} \leftarrow \basis(\Phat_{j-1}, \Phat_{j, k})$. For basis matrices $\P_1, \P_2$, we use $\P = \basis(\P_1, \P_2)$ to mean that $\P$ is a basis matrix with column span equal to the column span of $[\P_1, \P_2]$.
Assume that  \eqref{eq:rchdef1} and \eqref{eq:rchdef2} hold. Also assume that the conditions of Theorem \ref{thm1} holds with the lower bound on $\xmint$ relaxed to $\xmint \ge C (\zz \sqrt{\eta (r - r_\ch) \lambda^+} + (\zz + \Delta) \sqrt{\eta r_\ch \lambda_\ch^+})$.
Then, all conclusions of Theorem \ref{thm1} hold.
\end{corollary}

\pgfplotstableread[col sep = comma]{figures/final_files/MO_online_finalSE_TIT.dat}\ltmodata
\pgfplotstableread[col sep = comma]{figures/final_files/bern_online_finalSE_TIT.dat}\ltmodatabern

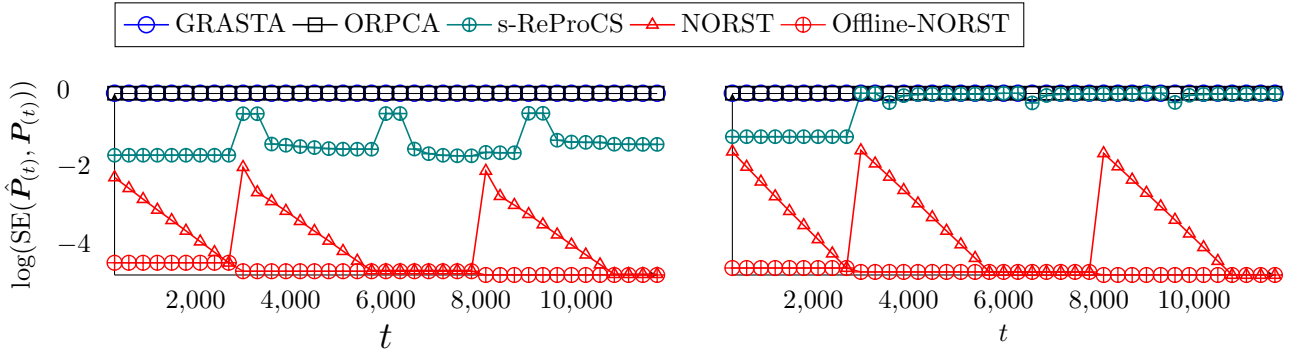
\begin{figure*}[t!]
\centering
\begin{tikzpicture}
    \begin{groupplot}[
        group style={
            group size=2 by 1,
            y descriptions at=edge left,
        },
        my stylecompare,
        enlargelimits=false,
        width = .5\linewidth,
        height=4cm,
    ]
       \nextgroupplot[
            my legend style compare,
            legend style={at={(0,1.5)}},
            legend columns = 7,
            xlabel=$t$,
            ylabel={\small{$\log(\SE(\Phat_{(t)}, \P_{(t)}))$}},
        ]
	        \addplot table[x index = {0}, y index = {1}]{\ltmodata};
	        \addplot table[x index = {0}, y index = {2}]{\ltmodata};
	        \addplot table[x index = {0}, y index = {3}]{\ltmodata};
	        \addplot table[x index = {0}, y index = {4}]{\ltmodata};
	        \addplot table[x index = {0}, y index = {5}]{\ltmodata};
	               \nextgroupplot[
            ymode=log,
            xlabel={\small{$t$}},
        ]
	        \addplot table[x index = {0}, y index = {1}]{\ltmodatabern};
	        \addplot table[x index = {0}, y index = {2}]{\ltmodatabern};
	        \addplot table[x index = {0}, y index = {3}]{\ltmodatabern};
	        \addplot table[x index = {0}, y index = {4}]{\ltmodatabern};
	        \addplot table[x index = {0}, y index = {5}]{\ltmodatabern};
    \end{groupplot}
\end{tikzpicture}
\caption{Left plot illustrates the $\lt$ error for outlier supports generated using Moving Object Model and right plot illustrates the error under the Bernoulli model. The values are plotted every $k\alpha - 1$ time-frames.}
\vspace{-0.1in}
\label{fig:Comparison}
\end{figure*}

\section{Empirical Evaluation}\label{sec:sims}
In this section we present the results for extensive numerical experiments on synthetic and real data to validate our theoretical claims. All time comparisons are performed on a Desktop Computer with Intel$^{\textsuperscript{\textregistered}}$ Xeon E$3$-$1240$ $8$-core CPU @ $3.50$GHz and $32$GB RAM and all synthetic data experiments are averaged over $100$ independent trials. The codes are available at \url{https://github.com/praneethmurthy/NORST}.

\subsection{Synthetic Data}
We perform three experiments on synthetic data to corroborate our theoretical claims.

\subsubsection{Experiment 1}
We compare the results of NORST and Offline-NORST with static RPCA algorithms, and Robust Subspace Tracking/Online RPCA methods proposed in literature. For our first experiment, we generate the changing subspaces using $\P_j = e^{\gamma_j \bm{B}_j} \P_{j-1}$ as done in \cite{grass_undersampled} where $\gamma_j$ controls the subspace change and $\bm{B}_j$'s are skew-symmetric matrices. In the first experiment we used the following parameters. $n = 1000$, $d = 12000$, $J = 2$, $t_1 = 3000$, $t_2 = 8000$, $r = 30$, $\gamma_1 = 0.001$, $\gamma_2 = \gamma_1$ and the matrices $\bm{B}_{1}$ and $\bm{B}_2$ are generated as $\bm{B}_{1} = (\tilde{\bm{B}}_1 - \tilde{\bm{B}}_1)$ and $\bm{B}_2 = (\tilde{\bm{B}}_2 - \tilde{\bm{B}}_2)$ where the entries of $\tilde{\bm{B}}_1, \tilde{\bm{B}}_2$ are generated independently from a standard normal distribution. We set $\alpha= 300$. This gives us the basis matrices $\bm{P}_{(t)}$ for all $t$. To obtain the low-rank matrix $\bm{L}$ from this we generate the coefficients $\at \in \mathbb{R}^{r}$ as independent zero-mean, bounded random variables. They are $(\at)_i \overset{i.i.d}{\sim} unif[-q_i, q_i]$ where $q_i = \sqrt{f} - \sqrt{f}(i-1)/2r$ for $i = 1, 2, \cdots, r - 1$ and $q_{r} = 1$. thus the condition number is $f$ and we selected $f=50$. For the sparse supports, we considered two models according to which the supports are generated. First we use Model G.24 \cite{rrpcp_dynrpca} which simulates a moving object pacing in the video. For the first $t_{\train} = 100$ frames,  we used a smaller fraction of outliers with parameters $s/n = 0.01$, $b_0 = 0.01$. For $t > t_\train$ we used $s/n = 0.05$ and $b_0 = 0.3$. Secondly, we used the Bernoulli model to simulate sampling uniformly at random, i.e., each entry of the matrix, is independently selected with probability $\rho$ or not selected with probability $1- \rho$. We generate the sparse supports using the Bernoulli model using $\rho = 0.01$ for the first $t_\train$ frames and $\rho = 0.3$ for the subsequent frames. The sparse outlier magnitudes for both support models are generated uniformly at random from the interval $[x_{\min}, x_{\max}]$ with $x_{\min} = 10$ and $x_{\max} = 20$.

We initialized the s-ReProCS and NORST algorithms using AltProj applied to $\Y_{[1,t_\train]}$ with $t_\train=100$. For the parameters to AltProj we used used the true value of $r$, $15$ iterations and a threshold of $0.01$. This, and the choice of $\gamma_1$ and $\gamma_2$ ensure that $\SE(\Phat_\init, \P_0) \approx \SE(\P_1, \P_0) \approx \SE(\P_2, \P_1) \approx 0.01$. The other algorithm parameters are set as mentioned in the theorem, i.e., $K = \lceil \log(c/\varepsilon) \rceil = 8$, $\alpha = C r \log n = 300$, $\omega = x_{\min}/2 = 5$ and $\xi = 7 x_{\min}/15 = 0.67$, $\lthres = 2 \varepsilon^2 \lambda^+ = 7.5 \times 10^{-4}$. For $l_1$ minimization we used the \texttt{YALL-1} toolbox \cite{yall1} and set the tolerance to $10^{-4}$. For the least-squares step we use the Conjugate Gradient Least Squares instead of the well-known ``backslash'' operator in \texttt{MATLAB} since this is a well conditioned problem. For this we set the tolerance as $10^{-10}$ and the number of iterations as $10$. We have not done any code optimization such as use of \texttt{MEX} files for various sub-routines to speed up our algorithm. For the other online methods we implement the algorithms without modifications. The regularization parameter for ORPCA was set as with $\lambda_1 = 1 / \sqrt{n}$ and $\lambda_2 = 1 / \sqrt{d}$ according to \cite{xu_nips2013_1}. Wherever possible we set the tolerance as $10^{-6}$ and $100$ iterations to match that of our algorithm. As shown in Fig. \ref{fig:Comparison}, NORST is significantly better than all the RST methods - s-ReProCS \cite{rrpcp_dynrpca}, and two popular heuristics from literature - ORPCA \cite{xu_nips2013_1} and GRASTA \cite{grass_undersampled}.

\begin{table}[t!]
\begin{center}
\caption{Comparison of $\|\Lhat - \L\|_F/\|\L\|_F$ for Online and offline RPCA methods. Average time for the Moving Object model is given in parentheses. The offline (batch) methods are performed once on the complete dataset.}
\resizebox{0.9\linewidth}{!}{
\begin{tabular}{cccccccc} \toprule
Outlier Model & GRASTA & s-ReProCS & ORPCA & {\bf NORST} & RPCA-GD & AltProj & \bf{Offline-NORST} \\
  & ($0.02$ ms) & ($0.9$ ms) & ($1.2$ms) & ($\bm{0.9}$ {\bf ms}) & ($7.8$ms) & ($4.6$ms) & ($\bm{1.7}$\bf{ms}) \\ \toprule
Moving Object & $0.630$ & $0.598$ & $0.861$ & $\bm{4.23 \times 10^{-4}}$ & $4.283$ & $4.441$ & $\bm{8.2 \times 10^{-6}}$ \\
Bernoulli & $6.163$ & $2.805$ & $1.072$ & $\bm{0.002}$ & $0.092$& $0.072$ & $\bm{2.3 \times 10^{-4}}$ \\ \bottomrule
\end{tabular}
\label{tab:offline}
}
\end{center}
\end{table}

We also provide a comparison of offline techniques in Table \ref{tab:offline}. We must mention here that we implemented the static RPCA methods once on the entire data matrix, $\Y$. We do this to provide a roughly equal comparison of the time taken. In principle, we could also implement the static techniques on disjoint batches of size $\alpha$, but we observed that this did not yield significant improvement in terms of reconstruction accuracy, while being considerably slower, and thus we report only the latter setting. As can be seen, offline NORST outperforms all static RPCA methods, both for the moving object outlier support model and for the commonly used random Bernoulli support model. For the batch comparison we used PCP, AltProj and  RPCA-GD. We set the regularization parameter for PCP $1/\sqrt{n}$ in accordance with \cite{rpca}. The other known parameters, $r$ for Alt-Proj, outlier-fraction for RPCA-GD, are set using the true values. Furthermore, for all algorithms (the IALM solver in case of PCP) we set the threshold as $10^{-6}$ and the number of iterations to $100$ as opposed to $10^{-3}$ and $50$ iterations which were set as default to provide a fair comparison with NORST and Offline-NORST. All results are averaged over $100$ independent runs.

\begin{figure}[t!]
\begin{center}
\resizebox{\linewidth}{!}{%
\begin{tikzpicture}
    \begin{groupplot}[
        group style={
            group size=2 by 1,
            horizontal sep=1.2cm
        },
        width = .5\linewidth,
        height = 4cm
    	]
        \nextgroupplot[
		        view={0}{90},
				xlabel={$b_0$},
               	ylabel={$r$},
                colormap/blackwhite,
                title={\small{(a) Phase Transition for Alt Proj}},
				x label style={at={(axis description cs:0.5,-0.1)},anchor=north},
				y label style={at={(axis description cs:-0.1,0.4)},anchor=west},
        ]
        \node [text width=1em,anchor=north west] at (rel axis cs: 0,1)
                {\subcaption{\label{fig:phasetransvsnbnd}}};

                	\addplot3[surf] file {figures/final_files/PhaseTransAltProj.dat};
                			
		            \nextgroupplot[
		       view={0}{90},
               xlabel={$b_0$},
               ylabel={$r$},
               colormap/blackwhite,
               title={\small{(b) Phase Transition for NORST}},
               scaled y ticks=false, tick label style={/pgf/number format/fixed},
               x label style={at={(axis description cs:0.5,-0.1)},anchor=north},
               y label style={at={(axis description cs:-0.1,0.4)},anchor=west},
        ]
    \node [text width=1em,anchor=north west] at (rel axis cs: 0,1)
                {\subcaption{\label{fig:phasetransvsngausr}}};

        	\addplot3[surf] file {figures/final_files/PhaseTransNORST.dat};
        	
        \end{groupplot}
\end{tikzpicture}%
}
\end{center}%
\caption{Numerically computed $\|\Lhat - \L\|_F^2 / \|\L\|_F^2$ for AltProj and for Offline NORST. Note that NORST indeed has a much higher tolerance to outlier fraction per row as compared to AltProj. Black denotes $0$ and white denotes $1$.}
\vspace{-0.5cm}
\label{fig:phase_trans}
\end{figure}
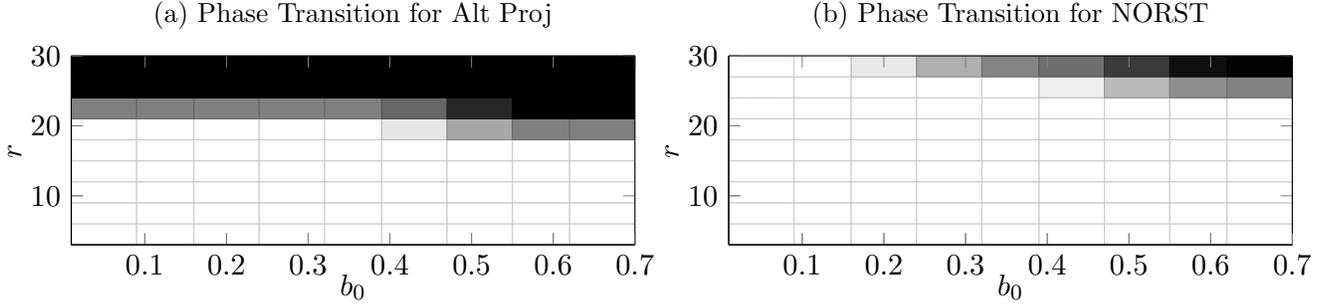

\subsubsection{Experiment 2}
 Next we perform an experiment to validate our claim of NORST admitting a higher fraction of outliers per row than the state of the art. In particular, since AltProj has the highest tolerance to $\outfracrow$, we only compare with it. The experiment proceeded as follows. We chose $10$ different values of each of $r$ and $b_0$ (we slightly misuse notation here to let $b_0 := \outfracrow$ for this section only). For each pair of $b_0$ and $r$ we implemented NORST and ALtProj over $100$ independent trials and computed the relative error in recovering $\L$, i.e., we computed $\|\Lhat - \L\|_F / \|L\|_F$ for each run. We computed the empirical probability of success, i.e., we enumerated the number of times out of $100$ the error seen by each algorithm was less than a threshold, $0.5$.

For each pair of $\{b_0, r\}$ we used the Bernoulli model for sparse support generation, the low rank matrix is generated exactly as done in the previous experiments with the exception that again to provide an equal footing, we increased the ``subspace change'' by setting $\gamma_1$ and $\gamma_2$ to $10$ times the value that was used in the previous experiment. For the first $t_\train$ frames we used $b_0 = 0.02$. We provide the phase transition plots for both algorithm in Fig. \ref{fig:Comparison}. Here, white represents success while black represents failure. As can be seen, NORST is able to tolerate a much larger fraction of outlier-per-row as compared to AltProj.

\pgfplotstableread[col sep = comma]{figures/final_files/xmin_variation.dat}\xmindata
\pgfplotsset{every axis title/.append style={at={(.5,1.15)}}}
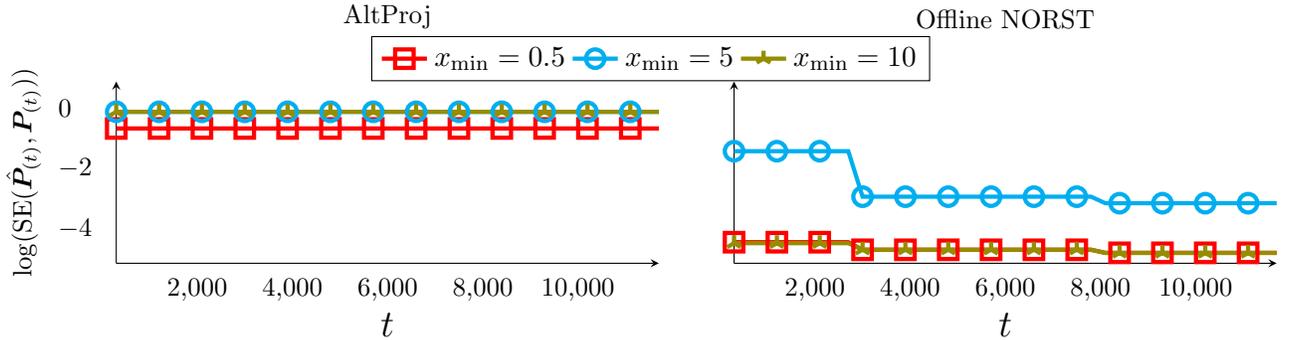
\begin{figure*}[t!]
\centering
\begin{tikzpicture}
    \begin{groupplot}[
        group style={
            group size=2 by 1,
            y descriptions at=edge left,
        },
        my stylecompare,
        enlargelimits=false,
        width = .5\linewidth,
        height=4cm,
        ymin=-5, ymax = 1
    ]
       \nextgroupplot[
            legend entries={
					$x_{\min} = 0.5$,
            		$x_{\min} = 5$,
            		$x_{\min} = 10$,	
            		},
            legend style={at={(1.5,1.25)}},
            legend columns = 3,
            xlabel=$t$,
            ylabel={\small{$\log(\SE(\Phat_{(t)}, \P_{(t)}))$}},
            title={\small AltProj},
        ]

        	        \addplot[red, line width=1.6pt, mark=square,mark size=3.5pt, mark repeat=3] table[x index = {0}, y index = {7}]{\xmindata};
	        \addplot[cyan, line width=1.6pt, mark=o,mark size=3.5pt, mark repeat=3] table[x index = {0}, y index = {8}]{\xmindata};
	        \addplot[olive, line width=1.6pt, mark=Mercedes star,mark size=3pt, mark repeat=3] table[x index = {0}, y index = {9}]{\xmindata};

	               \nextgroupplot[
            xlabel=$t$,
            title={\small Offline NORST}
        ]
	   	        \addplot[red, line width=1.6pt, mark=square,mark size=3.5pt, mark repeat=3] table[x index = {0}, y index = {4}]{\xmindata};
	        \addplot[cyan, line width=1.6pt, mark=o,mark size=3.5pt, mark repeat=3] table[x index = {0}, y index = {5}]{\xmindata};
	        \addplot[olive, line width=1.6pt, mark=Mercedes star,mark size=3pt, mark repeat=3] table[x index = {0}, y index = {6}]{\xmindata};

    \end{groupplot}
\end{tikzpicture}
\caption{In the above plots we show the variation of the subspace errors for varying $x_{\min}$. In particular, we set all the non-zero outlier values to $x_{\min}$. 
The results are averaged over $100$ independent trials.}
\vspace{-0.1in}
\label{fig:xmin_var}
\end{figure*}

\subsubsection{Experiment 3}
Finally we perform an experiment to analyze the effect of the lower bound on the outlier magnitude $x_{\min}$ with the performance of NORST and AltProj. We show the results in Fig. \ref{fig:xmin_var}. In the first stage, we generate the data exactly as done in the Moving Object sparse support model of the first experiment. The only change to the data generation parameters is that we now choose three different values of $x_{\min} = \{0.5, 5, 10\}$. Furthermore, we set all the non-zero entries of the sparse matrix to be equal to $x_{\min}$. This is actually harder than allowing the sparse outliers to take on any value since for a moderately low value of $x_{\min}$ the outlier-lower magnitude bound of Theorem \ref{thm1} is violated. This is indeed confirmed by the numerical results presented in Fig. \ref{fig:xmin_var}. (i) When $x_{\min} = 0.5$, NORST works well since now all the outliers get classified as the small unstructured noise $\vt$. (ii) When $x_{\min} = 10$, NORST still works well because now $\xmint$ is large enough so that the outlier support is mostly correctly recovered. (iii) But when $x_{\min} = 5$ the NORST reconstruction error stagnates around $10^{-3}$.

All AltProj errors are much worse than those of NORST because the outlier fraction per row is the same as in the first experiment. What can be noticed though is that the variation with varying $\xmint$ is not that significant. 

\begin{figure}[t!]
\begin{center}
\resizebox{.7\linewidth}{!}{
\begin{tabular}{@{}c@{}c@{}c@{}c@{}c@{}c@{}}
\\    \newline
	\includegraphics[scale=.2]{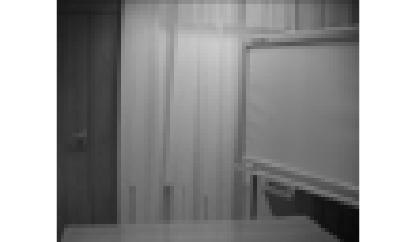}
&
	\includegraphics[scale=.2]{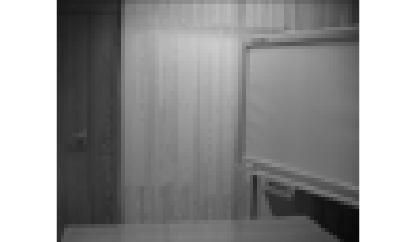}
&
	\includegraphics[scale=.2]{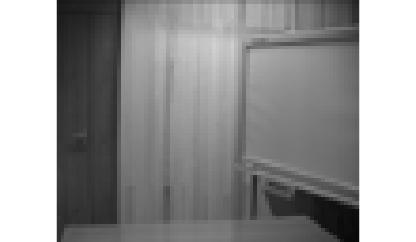}
&
	\includegraphics[scale=.2]{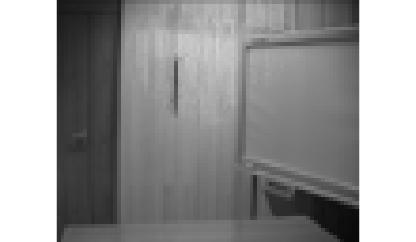}
&
	\includegraphics[scale=.2]{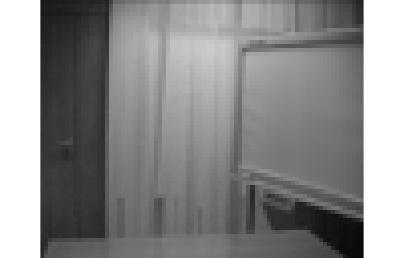}
&
	\includegraphics[scale=.2]{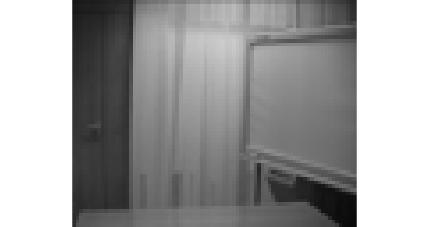}
\\    \newline
	\includegraphics[scale=.2]{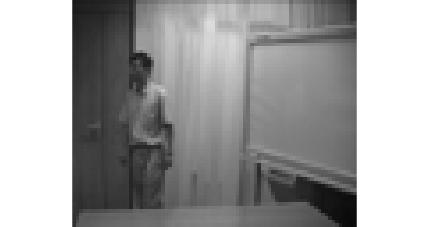}
&
	\includegraphics[scale=.2]{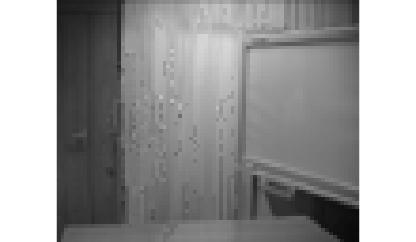}
&
	\includegraphics[scale=.2]{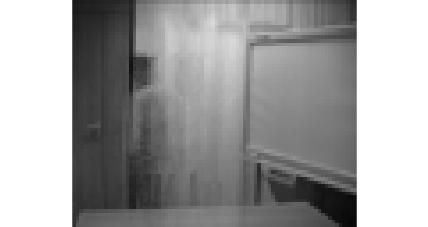}
&
	\includegraphics[scale=.2]{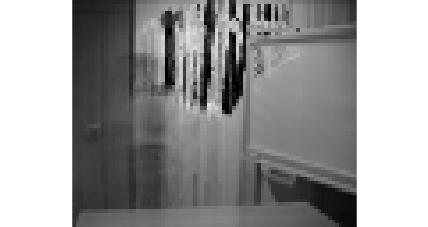}
&
	\includegraphics[scale=.2]{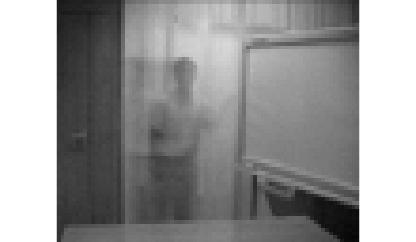}
&
	\includegraphics[scale=.2]{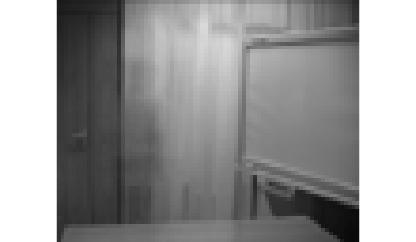}
\\    \newline
	\includegraphics[scale=.2]{figures/orig_MR1030.jpg}
&
	\includegraphics[scale=.2]{figures/reprocs_pca_bg_curtain_theory630.jpg}
&
	\includegraphics[scale=.2]{figures/ncrpca_BG_MR1030.jpg}
&
	\includegraphics[scale=.2]{figures/gd_BG_MR1030.jpg}
&
	\includegraphics[scale=.2]{figures/grasta_BG_MR1030.jpg}
&
	\includegraphics[scale=.2]{figures/rpca_BG_MR1030.jpg}
\\    \newline
	\subcaptionbox{original\label{1}}{\includegraphics[scale=.2]{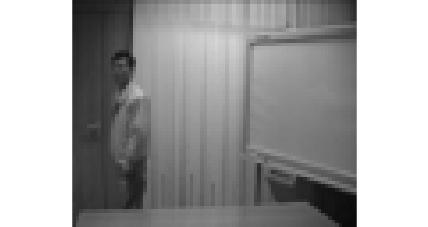}}
&
	\subcaptionbox{NORST\label{1}}{\includegraphics[scale=.2]{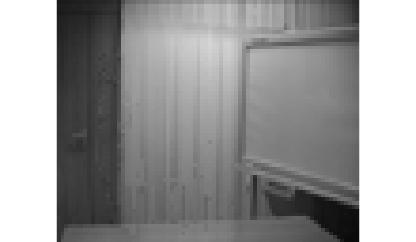}}
&
	\subcaptionbox{AltProj\label{1}}{\includegraphics[scale=.2]{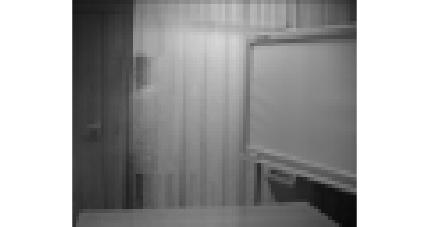}}
&
	\subcaptionbox{RPCA-GD\label{1}}{\includegraphics[scale=.2]{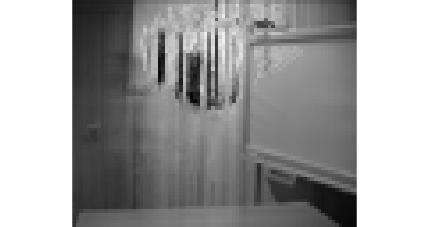}}
&
	\subcaptionbox{GRASTA\label{1}}{\includegraphics[scale=.2]{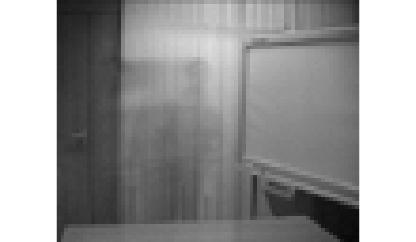}}
&
	\subcaptionbox{PCP\label{1}}{\includegraphics[scale=.2]{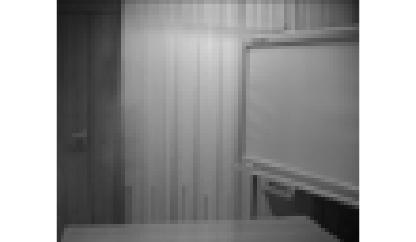}}
\end{tabular}
}

\caption{Comparison of visual performance in Foreground Background separation for the MR dataset. The images are shown at $t = t_{\text{train}} + 10, 140, 630, 760$.}
\label{fig:mr_full}
\end{center}

\end{figure}

\subsection{Real Data}
In this section we provide empirical results on real video for the task of Background Subtraction. 
For the AltProj algorithm we set $r = 40$. The remaining parameters were used with default setting.
For NORST, we set $\alpha = 60$, $K = 3$, $\xi_t = \|\bm{\Psi} \hat{\bm{\ell}}_{t-1}\|_2$. We found that these parameters work for most videos that we verified our algorithm on. For RPCA-GD we set the ``corruption fraction'' $\alpha = 0.2$ as described in their paper.

\begin{figure}[t!]
\begin{center}
\resizebox{.75\linewidth}{!}{
\begin{tabular}{@{}c@{}c@{}c@{}c@{}c@{}c@{}}
\\    \newline
	\includegraphics[scale=.15]{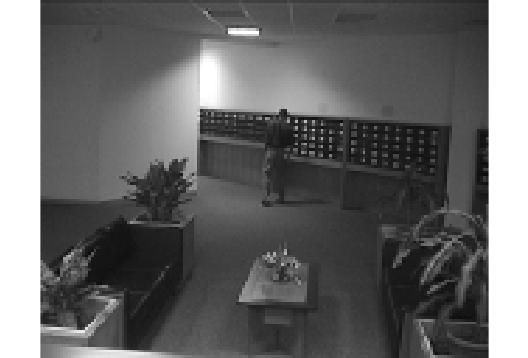}
&
	\includegraphics[scale=.15]{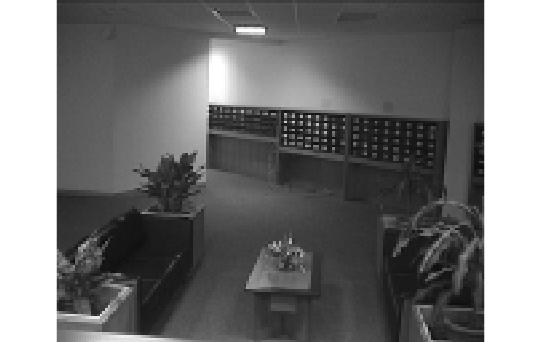}
&
	\includegraphics[scale=.15]{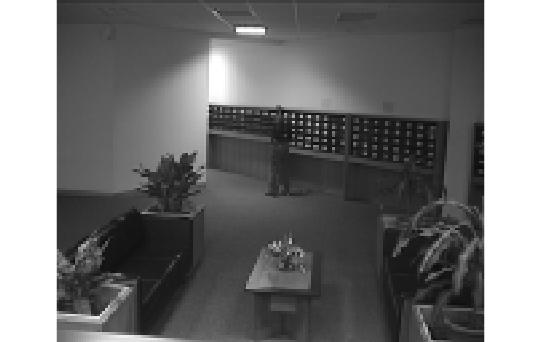}
&
	\includegraphics[scale=.15]{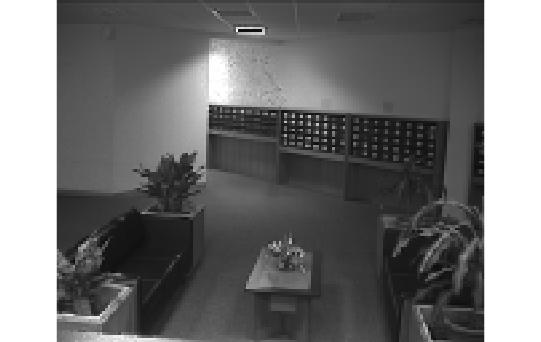}
&
	\includegraphics[scale=.15]{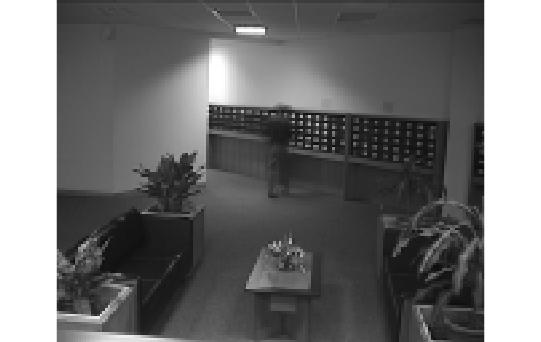}
&
	\includegraphics[scale=.15]{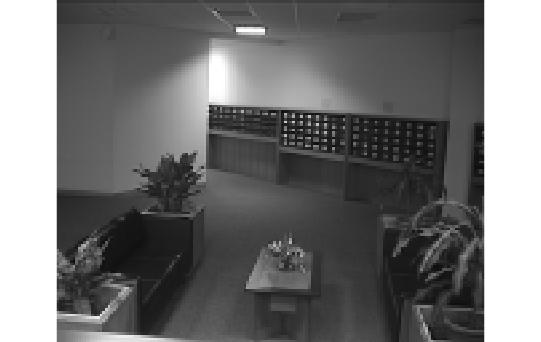}
\\    \newline
	\includegraphics[scale=.15]{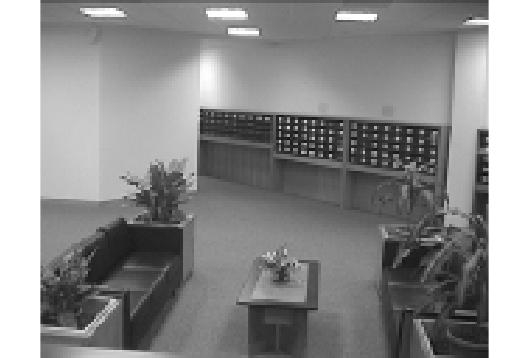}
&
	\includegraphics[scale=.15]{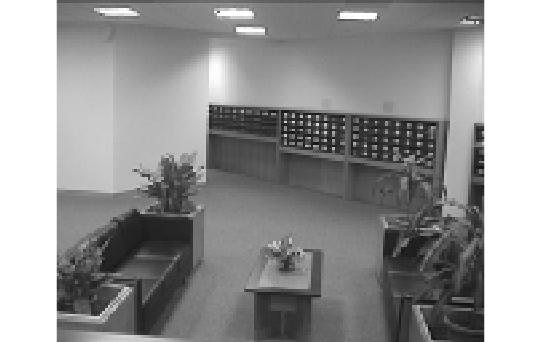}
&
	\includegraphics[scale=.15]{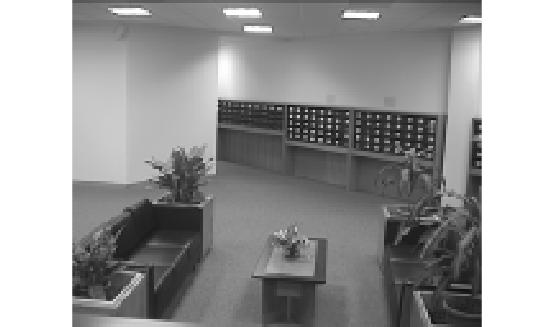}
&
	\includegraphics[scale=.15]{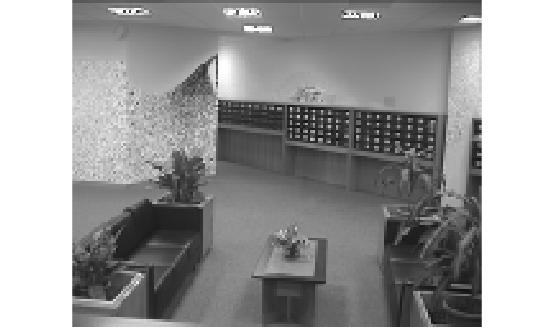}
&
	\includegraphics[scale=.15]{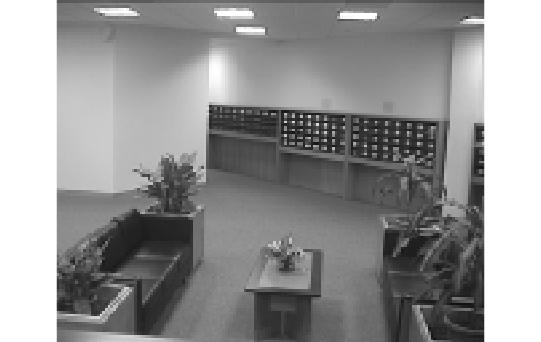}
&
	\includegraphics[scale=.15]{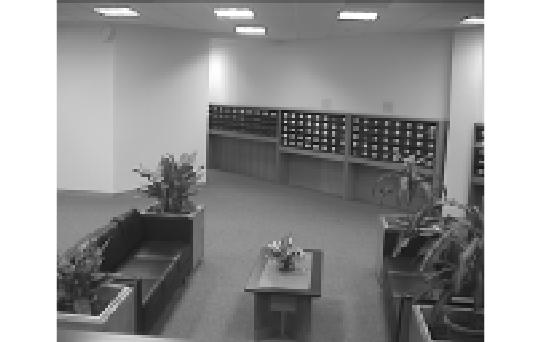}
\\    \newline
	\subcaptionbox{original\label{1}}{\includegraphics[scale=.15]{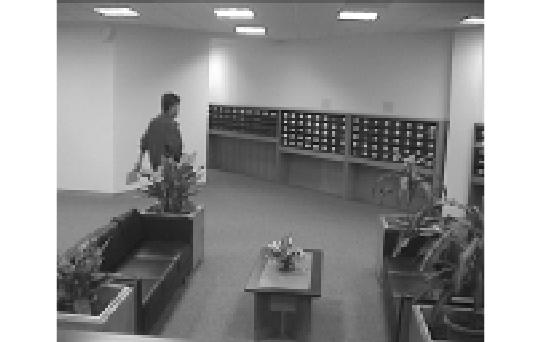}}
&
	\subcaptionbox{NORST\label{1}}{\includegraphics[scale=.15]{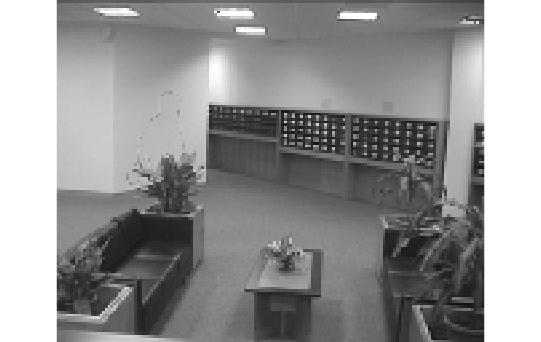}}
&
	\subcaptionbox{AltProj\label{1}}{\includegraphics[scale=.15]{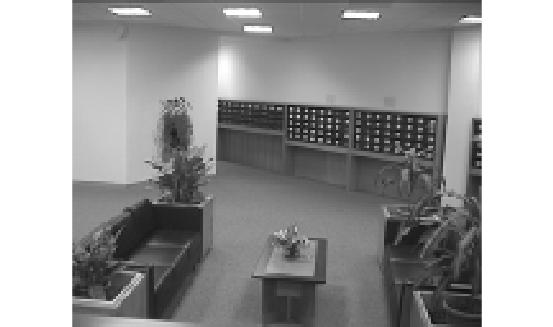}}
&
	\subcaptionbox{RPCA-GD\label{1}}{\includegraphics[scale=.15]{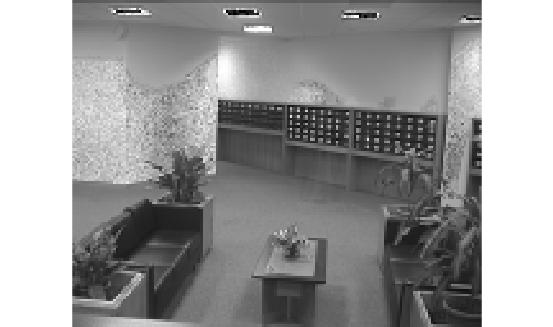}}
&
	\subcaptionbox{GRASTA \label{1}}{\includegraphics[scale=.15]{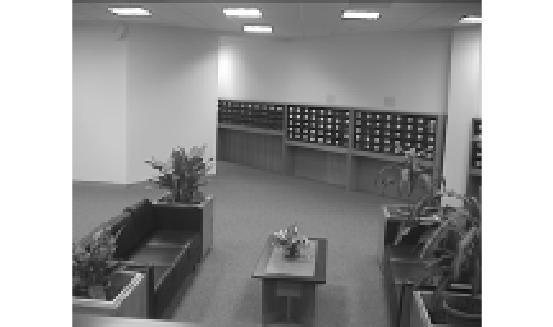}}
&
	\subcaptionbox{PCP\label{1}}{\includegraphics[scale=.15]{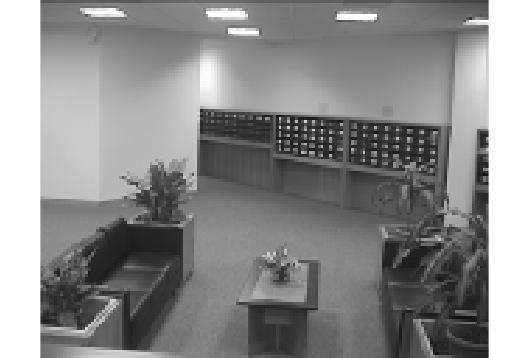}}
\end{tabular}
}
\caption{Comparison of visual performance in Foreground Background separation for the Lobby (LB) dataset. The recovered background images are shown at $t = t_\train + 260, 545, 610$.}
\label{fig:lb_full}
\end{center}
\end{figure}

{\em Meeting Room (MR) dataset}: The meeting room sequence is set of $1964$ images of resolution $64 \times 80$. The first $1755$ frames consists of outlier-free data. Henceforth, we consider only the last $1209$ frames. For NORST, we used $t_\train=400$. In the first $400$ frames, a person wearing a black shirt walks in, writes something on the board and goes back. In the subsequent frames, the person walks in with a white shirt. This is a challenging video sequence because the color of the person and the color of the curtain are hard to distinguish. NORST is able to perform the separation at around $43$ frames per second. We present the results in Fig. \ref{fig:mr_full}

{\em Lobby (LB) dataset}: This dataset contains $1555$ images of resolution $128 \times 160$. The first $341$ frames are outlier free. Here we use the first $400$ ``noisy'' frames as training data. The Alt Proj algorithm is used to obtain an initial estimate with rank, $r = 40$. The parameters used in all algorithms are exactly the same as above. NORST achieves a ``test'' processing rate of $16$ frames-per-second. We present the results in Fig. \ref{fig:lb_full}

\section*{Acknowledgments}
The authors would like to thank Praneeth Netrapalli and Prateek Jain of Microsoft Research India for fruitful discussions on strengthening the guarantee by removing assumptions on subspace change model.

\bibliographystyle{IEEEbib}
\bibliography{../../bib/tipnewpfmt_kfcsfullpap}

\begin{thebibliography}{10}

\bibitem{rrpcp_icmltemp}
P.~Narayanamurthy and N.~Vaswani,
\newblock ``Nearly optimal robust subspace tracking,''
\newblock {\em ICML}, pp. 3698--3706, 2018.

\bibitem{rrpcp_merop}
P.~Narayanamurthy and N.~Vaswani,
\newblock ``{A Fast and Memory-Efficient Algorithm for Robust PCA (MERoP)},''
\newblock in {\em IEEE Intl. Conf. Acoustics, Speech, Sig. Proc. (ICASSP)},
  2018.

\bibitem{rpca}
E.~J. Cand{\`e}s, X.~Li, Y.~Ma, and J.~Wright,
\newblock ``Robust principal component analysis?,''
\newblock {\em J. ACM}, vol. 58, no. 3, 2011.

\bibitem{cv_app}
J.P. Costeira and T.~Kanade,
\newblock ``A multibody factorization method for independently moving
  objects,''
\newblock {\em International Journal of Computer Vision}, vol. 29, no. 3, pp.
  159--179, 1998.

\bibitem{sigproc_app}
G.~S. Wagner and T.~J. Owens,
\newblock ``Signal detection using multi-channel seismic data,''
\newblock {\em Bulletin of the Seismological Society of America}, vol. 86, no.
  1A, pp. 221--231, 1996.

\bibitem{selin_reprocs}
A.~Ozdemir, E.~M. Bernat, and S.~Aviyente,
\newblock ``Recursive tensor subspace tracking for dynamic brain network
  analysis,''
\newblock {\em IEEE Transactions on Signal and Information Processing over
  Networks}, 2017.

\bibitem{rpca2}
V.~Chandrasekaran, S.~Sanghavi, P.~A. Parrilo, and A.~S. Willsky,
\newblock ``Rank-sparsity incoherence for matrix decomposition,''
\newblock {\em SIAM Journal on Optimization}, vol. 21, 2011.

\bibitem{rpca_zhang}
D.~Hsu, S.~M. Kakade, and T.~Zhang,
\newblock ``Robust matrix decomposition with sparse corruptions,''
\newblock {\em IEEE Trans. Info. Th.}, Nov. 2011.

\bibitem{robpca_nonconvex}
P.~Netrapalli, U~N Niranjan, S.~Sanghavi, A.~Anandkumar, and P.~Jain,
\newblock ``Non-convex robust pca,''
\newblock in {\em NIPS}, 2014.

\bibitem{rpca_gd}
X.~Yi, D.~Park, Y.~Chen, and C.~Caramanis,
\newblock ``Fast algorithms for robust pca via gradient descent,''
\newblock in {\em NIPS}, 2016.

\bibitem{rmc_gd}
Y.~Cherapanamjeri, K.~Gupta, and P.~Jain,
\newblock ``Nearly-optimal robust matrix completion,''
\newblock {\em ICML}, 2016.

\bibitem{rrpcp_perf}
C.~Qiu, N.~Vaswani, B.~Lois, and L.~Hogben,
\newblock ``Recursive robust pca or recursive sparse recovery in large but
  structured noise,''
\newblock {\em IEEE Trans. Info. Th.}, pp. 5007--5039, August 2014.

\bibitem{rrpcp_isit15}
B.~Lois and N.~Vaswani,
\newblock ``Online matrix completion and online robust pca,''
\newblock in {\em IEEE Intl. Symp. Info. Th. (ISIT)}, 2015.

\bibitem{rrpcp_aistats}
J.~Zhan, B.~Lois, H.~Guo, and N.~Vaswani,
\newblock ``{Online (and Offline) Robust PCA: Novel Algorithms and Performance
  Guarantees},''
\newblock in {\em Intnl. Conf. Artif. Intell. Stat. (AISTATS)}, 2016.

\bibitem{rrpcp_dynrpca}
P.~Narayanamurthy and N.~Vaswani,
\newblock ``Provable dynamic robust pca or robust subspace tracking,''
\newblock {\em arXiv:1705.08948 (submitted to IEEE Trans. Info. Theory)}, 2017.

\bibitem{past}
B.~Yang,
\newblock ``Projection approximation subspace tracking,''
\newblock {\em IEEE Trans. Sig. Proc.}, pp. 95--107, 1995.

\bibitem{adaptivesigproc_book}
T.~Adali and S.~Haykin, Eds.,
\newblock {\em Adaptive Signal Processing: Next Generation Solutions},
\newblock Wiley \& Sons, 2010.

\bibitem{petrels}
Y.~Chi, Y.~C. Eldar, and R.~Calderbank,
\newblock ``Petrels: Parallel subspace estimation and tracking by recursive
  least squares from partial observations,''
\newblock {\em IEEE Trans. Sig. Proc.}, December 2013.

\bibitem{local_conv_grouse}
L.~Balzano and S.~Wright,
\newblock ``Local convergence of an algorithm for subspace identification from
  partial data,''
\newblock {\em Found. Comput. Math.}, vol. 15, no. 5, 2015.

\bibitem{chordal_dist}
K.~Ye and L.~H. Lim,
\newblock ``Schubert varieties and distances between subspaces of different
  dimensions,''
\newblock {\em SIAM Journal on Matrix Analysis and Applications}, vol. 37, no.
  3, pp. 1176--1197, 2016.

\bibitem{xu_nips2013_1}
J.~Feng, H.~Xu, and S.~Yan,
\newblock ``Online robust pca via stochastic optimization,''
\newblock in {\em NIPS}, 2013.

\bibitem{candes_rip}
E.~Candes,
\newblock ``The restricted isometry property and its implications for
  compressed sensing,''
\newblock {\em C. R. Math. Acad. Sci. Paris Serie I}, 2008.

\bibitem{zhan_pcp_jp}
J.~Zhan and N.~Vaswani,
\newblock ``Robust pca with partial subspace knowledge,''
\newblock {\em IEEE Trans. Sig. Proc.}, July 2015.

\bibitem{l1_best}
Lin Xiao and Tong Zhang,
\newblock ``A proximal-gradient homotopy method for the l1-regularized
  least-squares problem,''
\newblock in {\em ICML}, 2012.

\bibitem{pca_dd}
N.~Vaswani and P.~Narayanamurthy,
\newblock ``Finite sample guarantees for pca in non-isotropic and
  data-dependent noise,''
\newblock {\em arXiv:1709.06255}, 2017.

\bibitem{musco2015randomized}
Cameron Musco and Christopher Musco,
\newblock ``Randomized block krylov methods for stronger and faster approximate
  singular value decomposition,''
\newblock in {\em Advances in Neural Information Processing Systems}, 2015, pp.
  1396--1404.

\bibitem{corpca_nips}
N.~Vaswani and H.~Guo,
\newblock ``Correlated-pca: Principal components' analysis when data and noise
  are correlated,''
\newblock in {\em Adv. Neural Info. Proc. Sys. (NIPS)}, 2016.

\bibitem{tail_bound}
J.~A. Tropp,
\newblock ``User-friendly tail bounds for sums of random matrices,''
\newblock {\em Found. Comput. Math.}, vol. 12, no. 4, 2012.

\bibitem{vershynin}
R.~Vershynin,
\newblock ``Introduction to the non-asymptotic analysis of random matrices,''
\newblock {\em Compressed sensing}, pp. 210--268, 2012.

\bibitem{hornjohnson}
R.~Horn and C.~Johnson,
\newblock {\em Matrix Analysis},
\newblock Cambridge Univ. Press, 1985.

\bibitem{dimreduce_app1}
E.~Keogh, K.~Chakrabarti, M.~Pazzani, and S.~Mehrotra,
\newblock ``Dimensionality reduction for fast similarity search in large time
  series databases,''
\newblock {\em Knowledge and information Systems}, vol. 3, no. 3, pp. 263--286,
  2001.

\bibitem{dimreduce_app2}
J.~C. Harsanyi and C.~I. Chang,
\newblock ``Hyperspectral image classification and dimensionality reduction: an
  orthogonal subspace projection approach,''
\newblock {\em IEEE Transactions on geoscience and remote sensing}, vol. 32,
  no. 4, pp. 779--785, 1994.

\bibitem{matcomp_candes}
E.~J. Candes and B.~Recht,
\newblock ``Exact matrix completion via convex optimization,''
\newblock {\em Found. of Comput. Math}, , no. 9, pp. 717--772, 2008.

\bibitem{grass_undersampled}
J.~He, L.~Balzano, and A.~Szlam,
\newblock ``Incremental gradient on the grassmannian for online foreground and
  background separation in subsampled video,''
\newblock in {\em IEEE Conf. on Comp. Vis. Pat. Rec. (CVPR)}, 2012.

\bibitem{yall1}
J.~Yang and Y.~Zhang,
\newblock ``Alternating direction algorithms for l1 problems in compressive
  sensing,''
\newblock Tech. {R}ep., Rice University, June 2010.

\end{thebibliography}

\appendices
\counterwithin{theorem}{section}


\section{Proof of Theorem \ref{thm1}}\label{sec:proof}
We divide the proof into $3$ parts for better clarity. We first prove the $\vt=0$ case for NORST (Algorithm \ref{algo:auto-reprocs-pca}), then prove the correctness of Offline NORST, and finally explain the changes needed when $\vt \neq 0$.

\subsection{Proof with $\vt = 0$}


%


\begin{remark}[Deriving the long expression for $K$ given in the discussion]
We have used $\outfracrow^\alpha \le b_0$ with $b_0 = 0.01/f^2$ throughout the analysis in order to simplify the proof. If we were not to do this, and if we used \cite{pca_dd}, it is possible to show that the ``decay rate'' $q_k$ is of the form $q_{k} = (c_2 \sqrt{b_0} f)^k q_0$ from which it follows that to obtain an $\zz$-accurate approximation of the subspace it suffices to have
\begin{align*}
K = \left \lceil \frac{\log \left( c_1\semax / \zz \right)}{-\log(c_2 \sqrt{b_0} f)}  \right \rceil.
\end{align*}
\label{K_long_rem}
\end{remark}

We first prove Theorem \ref{thm1} for the case when $t_j$'s are known, i.e.,  correctness of Algorithm \ref{norst_basic}.
\begin{proof}[Proof of Theorem \ref{thm1} with assuming $t_j$ known]
 In this case $\that_j = t_j$. The proof is an easy consequence of Lemmas \ref{lem:reprocspcalemone} and \ref{lem:reprocspcalemk}. Recall that $\Gamma_{j, K} \subseteq \Gamma_{j, K-1} \subseteq \cdots \Gamma_{j, 0}$ and $\Gamma_{J, K} \subseteq \Gamma_{J-1, K} \subseteq \cdots \subseteq \Gamma_{1, K}$. To show that the conclusions of the Theorem hold, it suffices to show that $\Pr(\Gamma_{J, K} | \Gamma_{0,0}) \geq 1 - 10dn^{-10}$. Using the chain rule of probability,
\begin{align*}
\Pr(\Gamma_{J, K} | \Gamma_{1, 0}) &= \Pr(\Gamma_{J,K}, \Gamma_{J-1, K}, \cdots, \Gamma_{1, K} | \Gamma_{1, 0}) \\
&= \prod_{j = 1}^{J} \Pr(\Gamma_{j, K} | \Gamma_{j, 0}) = \prod_{j=1}^{J} \Pr(\Gamma_{j,K}, \Gamma_{j, K-1}, \cdots, \Gamma_{j, 1} | \Gamma_{j, 0}) \\
&= \prod_{j = 1}^{J} \prod_{k=1}^K \Pr(\Gamma_{j, k} | \Gamma_{j, k-1}) \overset{(a)}{\geq} (1 - 10 n^{-10})^{JK} \geq 1 - 10 J K n^{-10}.
\end{align*}
where $(a)$ used $\Pr(\Gamma_{j, 1} | \Gamma_{j, 0}) \geq 1 - 10n^{-10}$ from Lemma \ref{lem:reprocspcalemone} and $\Pr(\Gamma_{j, k} | \Gamma_{j, k-1}) \geq 1 - 10n^{-10}$ from Lemma \ref{lem:reprocspcalemk}.
\end{proof}

\begin{proof}[Proof of Theorem \ref{thm1}]
Define
\[
\that_{j-1,fin}: =  \that_{j-1} + K \alpha, \  t_{j,*}=\that_{j-1,fin} + \left \lceil \frac{t_j - \that_{j-1,fin}}{\alpha} \right \rceil \alpha
\]
Thus, $\that_{j-1,fin}$ is the time at which the $(j-1)$-th subspace update is complete; w.h.p., this occurs before $t_j$. With this assumption, $t_{j,*}$ is such that $t_j$ lies in the interval $[t_{j,*}-\alpha+1,t_{j,*}]$.

\noindent Recall from the algorithm that we increment $j$ to $j+1$ at $t= \that_j+K\alpha:= \that_{j,fin}$.
Define the events
\ben
\item $\mathrm{Det0}:= \{\that_j = t_{j,*} \} = \{\lambda_{\max}(\frac{1}{\alpha} \sum_{t= t_{j,*}-\alpha+1}^{t_{j,*}} \bphi \lhat_t \lhat_t'\bphi) > \lthres\}$ and
\\ $\mathrm{Det1}:= \{\that_j = t_{j,*} + \alpha\} = \{\lambda_{\max}(\frac{1}{\alpha} \sum_{t= t_{j,*} +1}^{t_{j,*}+\alpha} \bphi \lhat_t \lhat_t'\bphi) > \lthres\} $,
\item $\mathrm{\subup}:=\cap_{k=1}^K  \mathrm{\subup}_k$ where $\mathrm{\subup}_k:= \{\SE(\Phat_{j,k}, \P_{j}) \le q_k\}$,

\item $\mathrm{NoFalseDets}:= \{\text{for all $\J^\alpha \subseteq [\that_{j,fin}, t_{j+1})$, } \lambda_{\max}(\frac{1}{\alpha} \sum_{t \in \J^\alpha} \bphi \lhat_t \lhat_t'\bphi) \le \lthres\}$
\item $\Gamma_{0,\ed}:= \{\SE(\Phat_0, \P_0) \le 0.25 \}$,
\item  $\Gamma_{j,\ed}:= \Gamma_{j-1,\ed} \cap
\big( (\mathrm{Det0} \cap \mathrm{\subup} \cap \mathrm{NoFalseDets}) \cup
(\overline{\mathrm{Det0}} \cap \mathrm{Det1} \cap \mathrm{\subup} \cap \mathrm{NoFalseDets}) \big)$.
\een

Let $p_0$ denote the probability that, conditioned on $\Gamma_{j-1,\ed}$, the change got detected at $t=t_{j,*}$, i.e., let
\[
p_0:= \Pr(\mathrm{Det0}|\Gamma_{j-1,\ed}).
\]
Thus, $\Pr(\overline{\mathrm{Det0}}|\Gamma_{j-1,\ed}) = 1- p_0$. It is not easy to bound $p_0$. However, as we will see, this will not be needed. Assume that $\Gamma_{j-1,\ed} \cap \overline{\mathrm{Det0}}$ holds. Consider the interval $\J^\alpha: = [t_{j,*}, t_{j,*}+\alpha)$. This interval starts at or after $t_j$, so, for all $t$ in this interval, the subspace has changed. For this interval, $\bphi = \I - \Phat_{j-1} \Phat_{j-1}{}'$. 
Applying the first item of Lemma \ref{lem:sschangedet}, w.p. at least $1-10n^{-10}$,
\[
\lambda_{\max} \left(\frac{1}{\alpha} \sum_{t \in \J^\alpha} \bphi \lhat_t \lhat_t'\bphi \right) \geq \lthres
\]
and thus $\that_j = t_{j,*} + \alpha$.
In other words,
\[
\Pr(\mathrm{Det1} | \Gamma_{j-1,\ed} \cap \overline{\mathrm{Det0}}) \ge 1 - 10n^{-10}.
\]

Conditioned on $\Gamma_{j-1,\ed} \cap \overline{\mathrm{Det0}} \cap \mathrm{Det1}$, the first SVD step is done at $t= \that_j + \alpha = t_{j,*} + 2\alpha$ and the subsequent steps are done every $\alpha$ samples. We can prove Lemma \ref{lem:reprocspcalemone} with $\Gamma_{j,0}$ replaced by $\Gamma_{j,\ed} \cap \overline{\mathrm{Det0}} \cap \mathrm{Det1}$ and Lemma \ref{lem:reprocspcalemk} with $\Gamma_{j,k-1}$ replaced by $\Gamma_{j,\ed} \cap \overline{\mathrm{Det0}} \cap \mathrm{Det1} \cap \subup_1 \cap \cdots \cap \subup_{k-1}$ and with the $k$-th SVD interval being $\J_k:=[\that_j+(k-1)\alpha, \that_j + k \alpha)$. Applying Lemmas \ref{lem:reprocspcalemone}, and \ref{lem:reprocspcalemk} for each $k$, we get
\[
\Pr(\subup |\Gamma_{j-1,\ed} \cap \overline{\mathrm{Det0}} \cap \mathrm{Det1}) \ge ( 1 - 10n^{-10})^{K+1}.
\]
We can also do a similar thing for the case when the  change is detected at $t_{j,*}$, i.e. when $\mathrm{Det0}$ holds. In this case, we replace $\Gamma_{j, 0}$ by $\Gamma_{j,\ed} \cap \mathrm{Det0}$ and $\Gamma_{j, k}$ by $\Gamma_{j,\ed} \cap \mathrm{Det0} \cap \subup_1 \cap \cdots \cap \subup_{k-1}$ and conclude that
\[
\Pr(\subup|\Gamma_{j-1,\ed} \cap \mathrm{Det0}) \ge ( 1 - 10n^{-10})^{K}.
\]

Finally consider the $\mathrm{NoFalseDets}$ event. First, assume that $\Gamma_{j-1,\ed} \cap \mathrm{Det0} \cap \subup$ holds.  Consider any interval $\J^\alpha \subseteq [\that_{j,fin}, t_{j+1})$. In this interval, $\Phat_{(t)} = \Phat_j$, $\bphi = \I -  \Phat_j \Phat_j{}'$ and $\SE(\Phat_j,\P_j) \le \zz$. Using the second part of Lemma \ref{lem:sschangedet} we conclude that w.p. at least $1- 10n^{-10}$,
\[
\lambda_{\max} \left(\frac{1}{\alpha} \sum_{t \in \J^\alpha} \bphi \lhat_t \lhat_t'\bphi \right)  < \lthres
\]
Since $\mathrm{Det0}$ holds, $\that_j = t_{j,*}$.
Thus, we have a total of $\lfloor \frac{t_{j+1} - t_{j,*} - K \alpha - \alphadel}{\alpha} \rfloor$ intervals $\J^\alpha$ that are subsets of $[\that_{j,fin}, t_{j+1})$. Moreover, $\lfloor \frac{t_{j+1} - t_{j,*} - K \alpha - \alphadel}{\alpha} \rfloor \le \lfloor \frac{t_{j+1} - t_j - K \alpha - \alphadel}{\alpha} \rfloor \le \lfloor \frac{t_{j+1} - t_j}{\alpha} \rfloor - (K+1)$ since $\alpha \le \alphadel$.
Thus,
\[
\Pr(\mathrm{NoFalseDets} | \Gamma_{j-1,\ed} \cap \mathrm{Det0} \cap \subup) \ge (1 - 10n^{-10})^{\lfloor \frac{t_{j+1} - t_j}{\alpha} \rfloor - (K)}
\]
On the other hand, if we condition on $\Gamma_{j-1,\ed} \cap \overline{\mathrm{Det0}} \cap \mathrm{Det1} \cap \subup$, then $\that_j = t_{j,*} + \alpha$. Thus,
\[
\Pr(\mathrm{NoFalseDets} | \Gamma_{j-1,\ed} \cap \overline{\mathrm{Det0}} \cap \mathrm{Det1} \cap \subup) \ge (1 - 10n^{-10})^{\lfloor \frac{t_{j+1} - t_j}{\alpha} \rfloor - (K+1)}
\]
We can now combine the above facts to bound $\Pr(\Gamma_{j,\ed}|\Gamma_{j-1,\ed})$. Recall that $p_0:= \Pr(\mathrm{Det0}|\Gamma_{j-1,\ed})$.
Clearly, the events $(\mathrm{Det0} \cap \subup \cap \mathrm{NoFalseDets})$ and $(\overline{\mathrm{Det0}} \cap \mathrm{Det1} \cap \subup \cap \mathrm{NoFalseDets})$ are disjoint. Thus,
\begin{align*}
& \Pr(\Gamma_{j,\ed}|\Gamma_{j-1,\ed}) \\
& = p_0 \Pr(\subup \cap \mathrm{NoFalseDets} |\Gamma_{j-1,\ed} \cap \mathrm{Det0})  \\
& + (1-p_0) \Pr(\mathrm{Det1}|\Gamma_{j-1,\ed} \cap \overline{\mathrm{Det0}}) \Pr(\subup \cap \mathrm{NoFalseDets} |\Gamma_{j-1,\ed}\cap \overline{\mathrm{Det0}} \cap \mathrm{Det1}) \\
& \ge p_0 ( 1 - 10n^{-10})^{K} (1 - 10n^{-10})^{\lfloor \frac{t_{j+1} - t_j}{\alpha} \rfloor - (K)} \\
& + (1-p_0) ( 1 - 10n^{-10}) ( 1 - 10n^{-10})^{K}  (1 - 10n^{-10})^{\lfloor \frac{t_{j+1} - t_j }{\alpha} \rfloor - (K+1)}  \\
& =  ( 1 - 10n^{-10})^{\lfloor \frac{t_{j+1} - t_j}{\alpha} \rfloor}
\ge ( 1 - 10n^{-10})^{t_{j+1}-t_j}.
\end{align*}
Since the events $\Gamma_{j,\ed}$ are nested, the above implies that
\begin{align*}
\Pr(\Gamma_{J,\ed}|\Gamma_{0,\ed}) = \prod_j \Pr(\Gamma_{j,\ed}|\Gamma_{j-1,\ed}) &\ge \prod_j ( 1 - 10n^{-10})^{t_{j+1}-t_j} = ( 1 - 10n^{-10})^d \\
&\ge  1 - 10d n^{-10}.
\end{align*}
\end{proof}

\subsection{Proof of Offline NORST}

We now provide the proof of the Offline Algorithm (lines 26-30 of Algorithm \ref{algo:auto-reprocs-pca}).
\begin{proof}[Proof of Offline NORST]
The proof of this follows from the conclusions of the online counterpart. Note that the subspace estimate in this case is not necessarily $r$ dimensional. This is essentially done to ensure that in the time intervals when the subspace has changed, but has not yet been updated, the output of the algorithm is still an $\zz$-approximate solution to the true subspace. In other words, for $t \in [\that_{j-1} + K\alpha, t_{j}]$, the true subspace is $\P_{j-1}$ and so in this interval
\begin{align*}
\SE(\Phat_{(t)}^{\offline}, \P_{j-1}) &= \SE([\Phat_{j-1}, (\I - \Phat_{j-1}\Phat_{j-1}{}') \Phat_j], \P_{j-1}) \\
&\overset{(a)}{=} \norm{[\I - (\I - \Phat_{j-1}\Phat_{j-1}{}')\Phat_j\Phat_j{}'(\I - \Phat_{j-1}\Phat_{j-1}{}')][\I - \Phat_{j-1}\Phat_{j-1}{}']\P_{j-1}} \\
&\leq \norm{[\I - (\I - \Phat_{j-1}\Phat_{j-1}{}')\Phat_j\Phat_j{}'(\I - \Phat_{j-1}\Phat_{j-1}{}')]}\SE(\Phat_{j-1}, \P_{j-1}) \leq \zz
\end{align*}
where $(a)$ follows because for orthogonal matrices $\P_1$ and $\P_2$,
\begin{align*}
\I - \P_1\P_1{}' - \P_2\P_2{}' =  (\I - \P_1\P_1{}')(\I - \P_2\P_2{}') = (\I - \P_2\P_2{}')(\I - \P_1\P_1{}')
\end{align*}
Now consider the interval $t \in [t_j, \that_{j} + K\alpha]$. In this interval, the true subspace is $\P_j$ and we have back propagated the $\zz$-approximate subspace $\Phat_j$ in this interval. We first note that $\Span([\Phat_{j-1}, (\I - \Phat_{j-1}\Phat_{j-1}{}')\Phat_j]) = \Span([\Phat_{j}, (\I - \Phat_{j}\Phat_{j}{}')\Phat_{j-1}])$. And so we use the latter to quantify the error in this interval as
\begin{align*}
\SE(\Phat_{(t)}^{\offline}, \P_{j}) &= \SE([\Phat_j, (\I - \Phat_j\Phat_j{}') \Phat_{j-1}], \P_{j}) \\
&= \norm{[\I - (\I - \Phat_{j}\Phat_{j}{}')\Phat_{j-1}\Phat_{j-1}{}'(\I - \Phat_{j}\Phat_{j}{}')][\I - \Phat_{j}\Phat_{j}{}']\P_{j}} \\
&\leq \norm{[\I - (\I - \Phat_{j}\Phat_{j}{}')\Phat_{j-1}\Phat_{j-1}{}'(\I - \Phat_{j}\Phat_{j}{}')]}\SE(\Phat_{j}, \P_{j}) \leq \zz
\end{align*}
\end{proof}

\subsection{Proof with $\vt \neq 0$}
In this section we analyze the ``stable'' version of RST, i.e., we let $\vt \neq 0$.
\begin{proof}
The proof is very similar to that of the noiseless case but there are two differences due to the additional noise term. The first is the effect of the noise on the sparse recovery step. The approach to address this is straightforward. We note that the error now seen in the sparse recovery step is bounded by $\|\bpsi(\lt + \vt)\|$ and using the bound on $\|\vt\|$, we observe that the error only changes by a constant factor. In particular, we can show that $\|\et\| \leq 2.4 (2\zz + \semax) \sqrt{\eta r \lambda^+}$. The other crucial difference is in updating subspace estimate. To deal with the additional uncorrelated noise, we use the following result.

Remark 4.18 of \cite{pca_dd} states the following for the case where the data contains unstructured noise $\vt$ that satisfies the assumptions of Theorem \ref{thm1}. Thus, in the notation of \cite{pca_dd}, $\lambda_v^+ \le c \zz^2 \lambda^+$ and $r_v = r$. The following result also assumes $r,n$ large enough so that $(r+\log n) \le r \log n$.
\begin{corollary}[Noisy PCA-SDDN]\label{cor:noisy_pca_sddn}
Given data vectors $\yt := \lt + \wt + \zt = \lt + \I_{\T_t} \M_{s,t} \lt + \zt$, $t=1,2,\dots,\alpha$, where $\T_t$ is the support set of $\wt$, and $\lt$ satisfying the model detailed above. Furthermore, $\max_t \|\M_{s,t} \P\|_2 \le q < 1$. $\zt$ is small uncorrelated noise such that $\ep[\zt \zt{}'] = \bm{\Sigma}_z$, $ \max_t \|\zt\|^2 := b_z^2 < \infty$. Define $\lambda_z^+ := \lambda_{\max}(\bm{\Sigma}_z)$ and $r_z$ as the ``effective rank'' of $\zt$ such that $b_z^2 = r_z \lambda^+_z$. Then for any $\alpha \geq \alpha_0$, where
\begin{align*}
\alpha_0 := \frac{C}{\varepsilon_{\text{SE}}^2}  \max\left\{\eta q^2 f^2 r \log n,\  \frac{b_z^2}{\lambda^-} f \log n \right\}
\end{align*}
the fraction of nonzeroes in any row of the noise matrix $[\w_1, \w_2, \dots, \w_\alpha]$ is bounded by $\bz$, and
\begin{align*}
3 \sqrt{\bz} q f + \lambda_z^+ / \lambda^- \le \frac{0.9  \varepsilon_{\text{SE}}}{1+\varepsilon_{\text{SE}}}
\end{align*}
For an $\alpha \ge \alpha_0$, let $\Phat$ be the matrix of top $r$ eigenvectors of $\D:=\frac{1}{\alpha} \sum_t \yt \yt'$. With probability at least $1- 10n^{-10}$, $\SE(\Phat,\P) \le \varepsilon_{\text{SE}}$.
\end{corollary}


We illustrate how applying Corollary \ref{cor:noisy_pca_sddn} changes the subspace update step. Consider the first subspace estimate, i.e., we are trying to get an estimate $\Phat_{j,1}$ in the $j$-th subspace change time interval. Define $(\e_{\l})_t = \I_{\Tt}{}' \left( \bpsi_{\Tt}{}' \bpsi_{\Tt}\right)^{-1} \bpsi_{\Tt}{}' \lt$ and $(\e_{\v})_t = \I_{\Tt}{}' \left( \bpsi_{\Tt}{}' \bpsi_{\Tt}\right)^{-1} \bpsi_{\Tt}{}' \vt$. We estimate the new subspace, $\Phat_{j, 1}$ as the top $r$ eigenvectors of $\frac{1}{\alpha}\sum_{t = \that_j}^{\that_j + \alpha - 1} \lhat_t \lhat_t{}'$. In the setting above, $\yt \equiv \lhat_t$, $\wt \equiv (\e_{\l})_t$, $\zt \equiv (\e_{\v})_t$, $\lt \equiv \lt$ and $\M_{s, t} = -\left( \bpsi_{\Tt}{}' \bpsi_{\Tt}\right)^{-1} \bpsi_{\Tt}{}'$ and so $\norm{\M_{s, t} \P} = \| \left( \bpsi_{\Tt}{}' \bpsi_{\Tt}\right)^{-1} \bpsi_{\Tt}{}' \P_j\| \leq \phi^+ (\zz + \SE(\P_{j-1}, \P_j)) := q_0$. Applying Corollary \ref{cor:noisy_pca_sddn} with $q \equiv q_0$, and recalling that the support, $\Tt$ satisfies the assumptions similar to that of the noiseless case and hence $b_0 \equiv \outfracrow^\alpha$. Now, setting $\varepsilon_{\text{SE}, 1} = q_0/4$, observe that we require
\begin{align*}
(i)\ \sqrt{b_0} q_0 f \leq \frac{0.5 \cdot 0.9 \varepsilon_{\text{SE}, 1}}{1 + \varepsilon_{\text{SE}, 1}}, \quad \text{and}, \quad(ii)\  \frac{\lambda_z^+}{\lambda^-} \leq \frac{0.5 \cdot 0.9 \varepsilon_{\text{SE}, 1}}{1 + \varepsilon_{\text{SE}, 1}}.
\end{align*}
which holds if (i) $\sqrt{b_0}f \leq 0.12$, and (ii) is satisfied as follows from using the assumptions on $\vt$ as follows. It is immediate to see that $\lambda_z^+ / \lambda^- \leq \zz^2  \leq .2 \varepsilon_{\text{SE}, 1}$. Furthermore, the sample complexity term remains unchanged due to the choice of $\vt$. To see this, notice that the only extra term in the $\alpha_0$ expression is $b_z^2 f \log n/ (\varepsilon_{\text{SE}}^2 \lambda^-)$ which simplifies to $\zz^2 f^2 r \log n / \varepsilon_{\text{SE}}^2$ which is what was required even in the noiseless case. Thus, from Corollary \ref{cor:noisy_pca_sddn}, with probability at least $1 - 10n^{-10}$, $\SE(\Phat_{j, 1}, \P_j) \leq \varepsilon_{\text{SE}, 1} = q_0/4$. The argument in other subspace update stages will require the same changes and follows without any further differences.

The final difference is in the subspace detection step. Notice that here too, in general, there will be some extra assumption required to provably detect the subspace change. However, due to the bounds assumed on $\|\vt\|$ and the bounds on using $\epsilon_{l,v} = \epsilon_{v,v} = 0.01 \zz$, we see that (i) the extra sample complexity term is the same as that required in the noiseless case. 
\end{proof}

\section{Proof of Lemmas \ref{lem:sumprinang} and \ref{lem:concm}}\label{app:concm}

\begin{proof}[Proof of Lemma \ref{lem:sumprinang}]
The proof follows from triangle inequality as
\begin{align*}
\SE(\Aa, \Ca) &= \norm{(\I - \Aa \Aa{}') \Ca} = \norm{(\I - \Aa \Aa{}')(\I - \Ba \Ba{}' + \Ba \Ba{}') \Ca} \\
&\leq \norm{(\I - \Aa \Aa{}')(\I - \Ba \Ba{}')\Ca} + \norm{(\I - \Aa \Aa{}') \Ba \Ba{}' \Ca} \\
&\leq \norm{(\I - \Aa \Aa{}')} \SE(\Ba, \Ca) + \SE(\Aa, \Ba) \norm{\Ba{}' \Ca} \leq \Delta_1 + \Delta_2
\end{align*}
\end{proof}

We need the following results for proving Lemma \ref{lem:concm}.
\renewcommand{\zt}{\bm{Z}_t}
\begin{theorem}[Cauchy-Schwartz for sums of matrices \cite{rrpcp_perf}] \label{CSmat}
For matrices $\bm{X}$ and $\bm{Y}$ we have
\begin{eqnarray}\label{eq:csmat}
\norm{\frac{1}{\alpha} \sum_t \bm{X}_t \bm{Y}_{t}{}'}^2 \leq \norm{\frac{1}{\alpha} \sum_t \bm{X}_t \bm{X}_t{}'} \norm{\frac{1}{\alpha} \sum_t \bm{Y}_t \bm{Y}_t{}'}
\end{eqnarray}
\end{theorem}
The following theorem is adapted from \cite{tail_bound}.
\begin{theorem}[Matrix Bernstein \cite{tail_bound}]\label{thm:matrix_bern}
Given an $\alpha$-length sequence of $n_1 \times n_2$ dimensional random matrices and a r.v. $X$. Assume the following holds. For all $X \in \mathcal{C}$, (i) conditioned on $X$, the matrices $\Z_t$ are mutually independent, (ii) $\mathbb{P}(\norm{\Z_t} \leq R | X)  = 1$,  and (iii) $\max\left\{\norm{\frac{1}{\alpha}\sum_t \ep{\left[\Z_t{}'\Z_t | X\right]}},\ \norm{\frac{1}{\alpha}\sum_t \ep{\left[\Z_t\Z_t{}' | X\right]}}\right\} \le \sigma^2$. Then, for an $\epsilon > 0$ and for all $X \in \mathcal{C}$,
\begin{align}
\mathbb{P}\left(\norm{\frac{1}{\alpha} \sum_t \Z_t} \leq \norm{\frac{1}{\alpha} \sum_t \ep{\left[\Z_t|X\right]}} + \epsilon\bigg|X\right) \geq 1 - (n_1 + n_2) \exp\left(\frac{-\alpha\epsilon^2}{2\left(\sigma^2 + R \epsilon\right)} \right).
\end{align}
\end{theorem}

The following theorem is adapted from \cite{vershynin}.
\begin{theorem}[Sub-Gaussian Rows \cite{vershynin}]\label{thm:versh}
Given an $N$-length sequence of sub-Gaussian random vectors $\bm{w}_i$ in $\mathbb{R}^{n_w}$, an r.v $X$, and a set $\mathcal{C}$. Assume the following holds. For all $X \in \mathcal{C}$, (i) $\bm{w}_i$ are conditionally independent given $X$; (ii) the sub-Gaussian norm of $\bm{w}_i$ is bounded by $K$ for all $i$. Let $\bm{W}:=[\bm{w}_1, \bm{w}_2, \dots, \bm{w}_N]{}'$.
Then for an $\epsilon \in (0, 1)$ and for all $X \in \mathcal{C}$
\begin{align}
\mathbb{P}\left(\norm{\frac{1}{N}\bm{W}{}'\bm{W} - \frac{1}{N}\ep{\left[\bm{W}{}'\bm{W} | X \right]}} \leq \epsilon \bigg| X\right) \geq 1 - 2\exp\left({n_w} \log 9 - \frac{c \epsilon^2 N}{4K^4}\right).
\end{align}
\end{theorem}


\begin{proof}[Proof of Lemma \ref{lem:concm}]
The proof  approach is similar to that of \cite[Lemma 7.17]{rrpcp_dynrpca} but the details are different since we use a simpler subspace model.

\emph{Item 1}: Recall that the $(\at)_i$ are bounded r.v.'s satisfying $|(\at)_i| \leq \sqrt{\eta \lambda_i}$. Thus, the vectors, $\at$ are sub-Gaussian with $\norm{\at}_{\psi_2} = \max_i \norm{(\at)_i}_{\psi_2} = \sqrt{\eta \lambda^+}$. We now apply Theorem \ref{thm:versh} with $K \equiv \sqrt{\eta \lambda^+}$, $\epsilon = \epsilon_0 \lambda^-$, $N \equiv \alpha$ and $n_w \equiv r$ to conclude the following: For an $\alpha \geq \alpha_{(0)} := {C (r \log 9 + 10\log n)f^2}$, 
\begin{align*}
\Pr\left( \norm{\frac{1}{\alpha} \sum_t \at \at{}' - \bm{\Lambda}} \leq \epsilon_0 \lambda^-  \right) \geq 1 - 10n^{-10}
\end{align*}
The Lemma statement assumes $\alpha = C f^2 r \log n$. For large $r,n$, this $\alpha > \alpha_{(0)} = C f^2 (r + \log n)$. Thus, the above holds under the Lemma statement.


\emph{Item 2}: For the second term, we proceed as follows. Since $\norm{\bphi} = 1$,
\begin{align*}
\norm{\frac{1}{\alpha} \sum_t \bphi \lt \et{}'\bphi} \leq \norm{\frac{1}{\alpha} \sum_t \bphi \lt \et{}'}.
\end{align*}
To bound the RHS above, we will apply Theorem \ref{thm:matrix_bern} with $\zt =  \bphi \lt \et{}'$. Conditioned on $\{\Phat_*,Z\}$, the $\zt$'s are mutually independent. We first bound obtain a bound on the expected value of the time average of the $\zt$'s and then compute $R$ and $\sigma^2$.
By Cauchy-Schwartz,
\begin{align}
\norm{\ep\left[\frac{1}{\alpha} \sum_t \bphi \lt \et{}'  \right]}^2 &= \norm{\frac{1}{\alpha} \sum_t \bphi \pt \bm{\Lambda} \pt{}' \mot{}' \mtt{}'}^2 \nonumber \\
&\overset{(a)}{\leq} \norm{\frac{1}{\alpha} \sum_t \left(\bphi \pt \bm{\Lambda} \pt{}' \mot{}'\right)\left( \mot \pt\bm{\Lambda}\pt{}'\bphi\right)} \norm{\frac{1}{\alpha} \sum_t \mtt \mtt{}'} \nonumber \\
&\overset{(b)}{\leq}  b_0 \left[ \max_t \norm{\bphi \pt \bm{\Lambda} \pt{}' \mot{}'}^2 \right] \nonumber \\
&\leq b_0 \SE^2(\Phat, \P) q^2 (\lambda^+)^2
\label{lt_wt_bnd}
\end{align}
where (a) follows by Cauchy-Schwartz (Theorem \ref{CSmat}) with $\bm{X}_t = \bphi \pt \bm{\Lambda} \pt{}' \mot{}'$ and $\bm{Y}_t =\mtt$, (b) follows from the assumption on $\M_{2,t}$. To compute $R$
\begin{align*}
\norm{\zt} \leq \norm{\bphi \lt} \norm{\et} &\leq \SE(\Phat, \P) q \eta \rfix \lfp := R
\end{align*}
Next we compute $\sigma^2$. Since $\wt$'s are bounded r.v.'s, we have
\begin{align*}
\norm{\frac{1}{\alpha}\sum_t\ep{ [\zt \zt{}' ]}} &= \norm{\frac{1}{\alpha} \sum_t \ep{\left[\bphi \lt \et{}' \et \lt{}' \bphi  \right]}}
 = \norm{ \frac{1}{\alpha}\ep{[ \norm{\et}^2  \bphi \lt \lt{}' \bphi ]}} \\
&\leq \left( \max_{\et} \norm{\et}^2 \right)\norm{\frac{1}{\alpha}\sum_t\ep{\left[\bphi \lt \lt{}' \bphi \right]}}\\
&\leq q^2 \SE^2(\Phat, \P) \eta r (\lambda^+)^2 \cdot  := \sigma^2
\end{align*}
it can also be seen that $\norm{\frac{1}{\alpha}\sum_t \ep{\left[\zt{}' \zt \right]}}$ evaluates to the same expression. Thus, applying Theorem \ref{thm:matrix_bern}
\begin{align*}
\Pr \left( \norm{\frac{1}{\alpha}\sum_t \bphi \lt \et{}'} \leq \SE(\Phat, \P) \sqrt{b_0} q \lambda^+ + \epsilon  \right) \\
\geq 1- 2n \exp \left( \frac{-\alpha}{4\max\left\{\frac{\sigma^2}{\epsilon^2},\ \frac{R}{\epsilon}\right\} }\right).
\end{align*}
Let $\epsilon = \epsilon_1 \lambda^-$, then, $\sigma^2/\epsilon^2 =c \eta f^2 r$ and $R/\epsilon = c \eta f  r$. Hence, for the probability to be of the form $1 - 2 n^{-10}$ we require that $\alpha \geq \alpha_{(1)} := C \cdot  \eta  f^2  (r \log n)$.

\emph{Item 3}: We use Theorem \ref{thm:matrix_bern} with $\zt:= \bphi \et \et{}' \bphi$. The proof is analogous to the previous item. First we bound the norm of the expectation of the time average of $\zt$:
\begin{align*}
\norm{\ep\left[ \frac{1}{\alpha} \sum \bphi \et \et{}' \bphi  \right]} &= \norm{\frac{1}{\alpha} \sum  \bphi \mtt \mot \pt \bm{\Lambda} \pt{}' \mot{}' \mtt{}' \bphi} \\
&\leq \norm{\frac{1}{\alpha} \sum  \mtt \mot \pt \bm{\Lambda} \pt{}' \mot{}' \mtt{}'} \\
&\overset{(a)}{\leq} \left(\norm{\frac{1}{\alpha} \sum_t \mtt \mtt{}'} \left[ \max_t \norm{\mtt \mot \pt \bm{\Lambda} \pt{}' \mot{} (\cdot){}'}^2\right]\right)^{1/2} \\
&\overset{(b)}{\leq} \sqrt{b_0} \left[\max_t \norm{\mot \pt \bm{\Lambda} \pt{}' \mot{}'\mtt{}'}\right] \leq \sqrt{b_0} q^2 \lambda^+
\end{align*}
where (a) follows from Theorem \ref{CSmat} with $\bm{X}_t = \M_{2,t}$ and $\bm{Y}_t = \mot \pt \bm{\Lambda} \pt{}' \mot{}'\mtt{}'$ and (b) follows from the assumption on $\M_{2,t}$. To obtain $R$,
\begin{align*}
\norm{\zt} = \norm{\bphi \et \et{}' \bphi} &\leq \max_t \norm{\bphi \Mt \P \at}^2 \leq q^2 r \eta \lfp := R
\end{align*}
To obtain $\sigma^2$,
\begin{align*}
\norm{\frac{1}{\alpha}\sum_t\ep{\left[\bphi\et (\bphi\et){}'(\bphi \et) \et{}' \bphi \right]}} &= \norm{\frac{1}{\alpha} \sum_t \ep{\left[\bphi\et \et{}' \bphi \norm{\bphi\et}^2 \right]}} \\
&\leq \left(\max_{\et} \norm{\bphi \et}^2 \right) \norm{\bphi \Mt \pt \bm{\Lambda} \pt{}' \Mt{}' \bphi}  \\
&\leq q^2 r \eta \lfp \cdot q^2 \lfp := \sigma^2
\end{align*}
Applying Theorem \ref{thm:matrix_bern}, we have
\begin{align*}
\Pr\left( \norm{\frac{1}{\alpha} \sum_t \bphi \et \et{}' \bphi} \leq \sqrt{b_0} q^2 \lambda^+ + \epsilon  \right) \geq 1- n \exp \left( \frac{- \alpha \epsilon^2}{2(\sigma^2 + R\epsilon)} \right)
\end{align*}
Letting $\epsilon = \epsilon_2 \lambda^-$ we get $R / \epsilon = c \eta rf$ and $\sigma^2 / \epsilon^2 = c \eta rf^2$. For the success probability to be of the form $1 - 2 n^{-10}$ we require $\alpha \geq \alpha_{(2)} := C \eta \cdot 11 f^2 (r \log n)$.
\end{proof}

The proof of the last two items follow from using \cite[Lemma 7.19]{pca_dd}.

\renewcommand{\zt}{\bm{z}_t}

\section{Proof of Extensions}\label{sec:proof_reprocs-pca-mc}
In this section we present the proof of the extensions stated in Sec. \ref{sec:extensions}.

\subsection{Static Robust PCA}
The proof follows directly from Theorem \ref{thm1} by setting $J=1$. 

\subsection{Subspace Tracking with missing data and dynamic Matrix Completion}
Here we present the proof of the subspace tracking with missing data problem.
The only changes needed for this proof are in the initialization step, i.e., for $j=0$. 
For this we use the following lemma.
\begin{lem}[Lemma 2.1, \cite{matcomp_candes}]\label{lem:candes_rom}
Set $\bar{r} = \max(r, \log n)$. Then there exist constants $C$ and $c$ such that the random orthogonal model with left singular vectors $\Phat_\init$ obeys
$
\Pr\left( \max_i \norm{\I_i{}' \Phat_\init}^2 \leq C \bar{r} / n \right) \geq 1 - c n^{-\beta} \log n.
$
\end{lem}
Thus,
\begin{align*}
\Pr\left( \max_i \norm{\I_i{}' \Phat_\init}^2 \leq \mu \bar{r}/n \right) \geq 1 - n^{-10}
\end{align*}
Consider the two scenarios (i) if $r \geq \log n$, then everything discussed before remains true, whereas, if (ii) $r \leq \log n$, we redefine $\mu^2  = C \log n / r$ and thus in the interval $[t_0, t_1]$ we require $\outfraccol \leq \frac{0.01}{\log n}$. Further, using the bound on $\outfraccol$ it follows from triangle inequality that
\begin{align*}
\max_{\T \leq 2s} \norm{\I_{\T}{}' \Phat_\init}^2 \leq 2s \max_i \norm{\I_i{}' \Phat_\init}^2 \leq \frac{2s \mu r}{n} < 0.01
\end{align*}

We only mention the changes needed for Lemma \ref{lem:reprocspcalemone} for when $j=0$ since the initialization is different. The rest of the argument of recursively applying the lemmas hold exactly as before. First, $t_\train=1$ since we use random initialization. Thus from the Algorithm, $\that_0 = t_\train=1$.

\begin{proof}
{\em Proof of item 1. }
Since the support of $\xt$ is known, the LS step gives
\begin{align*}
\shatt = \bm{I}_{\Tt}\left(\bpsi_{\Tt}{}'\bpsi_{\Tt}\right)^{-1}\bpsi_{\Tt}{}'(\bpsi \lt + \bpsi \st) = \bm{I}_{\Tt}\left(\bpsi_{\Tt}{}'\bpsi_{\Tt}\right)^{-1} \I_{\Tt}{}' \bpsi\lt + \st
\end{align*}
Thus $\et = \shatt-\st$ satisfies
\begin{align*}
\et &= \bm{I}_{\Tt}\left(\bpsi_{\Tt}{}'\bpsi_{\Tt}\right)^{-1} \I_{\Tt}{}' \bpsi\lt  
\end{align*}
Now, from the incoherence assumption on $\Phat_\init$, Lemma \ref{kappadelta}, the bound on $\outfraccol$, and recalling that in this interval, $\bpsi = \I - \Phat_\init \Phat_\init{}'$ we have
\begin{align*}
\max_{|\T| \leq 2s} \norm{\I_{\T}{}' \Phat_\init}^2 \leq \frac{2s \mu r}{n} \leq 0.09 \implies \delta_{2s}(\bpsi) \leq 0.3^2 < 0.15, \\
\norm{\left(\bpsi_{\Tt}{}'\bpsi_{\Tt}\right)^{-1}} \leq \frac{1}{1 - \delta_s(\bpsi)} \leq \frac{1}{1 - \delta_{2s}(\bpsi)} \leq \frac{1}{1- 0.15} < 1.2= \phi^+.
\end{align*}
Secondly,
\begin{align}
\norm{\itt {}' \bpsi \P_0} & \leq \left(\norm{\itt{}' \P_0} + \norm{\itt{}' \Phat_\init}\right) \leq 0.3 + 0.3 = 0.6
\label{use_Mstbnd}
\end{align}
Thus, $\norm{\itt {}' \bpsi \lt} \le 0.6 \sqrt{\eta r \lambda^+}$.

\emph{Proof of Item 2}:
Since $\lhat_t = \lt - \et$ with $\et$ satisfying the above equation, updating $\Phat_{(t)}$ from the $\lhatt$'s is a problem of PCA in sparse data-dependent noise (SDDN), $\et$. To analyze this, we use the PCA-SDDN result of Theorem \ref{cor:pcasddn} (this is taken from \cite{pca_dd}). Recall from above that, for $t \in [\that_0, \that_0 + \alpha]$, $\That_t = \T_t$, and $\lhat_t = \lt - \et$. Recall from the algorithm that we compute the first estimate of the $j$-th subspace, $\Phat_{j, 1}$, as the top $r$ eigenvectors of $\frac{1}{\alpha}\sum_{t = t_0}^{t_0 + \alpha - 1} \lhat_t \lhat_t{}'$. In the notation of  Theorem \ref{cor:pcasddn}, $\yt \equiv \lhat_t$, $\wt \equiv \et$, $\lt \equiv \lt$ and $\M_{s, t} = -\left( \bpsi_{\T_t}{}' \bpsi_{\T_t}\right)^{-1} \bpsi_{\T_t}{}'$ and so $\norm{\M_{s, t} \P} = \| \left( \bpsi_{\T_t}{}' \bpsi_{\T_t}\right)^{-1} \bpsi_{\T_t}{}' \P_0\| \leq \phi^+ \cdot 0.6 = 0.72 := q_0$. This follows using \eqref{use_Mstbnd}.
Also, $b_0 \equiv \outfracrow^\alpha$.

Applying Theorem \ref{cor:pcasddn} with $q \equiv q_0$, $b_0 \equiv \outfracrow^\alpha$ and setting $\varepsilon_{\text{SE}} = q_0/4$, observe that we require
\begin{align*}
\sqrt{b_0} q_0 f \leq \frac{0.9 (q_0/4)}{1 + (q_0/4)}.
\end{align*}
This holds if $\sqrt{b_0}f \leq 0.12$ as provided by Theorem \ref{thm1}. Thus, from Corollary \ref{cor:pcasddn}, with probability at least $1 - 10n^{-10}$, $\SE(\Phat_{j, 1}, \P_0) \leq q_0/4$. 
\end{proof}

\subsection{Fewer than $r$ directions change}
\begin{proof}[Proof of Corollary \ref{cor:rch}]
The proof is very similar to that of Theorem \ref{thm1}. The only changes occur in the \\
\emph{(a) Projected CS step}. With the subspace change model, we define
$
\lt = \P_{(t)} \at :=
\begin{bmatrix} \P_{j-1, \fx} & \P_{j, \chd} \end{bmatrix}
\begin{bmatrix}
\atf \\
\atr
\end{bmatrix}
$
where $\atf$ is a $(r - r_\ch) \times 1$ dimensional vector and $\atr$ is a $r_\ch \times 1$ dimensional vector. 
In the first $\alpha$ frames after the $j$-th subspace changes (or, the $j$-th subspace change is detected in the automatic case), recall that $\Phat_{(t)} = \Phat_{j-1}$. Thus, $\SE(\Phat_{(t)}, \P_{j-1, \fx}) =
\SE(\Phat_{j-1}, \P_{j-1, \fx}) \leq \SE(\Phat_{j-1}, \P_{j-1}) \leq \zz$ and so, the error can be bounded as
\begin{align*}
\norm{\bpsi \lt} \leq \norm{\bpsi \P_{j-1, \fx} \atf} + \norm{\bpsi \P_{j, \chd} \atr} \leq \zz \sqrt{\eta (r - r_\ch) \lambda^+} + (\zz + \SE(\P_{j-1}, \P_j)) \sqrt{\eta r_\ch \lambda_\ch^+}
\end{align*}
In the second $\alpha$ frames, have $\Phat_{(t)} = \basis(\Phat_{j-1}, \Phat_{j,1})$. Thus $\SE(\Phat_{(t)}, \P_{j-1,\fx}) \le
\SE(\Phat_{j-1}, \P_{j-1, \fx}) \leq \SE(\Phat_{j-1}, \P_{j-1}) \leq \zz$ and $\SE(\Phat_{(t)},\P_{j, \chd}) \le  \SE(\Phat_{j,1}, \P_{j, \chd}) \leq \SE(\Phat_{j,1}, \P_{j}) \leq 0.3 \cdot (\zz + \SE(\P_{j-1}, \P_j))$. Thus, the error in the sparse recovery step in the interval after the  first subspace update is  performed is given as
\begin{align*}
\norm{\bpsi \lt} \leq \zz \sqrt{\eta (r - r_\ch) \lambda^+} + 0.3 \cdot (\zz + \SE(\P_{j-1}, \P_j)) \sqrt{\eta r_\ch \lambda_\ch^+}
\end{align*}
The rest of the proof follows as before.
The error after the $k$-th subspace update is also bounded using the above idea.


\emph{(b) Subspace Detection step}:
The proof of the subspace detection step follows exactly analogous to Lemma \ref{lem:sschangedet}. One minor observation is noting that $\SE(\P_{j-1}, \P_J) = \SE(\P_{j-1, \ch}, \P_{j, \chd})$ in the proof of part (a) of Lemma \ref{lem:sschangedet}. The rest of the argument is exactly the same.
\end{proof}

\end{document}